\documentclass[review,hidelinks,onefignum,onetabnum]{siamart220329}

\usepackage{lipsum}
\usepackage{multirow}
\usepackage{amsfonts}
\usepackage{epstopdf}
\usepackage{algorithmic}
\ifpdf
  \DeclareGraphicsExtensions{.eps,.pdf,.png,.jpg}
\else
  \DeclareGraphicsExtensions{.eps}
\fi

\newsiamremark{remark}{Remark}
\newsiamremark{hypothesis}{Hypothesis}
\crefname{hypothesis}{Hypothesis}{Hypotheses}
\newsiamthm{claim}{Claim}

\headers{Path-following in Bayesian inference}{Zilai Si, Yucong Liu, and Alexander Strang}

\title{Path-following methods for Maximum a Posteriori estimators in Bayesian hierarchical models \\ How estimates depend on hyperparameters\thanks{Submitted to the editors \today.
\funding{This research received no specific grant from any funding agency in the public, commercial, or not-for-profit sectors.}}}

\author{Zilai Si\thanks{Department of Statistics, University of Chicago, Chicago, IL (\email{zilai@uchicago.edu, yucongliu@uchicago.edu, alexstrang@uchicago.edu}).}
\and Yucong Liu\footnotemark[2]
\and Alexander Strang\footnotemark[2]}

\usepackage{amsopn}

\ifpdf
\hypersetup{
  pdftitle={Path-following methods for Maximum a Posteriori estimators in Bayesian hierarchical models  How estimates depend on hyperparameters},
  pdfauthor={Zilai Si, Yucong Liu and Alexander Strang}
}
\fi

\begin{document}

\maketitle

\begin{abstract}
Maximum a posteriori (MAP) estimation, like all Bayesian methods, depends on prior assumptions. These assumptions are often chosen to promote specific features in the recovered estimate. The form of the chosen prior determines the shape of the posterior distribution, thus the behavior of the estimator and complexity of the associated optimization problem. Here, we consider a family of Gaussian hierarchical models with generalized gamma hyperpriors designed to promote sparsity in linear inverse problems. By varying the hyperparameters, we move continuously between priors that act as smoothed $\ell_p$ penalties with flexible $p$, smoothing, and scale. We then introduce a predictor-corrector method that tracks MAP solution paths as the hyperparameters vary. Path following allows a user to explore the space of possible MAP solutions and to test the sensitivity of solutions to changes in the prior assumptions. By tracing paths from a convex region to a non-convex region, the user can find local minimizers in strongly sparsity promoting regimes that are consistent with a convex relaxation derived using related prior assumptions. We show experimentally that these solutions are less error prone than direct optimization of the non-convex problem. 
\end{abstract}



\section{Background}

In an inverse problem, the inputs to a forward map $F(\cdot)$ must be inferred from a noisy output $b = F(x) + \epsilon$. Inverse problems arise across domains, including geophysics \cite{snieder1999inverse}, medical imaging \cite{arridge1999optical} signal processing \cite{figueiredo2007gradient} and machine learning \cite{de2005learning}.  Inverse problems are typically ill-posed so often cannot be solved stably without prior knowledge regarding the desired solution. For example, many real signals can be considered sparse or compressive in some domain \cite{candes2006stable,candes2005decoding,donoho2006most,donoho2005stable}. That is, the signal $x$ has many entries that are equal to or close to zero. 

Classically, sparsity is promoted by optimizing a regularized cost function:
\begin{equation}
\label{penalty}
F_p(x, \lambda)=\|b- F(x)\|^2+\lambda\|x\|_p^p
\end{equation}
where $\lambda$ is the regularization parameter that determines the severity of penalty term and $p$ is the shape parameter that determines solution properties. Algorithms that optimize (\ref{penalty}) when $1\leq p\leq2$ have been extensively studied  (c.f.~\cite{daubechies2010iteratively,tibshirani1996regression}). For $p<1$, obtaining the global minimum is challenging since the penalized function is non-convex. From a Bayesian perspective, solving  (\ref{penalty}) is equivalent to finding a maximum a posteriori (MAP) estimate with sparsity promoting priors \cite{park2008bayesian,polson2010shrink,calvetti2019hierachical,calvetti2020sparse}.

Consider the standard linear inverse problem: find $x \in \mathbb{R}^n$ given the data:
\begin{equation}
    b = Ax + \epsilon
\end{equation}
where $A \in \mathbb{R}^{m \times n}$ is a known matrix, and $\epsilon$ represents noise introduced by measurement error.  

Suppose that $\epsilon \sim \mathcal{N}(0,\Sigma)$. Then, the likelihood of sampling $b$ given $x$ is:
\begin{equation} \label{eqn: likelihood}
    \pi(b|x) = \frac{1}{\sqrt{2 \pi |\text{det}(\Sigma)|}} \exp \left(-\frac{1}{2}(b - A x)^\top \Sigma (Ax - b) \right).
\end{equation}

Maximum likelihood estimation (ML) proceeds by finding $x$ which maximize \eqref{eqn: likelihood} given $A$ and $b$. We assume, without loss of generality, that the noise covariance equals the identity, $\Sigma = I$, since, if the noise is not white, the problem can be whitened via a standard change of coordinates \cite{calvetti2019hierachical}. Then, the optimization problem reduces to minimizing the unweighted least squares objective, $\|Ax - b\|^2$.

From a Bayesian perspective, ML estimation fails to account for prior information. Suppose that $x$ is drawn from a prior distribution, $\pi_{\text{prior}}(x)$. Then, given an observed sample $b$, the input $x$ has posterior distribution:
\begin{equation} \label{eqn: Bayes formula}
    \pi_{\text{post}}(x|b) \propto \pi(b|x) \pi_{\text{prior}}(x)
\end{equation}
where the constant of proportionality is determined by the likelihood of sampling $b$. Since $\pi(b)$ depends on $b$ alone, it can be absorbed into the normalization factor. 

Maximum a posteriori estimation (MAP) aims to maximize the posterior \eqref{eqn: Bayes formula} instead of the likelihood. The MAP objective only differs from the ML objective via weighting by the prior. Maximizing the posterior is equivalent to minimizing its negative logarithm, which separates into a least squares term associated with the likelihood, and a regularization term associated with the prior. The least squares term ensures fidelity to the data, while the regularizer promotes estimates $x$ which could have plausibly been sampled from the prior. 

Different prior assumptions regularize the objective differently. Thus, the prior can promote different features in the estimator. Often, $\pi_{\text{prior}}(x)$ is chosen to promote sparsity, as when $\pi_{\text{prior}}(x) \propto \exp(- \lambda \|x\|_p)$ for $p < 2$. Then, MAP estimation reduces to the standard $\ell_p$ regularized problem \eqref{penalty}. Varying the prior varies the shape of the regularizer, shifting the estimator. Here we aim to study how the estimator, and associated estimation problem, changes as the prior changes.

We focus on the Gaussian hierarchical model introduced in \cite{calvetti2008hypermodels}. See \cite{calvetti2015hierarchical,calvetti2019brain,calvetti2020sparse,calvetti2020sparsity,pragliola2020overcomplete} for examples. See \cite{agrawal2021variational} and \cite{kim2022hierarchical} for extensions that allow uncertainty quantification via variational inference, and nonlinear forward models via Kalman filtering. 

The Gaussian hierarchical model supposes that $\epsilon$ is multivariate Gaussian, and that $x$ is drawn from a conditionally Gaussian prior distribution with variances $\theta$ drawn from a generalized gamma distribution. In particular, we assume that $\pi(x|\theta)$ is $\mathcal{N}(0,D_{\theta})$ where $D_v$ denotes a diagonal matrix with diagonal entries specified by the vector $v$. Note that $\theta \in \mathbb{R}^n$ is itself a vector of unknown variances that must be estimated. A hyper-prior follows. The variances $\theta$ are drawn independently from a generalized gamma distribution with parameters $r$, $\eta$ and $\vartheta_j$. Thus:
\begin{equation}
    \pi_{\text{hyper}}(\theta_j|r,\eta,\vartheta_j) =  \frac{|r|^n}{\Gamma(\beta)^n} \prod_{j=1}^n \frac{1}{\vartheta_j} \left( \frac{\theta_j}{\vartheta_j} \right)^{\eta + 1/2} \exp \left(-\left( \frac{\theta_j}{\vartheta_j} \right)^r \right)
\end{equation}
where $\eta = r \beta - 3/2$ \cite{calvetti2009conditionally,calvetti2020sparse}. 

Here $r$, $\eta$, and $\vartheta$ act as hyperparameters. The hyperparameter $r \in \mathbb{R} \setminus \{0\}$ is a shape parameter. The hyperparameter $\eta$ is also a shape parameter chosen so that $\beta > 0$. The hyperparameters $\vartheta_j > 0$ are scale parameters and can be absorbed into the definition of $x$ and $\theta$ via a standard scaling \cite{calvetti2020sparse}. 

The full posterior is $\pi(x,\theta|b) \propto \pi(b|x) \pi_{\text{prior}}(x|\theta) \pi_{\text{hyper}}(\theta|r,\eta,\vartheta)$. Note that the prior $ \pi_{\text{hyper}}(\theta|r,\eta,\vartheta)$ is not a conjugate prior, so the posterior is non-Gaussian. The negative logarithm of the posterior, or Gibbs energy, is:
\begin{equation}
\begin{aligned}
    \mathcal{G}(x, \theta) & = \frac{1}{2}\|b-\mathrm{A} x\|^{2}+\frac{1}{2} \sum_{j=1}^{n} \frac{x_{j}^{2}}{\theta_{j}}-\eta \sum_{j=1}^{n} \log \frac{\theta_{j}}{\vartheta_{j}}+\sum_{j=1}^{n}\left(\frac{\theta_{j}}{\vartheta_{j}}\right)^{r}\\
    & = \frac{1}{2}\|b-\mathrm{A} x\|^{2} + \mathcal{P}(x,\theta|r,\eta,\vartheta)
\end{aligned}
\end{equation}
up to an additive constant associated with normalization of the posterior. Here $\mathcal{P}(\cdot)$ is the penalty term associated with the prior. Note that the fidelity term, $\frac{1}{2}\|b-\mathrm{A} x\|^{2}$, depends exclusively on the likelihood, and thus the error model, but is entirely independent of the prior and thus the hyperparameters.

The hierarchical model is chosen for two reasons. First, MAP estimates that minimize $\mathcal{G}(x,\theta)$ can be computed efficiently using a coordinate descent scheme that alternately optimizes over $x$ and $\theta$. Let $\{z^k\}=\{(x^k,\theta^k\})$ denotes a sequence of iterates indexed by $k$. Then, the iterative alternating scheme (IAS) proceeds by:
\begin{enumerate}
    \item \textbf{Updating $x$ given $\theta$: } $x^{k+1} = \text{argmin}_{x \in \mathbb{R}^{n}}\{ \|Ax - b \|^2 + \|D_{\theta^k}^{-1/2} x\|^2 \}$
    \item \textbf{Updating $\theta$ given $x$: } $\theta^{k+1} = \text{argmin}_{\theta \in \mathbb{R}^{n+}}\{ \mathcal{P}(x^{k+1},\theta|r,\eta,\vartheta) \}$.
\end{enumerate}

Minimizing the fidelity term given fixed $\theta$ reduces to a Tikhonov penalized least squares problem. The penalty term separates into a sum of terms involving each variance $\theta_j$, so the second step can be evaluated via a precomputed update function that optimizes the penalty term given $x$ \cite{calvetti2020sparse}. Thus, IAS reduces to a reweighted least squares algorithm. The convergence of IAS is studied in \cite{calvetti2019hierachical}, where it is shown to converge quadratically off the support of a sparse signal $x$, but linearly on the support. The method converges linearly on the support, since, in order to increase a particular $x_j$, the corresponding variance $\theta_j$ must also grow. Since IAS updates $x$ and $\theta$ separately, it is forced to make incremental progress when converging on the support. Nevertheless, initial convergence off the support is often quick, especially if the least squares step is implemented efficiently \cite{calvetti2018bayes,calvetti2019hierachical}. These observations suggest the need for a second-order optimization scheme that can update $x$ and $\theta$ simultaneously. Such a scheme could accelerate IAS.

Second, the hierarchical model is chosen since it does not use conjugate priors, so the prior changes the form of the posterior distribution. By varying the shape parameters $r$ and $\eta$, a user can continuously adjust the form of the posterior. When $r > 0$ the hierarchical model induces an effective regularizer (defined by evaluating the penalty term at its minimum in $\theta$ given $x$), which acts as a smoothed $\ell_p$ penalty where $p$ depends on $r$, and the degree of smoothing depends on $\eta$. If $\eta$ converges to zero then the regularizer converges to the corresponding $\ell_p$ penalty, and the MAP estimation problem reduces to the standard $\ell_p$ regularized least squares problem \ref{penalty} \cite{calvetti2019hierachical}. Here, $p = 2 \frac{r}{1+r}$. Thus, when $r = 1$ the effective regularizer acts as a smoothed $\ell_1$ penalty, as $r$ approaches zero it approaches a smoothed $\ell_0$ penalty, and as $r$ approaches infinity it approaches an $\ell_2$ penalty. The smaller $r$, the more the prior promotes sparsity. 

Note that, because the regularizer is smoothed when $\eta > 0$, it only promotes quasi-sparsity in the MAP estimates. That is, MAP estimates typically have entries at two different scales. A small subset of entries are large (``on the support"), and the remainder are small but non-zero (``off the support") . The degree of shrinkage off the support is determined by $\eta$. 

Here, we investigate the relation between the hyperparameters and estimators.

\section{Introduction}

The MAP estimator of the Bayesian hierarchical model depends on the hyperparameters, which encode the underlying assumptions.
While the relation between the shape of the regularizer and the hyperparameters is well understood, the relation between MAP estimator and hyperparameters has only been studied via a few select examples with hand-chosen parameters. The explicit dependence of the estimator on the hyperparameters has not been explored in detail. 

Any MAP estimator is the solution to an optimization problem. Just as the MAP estimator depends on the hyperparameters, so does the optimization problem that defines it. When $r \geq 1$ the objective function $\mathcal{G}$ is globally convex, and IAS is guaranteed to converge to its global minimum. When $r < 1$ the objective is convex inside an $\ell_{\infty}$ ball (for sufficiently small $x$ and $\theta$) whose size is determined by $\eta$. In this setting, the MAP estimation problem is a non-convex optimization problem prone to local minima. The non-convex regime poses a numerical challenge. It produces sparser results but is less robust since optimizers may identify local minima. While sample IAS results in the non-convex regime can recover sparse solutions more accurately, there is no guarantee that the solutions returned by IAS are the global minimizer. Such a guarantee is typically impossible. How, then, to design an optimizer which chooses a local minimizer in a principled fashion?

Calvetti et.~al.~propose a hybrid approach in \cite{calvetti2020sparsity}. First, solve for the MAP estimate using hyperparameters chosen so that the objective is convex. Then initialize at the convex solution when solving in the non-convex regime. Thus, the hybrid scheme guides the solution in the non-convex regime via a solution to a nearby convex problem. Note that the nearby convex problem is derived using the same family of prior distributions. This convex relaxation approach produces more robust solutions. 

We study the relationship between the underlying assumptions, as encoded by the hyperparameters, and the MAP estimator. Unlike the hybrid approach, which jumps discontinuously between assumptions, we track the motion of the MAP estimate as the hyperparameters change continuously. Thus, a user can gradually tune assumptions, study the sensitivity of estimates, explore the space of solutions, and test how prior knowledge informs inference. Consider the following exploration method. A user chooses a path through hyperparameter space, then follows the estimator as the hyperparameters change.  By selecting a path that starts in a convex regime, and ends in a non-convex regime, the user may find solutions in the non-convex regime that are, if not global minimizers, at least consistent with the global minimizer in a convex regime. By path following, the user can select a minima in the non-convex regime that is connected to a unique solution when convex.  

Various authors have developed path-following algorithms which trace the solution under varying $\lambda$ to the regularized optimization problem:
\begin{equation} \label{eqn: traditional regularized problem}
\hat{x}(\lambda)=\arg \min _x \left\{ \sum_i \mathcal{C}\left(b_i, a_i \cdot x\right)+\lambda \mathcal{P}(x) \right\}
\end{equation}
where $a_i$ is the $i_{th}$ row of $A$, $\mathcal{C}$ is a convex loss function, and $\mathcal{P}$ is a penalty term. Some algorithms exploit the piecewise-linearity of the solution path to generate it exactly \cite{efron2004least,rosset2007piecewise,hastie2004entire}. These methods identify the points where the active set (variables with nonzero coefficients) changes, then draw the entire path via linear interpolation.  

When the solution path is not piecewise linear, we can only obtain an approximation at selected parameter values. Various sources have considered using ODE systems to approximate the solution path. Those systems can be derived from a first order optimality condition \cite{zhu2021algorithmic}, stationarity condition \cite{zhou2014generic}, or an extension of least angle regression (LARS) \cite{wu2011ordinary}. For example, Zhu \cite{zhu2021algorithmic} considers numerical ODE and Newton based methods for recovering $\ell_2$ regularized solution paths. Adaptive grids of regularization parameters is considered to balance the trade-off between complexity and efficiency \cite{ndiaye2019safe,ndiaye2021continuation}. We will introduce a predictor-corrector algorithm, with Newton correction, for solving an ODE derived by first-order optimality. Predictor-corrector algorithms have been considered by other authors. Allgower and Georg \cite{allgower1993continuation} introduced  predictor–corrector strategies for path-following with varying parameters. Rosset \cite{rosset2004following} introduced a path-following algorithm which uses a Newton corrector step and a redundant predictor step. 
Park and Hastie \cite{park2007l1} propose a predictor-corrector algorithm with an adaptive step size that increases accuracy. Wang \cite{wang2014optimal} proposes an approximate regularization path-following method for nonconvex loss or penalty functions. We also compare our proposed method to the repeated application of coordinate descent (IAS) \cite{calvetti2019hierachical}. For related discussions, see \cite{friedman2010regularization,mazumder2011sparsenet,simon2011regularization}. 
Note that the problem we aim to solve is more general than \eqref{eqn: traditional regularized problem}, since we aim to change the shape of the regularizer.


Our predictor-corrector algorithm uses a Newton type step for both prediction and correction. The success of our method depends on efficiently solving a series of linear systems involving the Hessian of the objective function. The Hessian presents an interesting numerical challenge since it is simultaneously large, extremely ill-conditioned when the estimates are near to sparse, highly structured, and near to low rank. Moreover, the Hessian changes continuously along the solution path, so each linear system is closely related to a linear system we have solved before. We introduce a preconditioning strategy that exploits the structure of the Hessian to efficiently and stably solve the required sequence of linear systems. Since Newton methods also require solving linear systems against the Hessian, the preconditioner can also be applied to perform Newton iteration with fixed hyperparameters. Thus, our method can significantly accelerate the convergence of MAP estimation by replacing IAS with Newton iteration near convergence.

The paper outline follows. Section \ref{sec: theory} introduces the theory which grounds our path-following approach. First, we derive a system of ODE's that govern the motion of the estimator (see Section \ref{sec: ODEsystem}). We simplify the Hessian in Section \ref{sec: hessian}. Section \ref{sec: invertibility} establishes the almost sure invertibility of the Hessian, then extends invertibility to uniqueness. In particular, we show that MAP estimates only bifurcate at locations where the Hessian is singular. Section \ref{sec: Methods} introduces the path-following algorithm. The success of the algorithm depends on a carefully designed preconditioner. We develop the preconditioner in Section \ref{sec: pre}, then use it to accelerate estimation in Section \ref{sec: acceleration}.

We conclude with a series of numerical experiments (see Section \ref{sec: example}). The experiments are all adapted from previous work on hierarchical Bayesian models (c.f.~\cite{agrawal2021variational,calvetti2019hierachical}). The predictor-corrector method is the most accurate of any method tested and is sufficiently fast to use on problems of intermediate scale. We show how to explore the space of possible estimates by varying the hyperparameters, and that using path-following to solve non-convex MAP estimation problems provides more robust results than direct solution. Note that, in the non-convex regime, points on the solution path are not necessarily global minimizers but are, at least, consistent with a global minimizer in the convex regime under smooth changes to the assumptions. We report the time cost of each stage of our algorithm, test the efficacy of the preconditioner, and monitor the conditioning of the Hessian to track the approach to possible bifurcations. We also show that Newton acceleration can rapidly improves MAP estimation.

\section{Theory} \label{sec: theory}

\subsection{ODE System} \label{sec: ODEsystem}
Let $\eta = r \beta-\frac{3}{2}$. Then, we aim to optimize the negative log-posterior 
\cite{calvetti2020sparse}: 
\begin{equation}\label{obj}
\mathcal{G}(x, \theta)=\mathcal{G}(x, \theta \mid r, \eta,  \vartheta)=\frac{1}{2}\|b-\mathrm{A} x\|^{2}+\frac{1}{2} \sum_{j=1}^{n} \frac{x_{j}^{2}}{\theta_{j}}-\eta \sum_{j=1}^{n} \log \frac{\theta_{j}}{\vartheta_{j}}+\sum_{j=1}^{n}\left(\frac{\theta_{j}}{\vartheta_{j}}\right)^{r}
\end{equation}

Let the hyperparameters $r(t)$, $\eta(t)$, $\vartheta(t)$ follow a differentiable path through the hyperparameter space parameterized by $t$. Let $\psi(t)$ denote the collection of hyperparameters  $\left(r(t), \eta(t), \vartheta(t) \right)$. Next, let $x(t),\theta(t)$ denote a minimizer of the MAP objective \eqref{obj} corresponding to hyperparameters $\psi(t)$.
Let $z$ denote the pair $[x; \theta]$ and $z_{*}(t)$ denote the minimizer at time $t$. 

First order optimality requires that, at all times $t$, 
\begin{equation} \label{eqn: first order optimality}
\nabla_{z} \mathcal{G}\left(z_{*}(t) \mid 
\psi(t)\right)=0.
\end{equation}

Equation \ref{eqn: first order optimality}  constrains the motion of all possible minimizers as the hyperparameters vary. To study that motion, differentiate with respect to time: 
$$
\begin{aligned}
\partial_{t} \nabla_{z} \mathcal{G}\left(z_{*}(t) \mid \psi(t) \right) 
= & H\left(z_{*}(t) \mid \psi(t) \right) \frac{d}{d t} z_{*}(t)+\partial_{r} \nabla_{z} \mathcal{G}\left(z_{*}(t) \mid \psi(t) \right) \frac{d}{d t} r(t) +   \\
& \partial_{\eta} \nabla_{z} \mathcal{G}\left(z_{*}(t) \mid \psi(t) \right) \frac{d}{d t} \eta(t) +\partial_{\vartheta} \nabla_{z} \mathcal{G}\left(z_{*}(t) \mid \psi(t) \right) \frac{d}{d t} \vartheta(t).
\end{aligned}
$$
Here, $H\left(z_{*}(t) \mid \psi(t) \right)$ is the Hessian of the objective function evaluated at the minimizer $z_*(t)$ given hyperparameters $\psi(t)$. 

The right hand side of the constraint \eqref{eqn: first order optimality} equals zero at all times, so its time derivative is zero. Therefore, when differentiable, $z_{*}(t)$ satisfies the ODE:
\begin{equation}\label{ODE}
\begin{aligned}
& H\left(z_{*}(t) \mid \psi(t) \right) \frac{d}{d t} z_{*}(t) = - \nabla_z \left(\nabla_\psi \mathcal{G}(z_*(t) \mid \psi(t)) \cdot \frac{d}{dt} \psi(t) \right)
\end{aligned}
\end{equation}

The existence and uniqueness of solutions to \eqref{ODE} depends on the invertibility of $H$. If $H$ is invertible, then there is an explicit expression for $\frac{d}{dt}z_{*}(t)$. Simply multiply Equation \eqref{ODE}  by $H^{-1}$ on both sides. If $H$ is not invertible, then the system of equations may not have a solution, or may admit infinitely many. The invertibility of the Hessian is explored in detail in Section \ref{sec: invertibility}. There we show that the Hessian is invertible for almost all pairs $z = [x,\theta]$ where $\theta$ is optimized given $x$.

\begin{figure}
\begin{center}
   \begin{minipage}{0.24\textwidth}
     \centering
     \includegraphics[width=.9\linewidth]{ 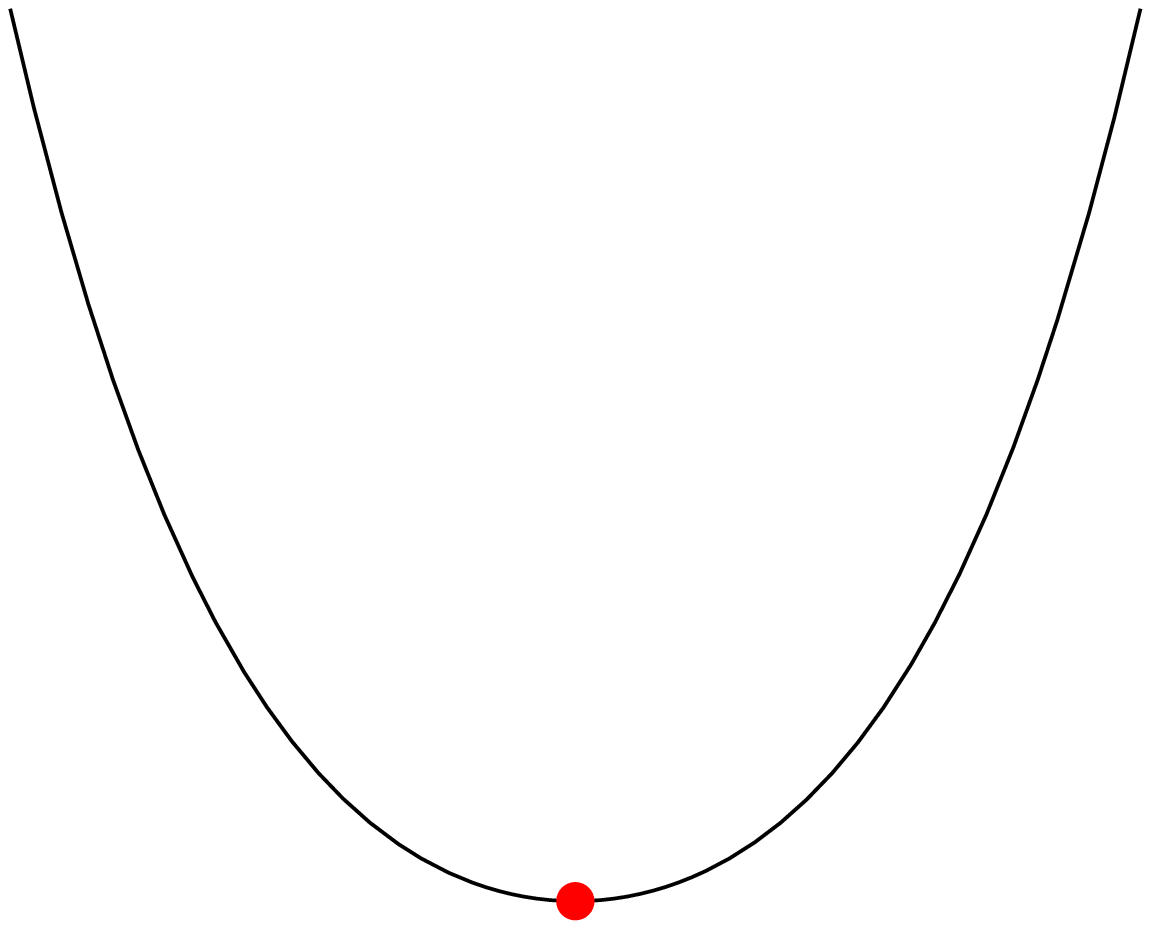}
   \end{minipage}\hfill
   \begin{minipage}{0.23\textwidth}
     \centering
     \includegraphics[width=.9\linewidth]{ 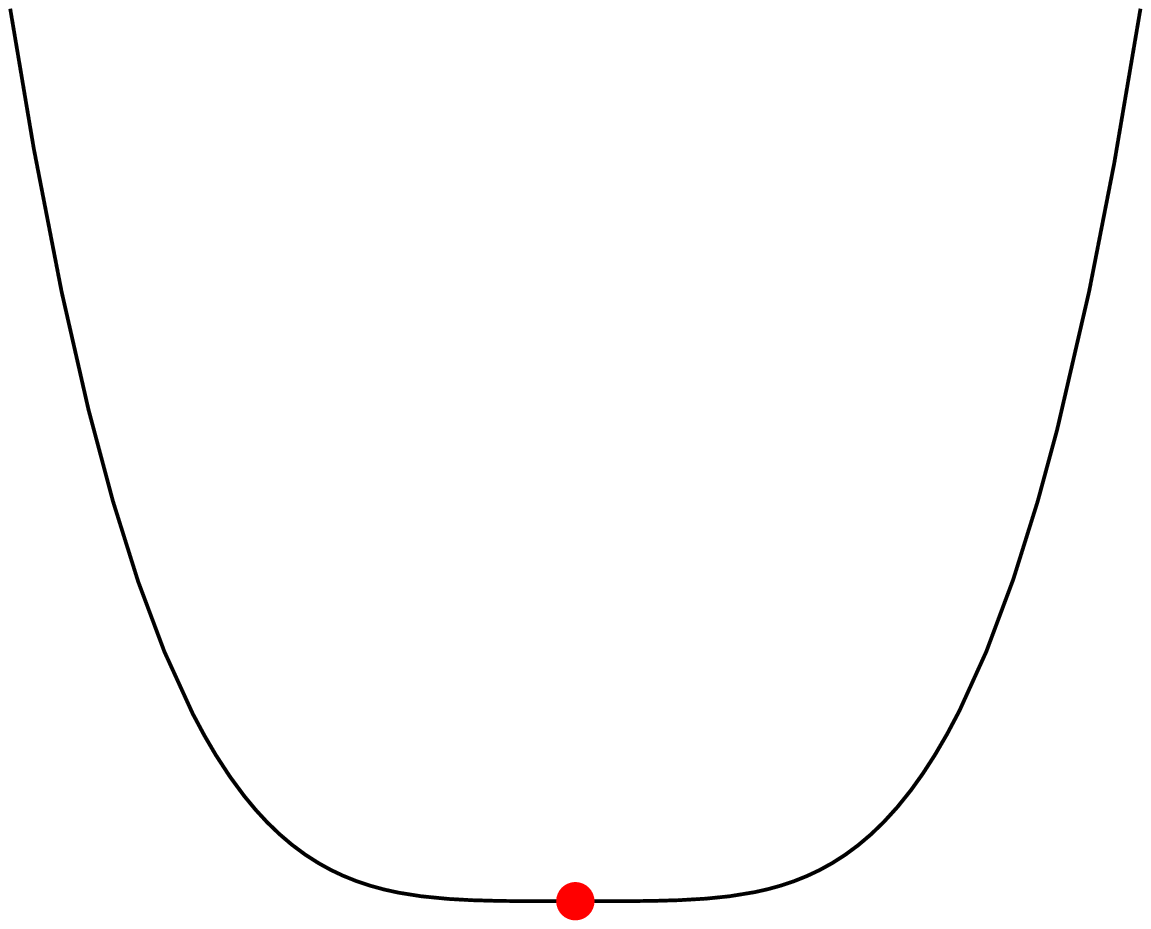}
   \end{minipage}
      \begin{minipage}{0.25\textwidth}
     \centering
     \includegraphics[width=.9\linewidth]{ 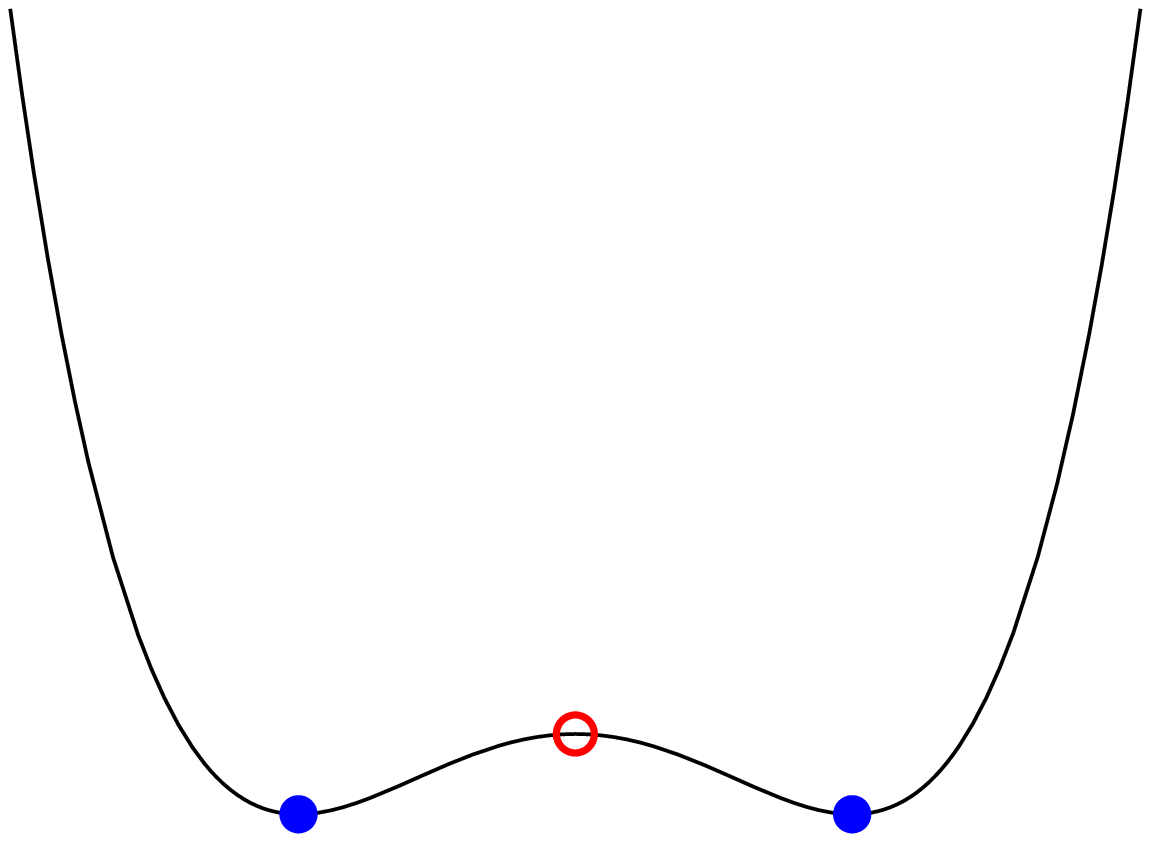}
   \end{minipage}
      \begin{minipage}{0.26\textwidth}
     \centering
     \includegraphics[width=.9\linewidth]{ 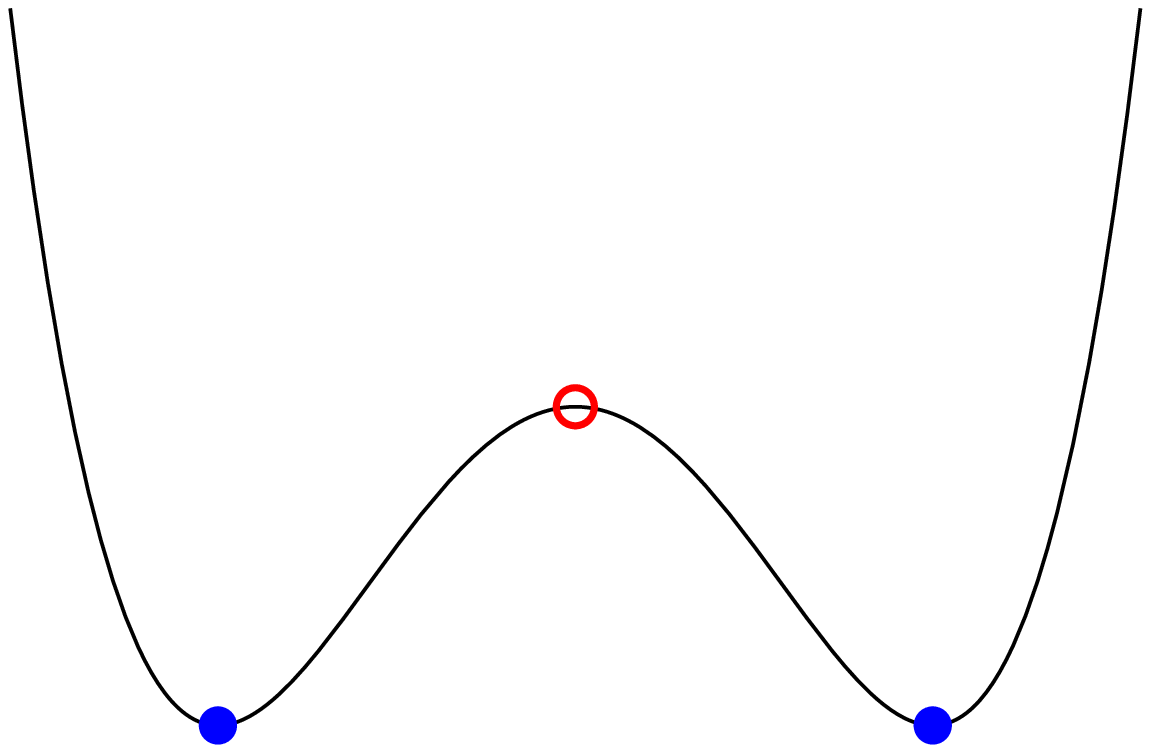}
   \end{minipage}
   \end{center}
   \caption{ A Minimum bifurcating. Note that, at the instant of the bifurcation (panel 2), the objective function has second derivative equal to zero at the minimizer.}
\label{fig:bif}
\end{figure}

The Hessian determines the shape of the objective function \eqref{obj} near the minimizer via the expansion, 
$$
\mathcal{G}(z) = \mathcal{G}(z_{*}) + (z-z_{*})^{\top}H(z_{*})(z-z_{*}) + \mathcal{O}\left((z-z_{*})^{3} \right).
$$
The smaller the smallest singular value of the Hessian, the flatter the objective, so the more uncertain the posterior. For example, the Laplace approximation to the posterior uses a Gaussian distribution with mean $z_*$ and covariance $H^{-1}$. Thus, the inverse of the Hessian approximates the covariance in the posterior (c.f.~\cite{agrawal2021variational}).

If the smallest singular value equals zero, then the minimizer is not unique, and the objective is flat (to second order) along the corresponding singular vector and $H$ is not invertible. This situation may occur if, as the hyperparameters change in time, a minimizer becomes a saddle point, as in a pitchfork bifurcation (see Figure \ref{fig:bif}). More strongly, when the Hessian $H$ is invertible, the path following ODE \eqref{ODE} admits a unique solution. Therefore, bifurcations only occur where $H$ is singular. We will show that $H$ is invertible almost everywhere, so bifurcations occur almost nowhere.

\subsection{Simplifying the Hessian}\label{sec: hessian}

Equation \eqref{ODE} is a linear system involving the Hessian. Similar systems arise whenever a Newton method is used to minimize the objective \eqref{obj}. The Hessian of the objective function takes the block form \cite{calvetti2020sparse}:
\begin{equation}
H(z \mid \psi)= \left[\begin{array}{cc}
A^{\top} A+D_{\theta}^{-1} & -D_{\theta}^{-2} D_{x} \\
-D_{\theta}^{-2} D_{x} & D_{\theta}^{-3} D_{x}^{2}+\eta D_{\theta}^{-2}+r(r-1)\vartheta^{-r} D_{\theta}^{r-2}
\end{array}\right]
\end{equation}
where $D$ denotes a diagonal matrix whose subscript fixes the diagonal entries.

Solving linear systems involving the Hessian can be expensive and unstable, especially when the Hessian is ill-conditioned. The Hessian is ill-conditioned when the solution $x$ is compressible, since many of the weights $\theta$ are very small. Taking $\eta$ to zero sharpens the solution by making these weights even smaller. In \cite{calvetti2020sparse}, the authors introduce a lower  bound on the optimal $\theta$, 
\begin{equation}\label{thetalowerbound}
    \theta \ge \vartheta\left(\frac{\eta}{r}\right)^{1/r}.
\end{equation}
The bound is tight when $x = 0$ and $\theta$ is optimized given $x$. Thus,  $\theta$ corresponding to small $x$ vanish as $\eta$ goes to zero and terms of the form $\theta^{-1}, \theta^{-2}, \theta^{-3}$ diverge. 

To solve equation \eqref{ODE}, we take advantage of the structure of Hessian which, after an appropriate change of variables and scaling, is sparse, and nearly a low rank perturbation of an explicitly invertible tridiagonal matrix. The subsequent section introduces the transformations needed to exploit that structure. These transformations drastically reduce the conditioning of the Hessian. The analysis also leads to an approximate inverse, which, if used as a preconditioner, drastically reduces the computational cost needed to solve linear systems involving $H$.

\subsubsection{Change of Variables, Rescaling,
and Decoupling} \label{sec: change of variables}


First, change variables. Let $\phi=\log (\theta)$. Then:
\begin{equation}\label{Hessian}
H(z \mid \psi)=\left[\begin{array}{cc}
A^{\top} A+D_{\theta}^{-1} & -D_{\theta}^{-1} D_{x} \\
-D_{\theta}^{-1} D_{x} & \frac{1}{2}D_{\theta}^{-1} D_{x}^{2}+r^2\vartheta^{-r} D_{\theta}^{r}
\end{array}\right] .
\end{equation}
Here, $H$ stands for the Hessian with respect to $z = [x,\phi]$. Partials are evaluated with respect to $x$ and $\phi$, but are presented in terms of $\theta = \exp(\phi)$ for concision.

Recall that $\theta$ are the variances for each $x$, so $x^{2} / \theta$ is approximately order one in expectation. More precisely, $x^2/\theta$ is order $x^{2}$ for small $x$, and order $x^{2-\frac{2}{r+1}}=x^{\frac{2 r}{r+1}}$  for large $x$ when $r > 0$. Further asymptotics are available in \cite{calvetti2020sparse}.
Therefore, to avoid very large entries, we aim to pair $\theta^{-1}$ with $x^{2}$ and $\theta^{-1 / 2}$ with $x$.

Scale the Hessian from the left and right by the matrix $D_{[\theta^{1/2} ; 1]}$ where 1 stands for the vector of $n$ ones.
Then,
\begin{equation} \label{eqn: scaled Hessian}
D_{[\theta^{1/2} ; 1]} H D_{[\theta^{1/2} ; 1]} =\left[\begin{array}{cc}
D_\theta^{\frac{1}{2}} A^{\top} A D_\theta^{\frac{1}{2}} + I & -D_{\theta}^{-\frac{1}{2}} D_{x} \\
-D_{\theta}^{-\frac{1}{2}} D_{x} & \frac{1}{2}D_{\theta}^{-1} D_{x}^{2}+r^2\vartheta^{-r} D_{\theta}^{r}
\end{array}\right] = H_S .
\end{equation}
Here $H_S$ denotes the scaled Hessian. Both diagonal matrices are invertible, so any system involving $H$ can be converted to a system involving the $H_S$.

The scaled Hessian, \eqref{eqn: scaled Hessian}, breaks into two components associated with the fidelity and penalty terms of the objective. Accordingly, we separate $H_S$ into a fidelity term, $H_A$, and a penalty term $H_P$. The fidelity term is denoted $H_A$, since it depends primarily on the forward model, $A$.  Then, $H_S = H_A + H_P$ where:
\begin{equation} \label{eqn: fidelity and penalty}
\begin{aligned}
H_{A} &=\left[\begin{array}{cc}
D_{\theta}^{1 / 2} A^{\top} A D_{\theta}^{1 / 2} & 0 \\
0 & 0
\end{array}\right], 
H_{P} =\left[\begin{array}{cc}
I & -D_{\theta}^{-1 / 2} D_{x} \\
-D_{\theta}^{-1 / 2} D_{x} & \frac{1}{2} D_{\theta}^{-1} D_{x}^{2}+r^{2} \vartheta^{-r} D_{\theta}^{r}
\end{array}\right] .
\end{aligned}
\end{equation}

The penalty term, $H_P$ is always invertible. Specifically, the tridiagonal matrix can be broken into a product of triangular and diagonal matrices:
\begin{equation} \label{eqn: penalty decomp}
H_{P}=R^{\top} S R,
\end{equation}
where:
\begin{equation} \label{eqn: R and S}
\begin{aligned}
R =\left[\begin{array}{cc}
I & -D_{\theta}^{-1 / 2} D_{x} \\
0 & I
\end{array}\right], \quad
S =\left[\begin{array}{cc}
I & 0 \\
0 & \frac{1}{2} D_{\theta}^{-1} D_{x}^{2}+r^{2} \vartheta^{-r} D_{\theta}^{r}
\end{array}\right].
\end{aligned}
\end{equation}

The triangular factor $R$ is easily inverted by negating its off-diagonal block:
\begin{equation} \label{eqn: R inverse}
R^{-1}=\left[\begin{array}{cc}
I & D_{\theta}^{-1 / 2} D_{x} \\
0 & I
\end{array}\right].
\end{equation}

Hence $H_{P}^{-1}=R^{-1} S^{-1} R^{-\top}$.

Thus, after changing variables and scaling, the Hessian can be expressed as the combination of a fidelity term, and an explicitly invertible penalty term. When $x$ is compressible, most of the weights, $\theta$, are small, so $H_A$ is close to sparse. In fact, the nonzero block in $H_A$ is exactly the matrix introduced in \cite{calvetti2019hierachical}, where it was shown that $H_A$ is near to low rank when $x$ is compressible, and, that the effective rank of $H_A$ approaches the cardinality of the true support. Thus, the scaled Hessian is near to a low rank perturbation of the penalty term when the MAP solution $x$ is compressible.


\subsection{Invertibility of the Hessian} \label{sec: invertibility}

In this section, we show that the Hessian is invertible for almost all combinations of $x$ and $\theta$. Section \ref{sec:unique} establishes that, when the Hessian is invertible, solutions to Equation (\ref{ODE}) are unique, so the path-followin ODE admits unique solution for almost all combinations of $x$ and $\theta$. To simplify the analysis, we continue reducing the Hessian into simpler factors.

\subsubsection{Transforming the Hessian}

To study the invertibility of the Hessian, we study the scaled Hessian. The scaled Hessian is related to the original Hessian by $D_{[\theta^{1/2} ; 1]}$. Since $\theta \geq 0$ (see \eqref{thetalowerbound}), $D_{[\theta^{1/2} ; 1]}$ is invertible. Therefore, the Hessian $H$ is invertible if and only if the scaled Hessian $H_{S}$ is invertible.

As before, the scaled Hessian separates into a fidelity and penalty term,
$$
H_{S}=H_{A}+H_{P}=H_{A}+R^{\top}SR
$$

The triangular factor $R$ is invertible. So, $H_{S}$ is invertible if and only if $R^{-\top}H_{S}R^{-1}$ is invertible. The latter product gives:

\begin{equation*}
\begin{aligned}
R^{-\top}H_{S}R^{-1} & = R^{-\top}H_{A}R^{-1} + S\\ & =\left[\begin{array}{cc}
D_\theta^{\frac{1}{2}} A^{\top} A D_\theta^{\frac{1}{2}} + I & D_\theta^{\frac{1}{2}} A^{\top} A D_{x} \\
D_{x} A^{\top} A D_\theta^{\frac{1}{2}} &
D_{x} A^{\top} A D_{x} +  \frac{1}{2}D_{\theta}^{-1} D_{x}^{2}+r^2\vartheta^{-r} D_{\theta}^{r}
\end{array}\right] .
\end{aligned}
\end{equation*}

Assume $x_{i} \neq 0$ for each $i$. Then, the diagonal matrix $D_{[\sqrt{\theta},x]}^{-1}$ is invertible. Then, $H_S$ is invertible if and only if  $\hat{H}$ is invertible, where
$$
\hat{H}=D_{[\sqrt{\theta},x]}^{-\top} R^{-\top}H_{S}R^{-1} D_{[\sqrt{\theta},x]}^{-1}  =  \left[\begin{array}{cc}
A^{\top}A + D_{\theta}^{-1} & A^{\top}A \\
A^{\top} A  & A^{\top} A  
+\frac{1}{2} D_{\theta}^{-1}
+ r^{2} \vartheta^{-r} D_{\theta}^{r} D_{x}^{-2}
\end{array}\right].
$$

The matrix $\hat{H}$ is easier to analyze than $H$ or $H_S$ since only the diagonal depends on the variables and hyperparameters.  We use $H_S$ for numerics, and $\hat{H}$ for analysis.

To simplify  $\hat{H}$ further, we follow \cite{calvetti2020sparse} and introduce the non-dimensional parameters $u_{j}$ and $\xi_{j}$ such that
$$
x_{j}=\vartheta_{j}^{1/2}u_{j}, \qquad \theta_{j}=\vartheta_{j}\xi_{j}.
$$

In the non-dimensionalized coordinates, the matrix $\hat{H}$ is:
\begin{equation}\label{newhessian}
\hat{H} = \left[\begin{array}{cc} 
\hat{A}^{\top} \hat{A} + D_{\xi}^{-1} & \hat{A}^{\top} \hat{A} \\
\hat{A}^{\top} \hat{A} & \hat{A}^{\top} \hat{A} +\frac{1}{2} D_{\xi}^{-1} + r^{2} D_{\xi}^{r} D_{u}^{-2}
\end{array}\right]
\end{equation}
where $\hat{\mathrm{A}}$ is the column scaled version of $\mathrm{A}$, such that
$$
\hat{\mathrm{A}} = \mathrm{A} D_{\vartheta}^{1/2}.
$$

Since $\vartheta_{j} > 0$ for all $j$, the original Hessian $H$ is invertible if $\hat{H}$ is invertible and $u_j \neq 0$ for all $j$.

\subsubsection{Invertibility}

Here, we  show that the $n \times n$ Hessian matrix is  invertible for almost all $x$. To start, we consider a simplified version of the problem.
\begin{lemma}\label{naivenonsingular}
    Let $D_x$ be a diagonal matrix with diagonal entries $x$, and $A$ be a fixed matrix. Then $B(x) = A+D_x$ is invertible for almost all $x$. If the entries $x$ are continuously distributed, then $B(x)$ is almost always invertible.
\end{lemma}

\begin{proof}
Let $D_x = \mathrm{diag}(x_{1}, \cdots , x_{n})$ and $A$ have columns $a_1, a_2, \hdots, a_n$. Then
$$
B(x) = A+D_x = \left[ \begin{array}{ccc}
a_{1}+x_{1}e_{1} & \cdots & a_{n}+x_{n}e_{n}  
\end{array}\right],
$$
where $e_{i}$ is the $i^{th}$ column of the $n \times n$ identity matrix.

Let $b_{i} = a_{i} + x_{i}e_{i}$. Fix $x_{1} \dots x_{n-1}$, Let $\mathcal{S} = \text{span}(b_{1}, \dots b_{n-1})$. The remaining column, $b_{n}(x)$ depends only on $x_{n}$. The
set of possible $b_{n}(x)$ is a line, $\mathcal{L}$,  which passes through $a_n$ in the $n^{th}$ coordinate direction, $e_n$. There three possible geometries relating $\mathcal{S}$ and $\mathcal{L}$ labelled Intersection, Parallel, and Inside in Fig \ref{fig:location}.

\begin{figure}
    \begin{center}
  \begin{minipage}{1\textwidth}
     \centering
     \includegraphics[width=.7\linewidth]{ 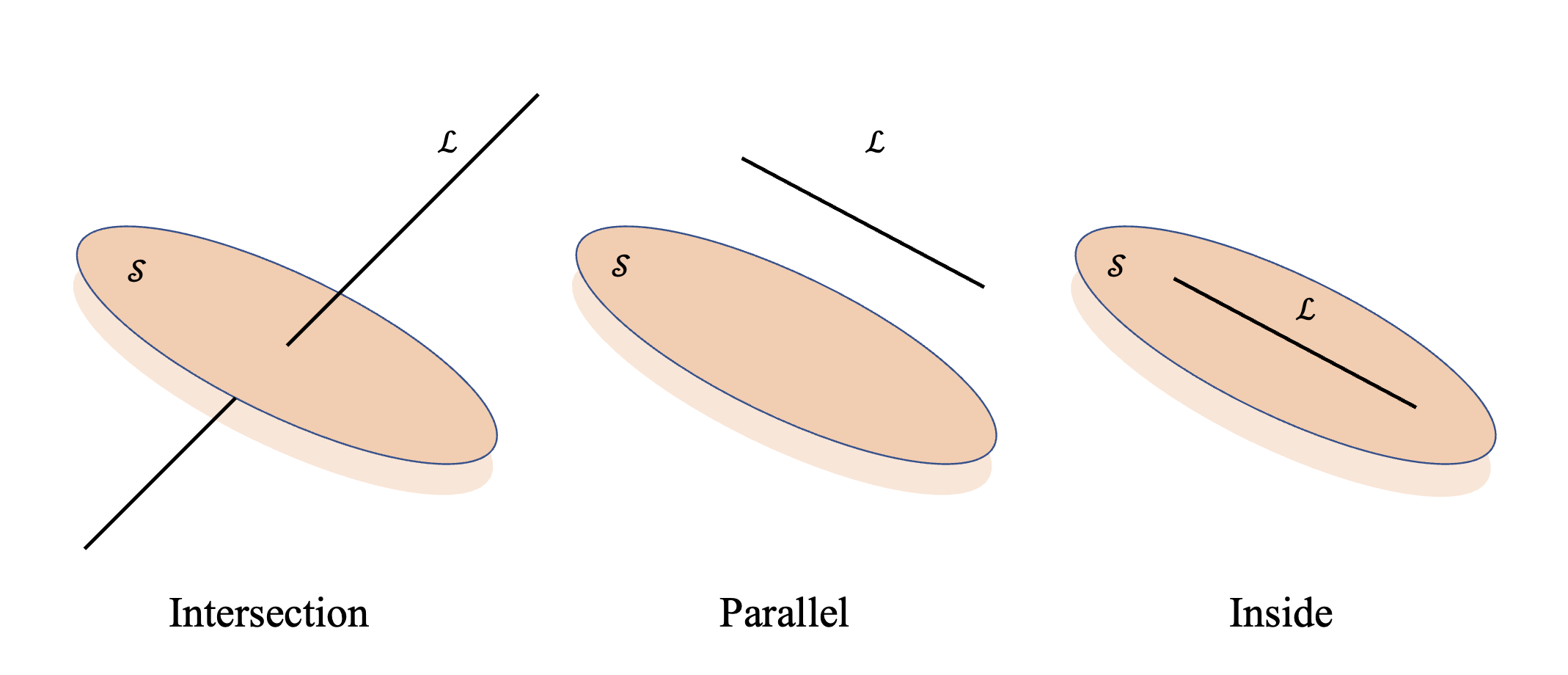}
  \end{minipage}\hfill
  \end{center}
  \caption{Possible line–plane intersection.}
\label{fig:location}
\end{figure}

\begin{enumerate}

\item{\textbf{Intersection:}} In the Intersection case, $e_n$ is not contained in $\mathcal{S}$ so $\mathcal{L}$ and $\mathcal{S}$ intersect at a single point. Then there is only one $x$ such that $b_{n}(x)$ is in $\mathcal{S}$. So, if $e_n \notin \mathcal{S}$, then, for any choice of $x_1,\hdots x_{n-1}$, there is precisely one $x_n$ such that $B(x)$ is singular. Thus, the set of $x$ where $B(x)$ is singular and $e_n \notin \mathcal{S}$ is a manifold of codimension 1 and has measure zero.                                                               

\item{\textbf{Parallel: }} Alternatively, suppose that $e_n \in \mathcal{S}$. Then $\mathcal{L}$ is parallel to $\mathcal{S}$, so is either contained inside of $\mathcal{S}$ for all $x$ (Inside), or never intersects the subspace (Parallel). In the parallel case, $b_n(x)$ is never in $\mathcal{S}$, so $B(x)$ is invertible for all $x_n$. 

\item{\textbf{Inside:} } Suppose that $\mathcal{L} \subset \mathcal{S}$. Then, $a_{n} \in \mathcal{S}$, and $e_{n} \in \mathcal{S}$. If $e_n \in \mathcal{S}$, then, by definition, there exist coefficients $t = t_{1}, \dots, t_{n-1}$ not all zero, such that $\sum_{i=1}^{n-1} t_{i}b_{i}=e_{n}$.  Let $B^{(m)}$ denote the $m\times m$ minor of $B$ consisting of its first $m$ rows and columns. Since all but the $n^{th}$ entry of $e_n$ are zero, $e_n \in \mathcal{S}$ requires $B^{(n-1)} t = 0$. Therefore, the line $\mathcal{L}$ is only contained in $\mathcal{S}$ if $B^{(n-1)}$ is singular. If the minor is nonsingular, then $\mathcal{L}$ is never contained in $\mathcal{S}$, so $B(x)$ is only singular on a set of measure zero.

\end{enumerate}

Now the argument recurses. 
If $B^{(m-1)}(x)$ is singular on a set of measure zero, so is $B^{(m)}(x)$. Induction follows. 

All that remains is the base case, $m = 1$. If $m = 1$ then $B^{(1)}(x) = a_{11} + x_1$ is a $1\times1$ matrix, so is singular if and only if it is zero. But $a_{11} + x_1 = 0$ only holds for exactly one $x_1$. Thus, $B^{(1)}(x)$ is singular on a set of measure zero in $\mathbb{R}$. Then, by induction, $B(x)$ is singular on a set of measure zero in $\mathbb{R}^n$.
\end{proof}

To apply Lemma \ref{naivenonsingular}, note that $\hat{H}$ \eqref{newhessian} is the sum of a fixed matrix and a diagonal matrix that depends on $u$ and $\xi$,
$$
\hat{H} = \left[\begin{array}{cc} 
\hat{A}^{\top} \hat{A} & \hat{A}^{\top} \hat{A} \\
\hat{A}^{\top} \hat{A} & \hat{A}^{\top} \hat{A}
\end{array}\right]
+ \left[\begin{array}{cc}
D_{\xi}^{-1} & 0 \\
0 & \frac{1}{2} D_{\xi}^{-1} + r^{2} D_{\xi}^{r} D_{u}^{-2}
\end{array}\right].
$$

Without restricting to optimal solutions, $\hat{H}$ is of the form given in Lemma \ref{naivenonsingular} since $u$ and $\xi$ can be chosen independently. Thus, Lemma \ref{naivenonsingular} establishes that the Hessian is invertible for almost all $x$ and $\theta$. Crucially, this result is not strong enough to ensure invertibility along solution paths, since, at solutions to the optimization problem, $z = [x,\theta]$ is restricted to a manifold specified by $\xi = \varphi(u)$ (see Equation \eqref{eqn: manifold}). Let $\phi$ denote the corresponding function relating the original coordinates. Then, $\theta_j = \phi(x_j)$ defines the solution manifold. To show that the ODE is well defined almost everywhere we must show that the Hessian is invertible almost everywhere on the solution manifold. Note that, in practice Lemma \ref{naivenonsingular} is enough to ensure that the Hessian experienced numerically will be invertible. Any numerical method will retain random errors, so the iterates will almost never lie exactly on the solution manifold. 

\begin{theorem}[Invertibility] \label{thm: invertibility}
The Hessian matrix $H(z)$ \eqref{Hessian} is invertible for almost all $z = [x,\theta]$ such that $\theta = \phi(x)$. That is, for almost all $z$ in the solution manifold.  
\end{theorem}

The proof proceeds inductively. \ref{sec:2 by 2} establishes the 2×2 base case. Each unknown
entry of the signal is coupled to an unknown variance, so we must introduce two columns
and two rows to the Hessian at a time. \ref{naivenonsingular} proves that the induction hypothesis holds when a new column and row corresponding to $x_j$ are added. We show that the induction hypothesis also holds when appending a row and column corresponding to $\theta_j$ by explicitly converting a linear dependency check into an algebraic condition that holds almost nowhere. See \ref{sec:invertibility} for details.

\subsection{Uniqueness of Solutions}\label{sec:unique}

When $H$ is invertible, the ODE \eqref{ODE} admits unique solutions. Then, since the Hessian is invertible at almost everywhere (Theorem \ref{thm: invertibility}), the ODE (\ref{ODE}) admits unique solution for all most all $z$. 

\begin{theorem}\label{thm: uniqueness}
Assume $\psi(t)$ is continuously differentiable. When Hessian $H$ is invertible at $\left(z_{0}, \psi(t_{0})\right)$, the path following ODE \eqref{ODE} with initial value $z(t_{0})=z_{0}$ has a unique solution $z(t)$ on a closed interval containing $t_{0}$.
\end{theorem}

See \ref{sec: Uniqueness} for the proof of \ref{thm: uniqueness}.

Theorem \eqref{thm: uniqueness} ensures that the ODE \eqref{ODE} admits a unique solution on an open neighborhood of any initial point where the Hessian $H$ is invertible. Theorem \eqref{thm: invertibility} ensures that $H$ is invertible for almost all $z$. Therefore, the solution to the path following ODE remains unique for almost all $z$.



\section{Methods} \label{sec: Methods}

\subsection{Path Following}
Here, we propose a Predictor-Corrector Algorithm, that traces the entire path of the estimates $(x,\theta)$ as the hyperparameters $\left(r, \beta, \vartheta\right)$ vary.

Let $\psi(t) = \left(r(t), \beta(t), \vartheta(t)\right)$ be a smooth path through the hyperparameter space that starts from $\psi(0)$ and arrives at  $\psi(T)$. Let $z_*(0) = (x(0), \theta(0))$ be a minimizer at time $t=0$. If the solutions depend continuously on the hyperparameters then the path $\psi(t)$ corresponds to a path of solutions $z_*(t)$. Unlike the IAS algorithm \cite{calvetti2020sparse}, where the hyerparameters are fixed or hybrid IAS algorithm \cite{calvetti2020sparsity}, where the hyperparameters jump, the Predictor-Corrector Algorithm updates the hyperparameters continuously.  It allows a user to explore the space of possible solutions, to study the sensitivity and robustness of solutions to changes in the hyperparameters, and to study how specific changes in those assumptions change the solution. 

We are particularly interested in paths that start at a convex region and end at a non-convex region because, when convex, the MAP estimation problem admits a unique solution that can be accurately obtained by IAS. That solution can initalize the Predictor-Corrector algorithm. The convex regime does not strongly promote sparsity, so a non-convex prior model that strongly promotes sparsity is often desired. In this context, the path-following approach acts as a convex relaxation of a non-convex problem. The non-convex problem may admit local minima, so it is usually not possible to recover a global minimizer. Nevertheless, it is possible to find a minimizer which is a consistent extension of the unique global minimizer in a convex relaxation of the non-convex problem. Thus, when moving from a convex to non-convex setting, path-following provides a principled method for selecting among solutions of the non-convex problem. We provide examples to show that these solutions are more accurate than direct minimization (see Section \ref{sec: deconvolution}).  

\subsection{Predictor-Corrector}

The hyperparameters change continuously, so the corresponding  solutions should also change smoothly. Therefore, the solution at the current hyperparameters provides a good starting point for finding the solution at nearby hyperparameters. The predictor step predicts how the current solution will change after changing the hyperparameters. It provides an initial estimate to the solution at nearby hyperparameters. Correction revises the prediction. The algorithm iteratively alternates between an ODE (\ref{ODE}) based predictor step and a Newton based corrector step. All steps use warm starts (are initialized at the previous solutions).

\subsubsection{The Prediction Step}
The ODE system (\ref{ODE}) introduced in section \ref{sec: ODEsystem} enables prediction. It accounts for the rate of change in the hyperparameters and allows larger steps than correction alone. By solving \ref{ODE}, we predict the solution at the next time step. For simplicity, we take an Euler step $ z^{p}_{t+1} = z^{*}_{t} + \frac{d}{d t} z^{*}_{t} \Delta t  $. The size of the derivative, and subsequent step, measures the sensitivity of the solution.

To conceptualize the update, note that the Euler forward predictor step is a Newton step on the linear approximation to the future local quadratic model. Let $\mathcal{Q}(z_{*} \mid r, \eta, \vartheta)$ be the local quadratic model to $\mathcal{G}(z_{*} \mid r, \eta, \vartheta)$. If $\nabla_{z} \mathcal{G} (z_{*} \mid r, \eta, \vartheta)=0$, then the right hand side of equation \eqref{ODE} is:
\begin{equation}
\begin{aligned}
& -\nabla_{z}\left(\mathcal{Q}(z_{*} \mid \psi)+ \nabla_{\psi} \mathcal{Q}(z_* \mid \psi) \cdot \frac{d}{dt} \psi(t) \Delta t \right)
\end{aligned}
\end{equation}
Hence each Euler update acts   as a Newton step on the linear approximation to the next local quadratic model.

The forward Euler step requires  $\frac{d}{dt}z^*_t$. Each derivative is a solution to the linear system \eqref{ODE}. At first glance, finding $\frac{d}{dt} z_*(t)$ ought to be very expensive, since every update step requires solving a new linear system. That said, the Hessian depends continuously on the solution and hyperparameters, as does the right hand side of \eqref{ODE}. Then, subsequent systems are close to identical when the updates are small. Since we repeatedly solve similar linear systems involving $H$, we can reuse past solutions as initial iterates in iterative linear system solvers. Then, each update is computed iteratively starting from an iterate that is close to the true derivative. Iterative solvers are especially well-suited when the forawrd model $A$, and consequently, $H$, is sparse. In practice, we compute a preconditioner (see \ref{sec: pre}) that approximates the inverse  Hessian. Preconditioning speeds convergence. 

The ODE system (\ref{ODE}) holds the gradient of the objective constant, whether or not it is zero, so may accumulate error. We correct by directly optimizing the objective, starting from the predicted solution step. The correction step ensures fidelity to the solution path, thereby evading a more sophisticated ODE update. Future implementations could replace a forward Euler step with an alternative ODE update, could adopt an adaptive time-step that tunes $\Delta t$ according to $\frac{d}{dt}z^*_t$, or could modify the path-tracing ODE to incorporate a gradient descent term.


\subsubsection{The Correction Step}\label{Correct}

The ODE used in prediction is unaware of the underlying optimization problem. Instead, it holds the gradient constant along any solution path. 
Hence we need a correction step to correct the errors introduced by prediction. We use a second-order Newton based corrector, because $z^p_t$ is usually close to the exact solution $z^*_t$, and prediction uses a Newton step, so must already be implemented efficiently. If prediction is sufficiently cheap, so is second-order correction. 

When close to the optimizer, second-order correction is much more accurate than cheaper first-order correction since second-order methods converge quadratically. Faster convergence rates ensure more accurate correction when accrued errors are small. The advantages of second order correction are illustrated in Section \ref{sec: deconvolution}. There, we compare second-order correction with a first-order IAS corrector, which converges linearly on the support of the true solution \cite{calvetti2019hierachical}.Newton correction outperforms IAS correction since it converges quadratically near the minimizer.

The Newton direction $\delta z^{p}$ is the solution to:
\begin{equation}\label{Newton}
H\left(z^{p} \mid \psi \right) \delta z^{p}=-\nabla_z \mathcal{G}\left(z^{p} \mid \psi\right)
\end{equation}
Therefore, correction and prediction are both linear systems of the same form. Accordingly, we adopt the same methods used to speed prediction to speed correction. We precondition and use a Krylov iterative solver \cite{knoll2004jacobian} with warm starts.
We backtrack to prevent over-stepping. In the computed examples, we adopt Armijo's sufficient decrease backtracking strategy.
Since Newton is prone to over-stepping in poorly scaled problems, we select large $\alpha$ values that keep initial steps small.

\subsubsection{Solving the Linear Systems}


Even after the transforms introduced in \ref{sec: hessian}, the Hessian is often still ill-conditioned, especially near sparse estimates. We propose a preconditioning strategy in section \ref{sec: pre} which both speeds convergence and avoids amplifying errors. The preconditioner is carefully designed to take advantage of the sparse, near low-rank nature of the Hessian. In practice, it is both cheap to compute, and significantly boosts the performance of the iterative solvers. After preconditioning, we apply CGLS initialized at the previous estimate.

\subsubsection{The Algorithm}

The full Predictor-Corrector algorithm is summarized in Algorithm \ref{alg:PC}.

  \begin{algorithm}
  \caption{Predictor-Corrector Algorithm}
  \begin{algorithmic}[]    
    \STATE {\bf Input:}
Noisy data $b \in \mathbb{R}^{m}$,
linear forward operator $A \in \mathbb{R}^{m \times n}$, prior hyperparameters path from $\left(r(0), \beta(0), \vartheta(0)\right)$ to $\left(r(T), \beta(T), \vartheta(T)\right)$, maximum iterations of Krylov method $M$, tolerance $\epsilon$
     \STATE {\bf Output:}                                            
\text {estimated signal and variance path} $x_t^{*}, \theta_t^{*} \in \mathbb{R}^{n}$

    \STATE {\bf Initialize:} set $z^*_0 = (\theta^*_0,x^*_0)$ = minimizer with hyperparameter   $\left(r(0), \beta(0), \vartheta(0)\right)$, $\frac{d}{d t} z^{*}_{0}=0$, $\delta z_0^p = 0$
    \STATE \hspace*{0.02in}{\bf for:} t = 1,2,...N \hspace*{0.02in}{\bf do:}
    
    \STATE \qquad {\bf if:}
    iterations $> M$
    \STATE \qquad \qquad Build the preconditioner ({\ref{precondition}})
   \STATE \qquad {\bf endif:}
    \STATE \qquad Solve $\frac{d}{d t} z^{*}_{t}$ (ODE (\ref{ODE})) using preconditioned CGLS initialized with $\frac{d}{d t} z^{*}_{t-1}$.
   \STATE \qquad Predict: $ z^{p}_{t+1} = z^{*}_{t} + \Delta t \frac{d}{d t} z^{*}_{t}$.
    \STATE \qquad {\bf if:}
    iterations $> M$
    \STATE \qquad \qquad Build the preconditioner ({\ref{precondition}})
    \STATE \qquad {\bf endif:}
    \STATE\qquad Solve $\delta z_t^p$  (equation (\ref{Newton})) using preconditioned CGLS initialized with $\delta z^p_{t-1}$
    \STATE \qquad {\bf while:}
    $ ||\delta z^p_t|| > \epsilon$
    \STATE \qquad \qquad Correct: $z^*_t = z^p_t + \alpha \delta z^p_t$ with backtracking globalization
    \STATE \qquad {\bf endwhile:}
    \STATE \hspace*{0.02in}{\bf end for}
  \end{algorithmic}
  \label{alg:PC}
\end{algorithm}

\subsection{Preconditioning}\label{sec: pre}

Recall that, after changing variables, the scaled Hessian can be broken into a fidelity term, $H_A$, associated with the forward model, and a penalty term, $H_P$, associated with the effective regularizer (see equation \eqref{eqn: fidelity and penalty}). The penalty term is tridiagonal and explicitly invertible (see equations \eqref{eqn: penalty decomp} to \eqref{eqn: R inverse}). The scaled fidelity term $D_{\theta}^{1 / 2} A^{\top} A D_{\theta}^{1 / 2}$ is the same matrix studied in \cite{calvetti2020sparse}. When $x$ is sparse, the fidelity 
term is close to low rank. Therefore, $\widetilde{H}$ is, a, approximately, a low rank perturbation of an explicitly invertible matrix. Accordingly, we use the Woodbury matrix identity \cite{sherman1950adjustment} to build an approximate inverse preconditioner.

If we can find a low-rank approximation, $\widetilde{H}_{A}=U U^{\top}+E$ where $U$ is $2 n \times r$, $r \ll n$ is the effective rank, and $E$ is the error in the low rank approximation, then
the Woodbury inverse \cite{sherman1950adjustment} is:
\begin{equation}\label{precondition}
\widetilde{H}^{-1} \simeq \widetilde{H}_{P}^{-1}-\widetilde{H}_{P}^{-1} U\left(I_{r \times r}+U^{\top} \widetilde{H}_{P}^{-1} U\right)^{-1} U^{\top} \widetilde{H}_{P}^{-1}
\end{equation}
The inner term is $r \times r$ so is cheap to evaluate if $r$ is small.
If an iterative solver is used then all other products need not be performed explicitly.

To build a sparse low rank approximation of the fidelity term, $D_{\theta}^{1 / 2} A^{\top} A D_{\theta}^{1 / 2}$, pick an accuracy tolerance $\epsilon<1$. Then, set all columns with column sum less than $\frac{1}{2} \epsilon$ equal to zero. Near a sparse solution, most variances $\theta$ are small, so usually only a small subset of columns remain. Then, build a low rank approximation to these columns (for example, use a truncated SVD, truncated at the rank such that the $l_{\infty}$ error in the approximation is less than $\frac{1}{2} \epsilon$ ). The low rank approximation is sparse since most columns/rows are set to zero and is cheap since it only requires a partial SVD of a small subset of the original matrix.     

 By construction, the error $\|E\|_{\infty}$ is less than $\frac{1}{2} \epsilon$ + $\frac{1}{2} \epsilon = \epsilon$. Then, by the Gershgorin disk theorem, the largest eigenvalue of $E=D_{\theta}^{-1 / 2} A^{\top} A D_{\theta}^{-1 / 2}-U U^{\top}$ is smaller in magnitude than $\|E\|_{\infty}=\max \left\{\sum_{j}\left|e_{i j}\right|\right\} \leq \epsilon$. So, to ensure that the spectral radius of $E$ is less than 1, we only need to pick $\epsilon<1$.  
 Iterative methods for solving the linear system will converge quickly when the spectral radius of $E$ is less than 1 since iterative methods implicitly build a polynomial that approximates the power series expansion in $E$ to the inverse of  $\tilde{H}$.
 
The resulting preconditioner, \eqref{precondition}, aims to speed iterative methods. Thus, the preconditioner is ineffective if iterative solvers require many iterations. We monitor the iteration count and recompute the preconditioner when there are too many. The computed examples in section \ref{sec: example} are of intermediate scale, so the preconditioner can be computed efficiently. In those examples, we recompute the preconditioner after every change in the hyperparameters since it is not excessively expensive. 

A user could also update the preconditioner perturbatively since it should change continuously. Let $P_{\text{old}}$ denote the current preconditioner. Then, the new preconditioner $P_{\text{new}}$ can be updated from the old precondtioner $P_{\text{old}}$ by:
$$P_{\text {new}}=P_{\text{old}}-P_{\text{old}}\left(H_{\text{new}}-H_{\text{old}}\right) P_{\text{old}}$$ 
In practice this update does not work well since small changes in the Hessian can lead to large changes in its inverse, so the preconditioner may have to change rapidly. Similar naive updates sacrifice the low-rank complexity of the Woodbury approach, and, in the computed examples, the occasional recompute policy was fast enough. Future work could use perturbative low-rank approximation to update the preconditioner while preserving its low-rank complexity.

\subsection{IAS with Newton acceleration} \label{sec: acceleration}

Preconditioning enables fast, stable Newton steps. Accordingly, it may also accelerate MAP estimation.

The standard MAP estimation method, given fixed hyperparameters, is IAS. The IAS method converges linearly on the support of the estimated solution since it is a reweighted least squares algorithm based on coordinate descent \cite{calvetti2019hierachical} . To increase $\theta$ it must increase $x$, but to increase $x$ it must increase $\theta$. Any  such alternating algorithm will make slow progress. Consequently, 
IAS converges slowly once compressible, often producing an inference that systematically underestimates the true signal (c.f.~\cite{agrawal2021variational}). Compared to IAS, inexact IAS \cite{calvetti2020sparse} is more efficient but is inexact.

To speed convergence, we propose Newton acceleration: start with a series of IAS steps then switch to Newton. Unlike IAS, the Newton updates both $x$ and $\theta$ simultaneously, so does not require many alternating iterations. IAS converges rapidly off the support; Newton converges rapidly on the support.

\section{Demonstration}\label{sec: example}

In this section, we present two computed examples that illustrate the efficacy and utility of path-following via prediction-correction. 

\subsection{Deconvolution} \label{sec: deconvolution}

First, we consider the standard $1 \mathrm{D}$ deconvolution test problem used in \cite{calvetti2019hierachical,agrawal2021variational,calvetti2020sparse,kim2022hierarchical}.

Let $f:[0,1] \rightarrow \mathbb{R}$ be a piecewise constant function with $f(0)=0$. We observe $y$:
$$
y_{j}=\int_{0}^{1} A\left(s_{j}-t\right) f(t) d t+\epsilon_{j}, \quad 1 \leq j \leq n, \quad A(t)=\left(\frac{J_{1}(\kappa|t|)}{\kappa|t|}\right)^{2}
$$
where $J_{1}$ is the Bessel function of the first kind, $\kappa$ is a scalar controlling the width of the kernel, and $\epsilon$ is Gaussian noise. We set $\kappa=40$ and $s_{j}=(4+j) / 100$. The convolution can be discretized:
$$
y=A v+\eta, \quad A_{j k}=w_{k} A\left(s_{j}-t_{k}\right), \quad \eta \sim \mathcal{N}\left(0, \gamma^{2} I_{n}\right),
$$
where $v \in \mathbb{R}^{d}$ has components $v_{k}=f\left(t_{k}\right)$ with $t_{k}=(k-1) /(n-1)$, and the $w_{k}$ are quadrature weights for discretization of the integral. We set the standard deviation $\gamma$ to  $1 \%$ the max-norm of the noiseless signal.

Let $y_{j}=x_{j}-x_{j-1}$ with $x_{0}=0$. Since $x$ is piecewise constant, $y$ is sparse. 

$$
y=L x, \quad L = \left[\begin{array}{cccc}
1 & 0 & \ldots & 0 \\
-1 & 1 & \ldots & 0 \\
& & \ddots & \\
0 & \ldots & -1 & 1
\end{array}\right] \in \mathbb{R}^{n \times n}
$$

Then $x=L^{-1} y$ with
$$
L^{-1}=\left[\begin{array}{cccc}
1 & 0 & \ldots & 0 \\
1 & 1 & \ldots & 0 \\
\vdots & & \ddots & \\
1 & \ldots & 1 & 1
\end{array}\right] \in \mathbb{R}^{n \times n}
$$

The inverse problem is to estimate the assumed sparse vector $y$ from the noisy data vector $b$, given the forward operator $AL^{-1}$.
$$
b=A L^{-1} y+\epsilon, \quad \epsilon \sim \mathcal{N}\left(0, \sigma^{2} I\right), \quad a_{j k}=w_{k} A\left(s_{j}-t_{k}\right)
$$

First we adopt the IAS algorithm for convex setting hyperparameters $(r,\eta,\vartheta)$ $=(1.5,10^{-5},10^{-6})$. The left panel of Figure \ref{fig:tias} shows the recovered signal initialized with all ones. The solution is insufficiently sparse. 

Next, we solve the same problem in the non-convex regime where $r = 0.5$. In the non-convex regime, the solution is sensitive to initialization. We run IAS 1000 times with randomly sampled initialization. The right panel of Figure \ref{fig:tias} displays an envelope of resulting solutions. The shaded area between the minimum and maximum of recovered estimate when randomly initialized. While IAS recovers sparse local minima, it fails to uniquely select an accurate solution. In particular, it consistently breaks the fourth jump into two separate jumps, so systematically fails to recover the support of the true signal. We will show that the path-following algorithm, when initialized in the convex region, is more stable and accurate.

\begin{figure}
\begin{center}
   \begin{minipage}{0.5\textwidth}
     \centering
     \includegraphics[width=.8\linewidth]{ 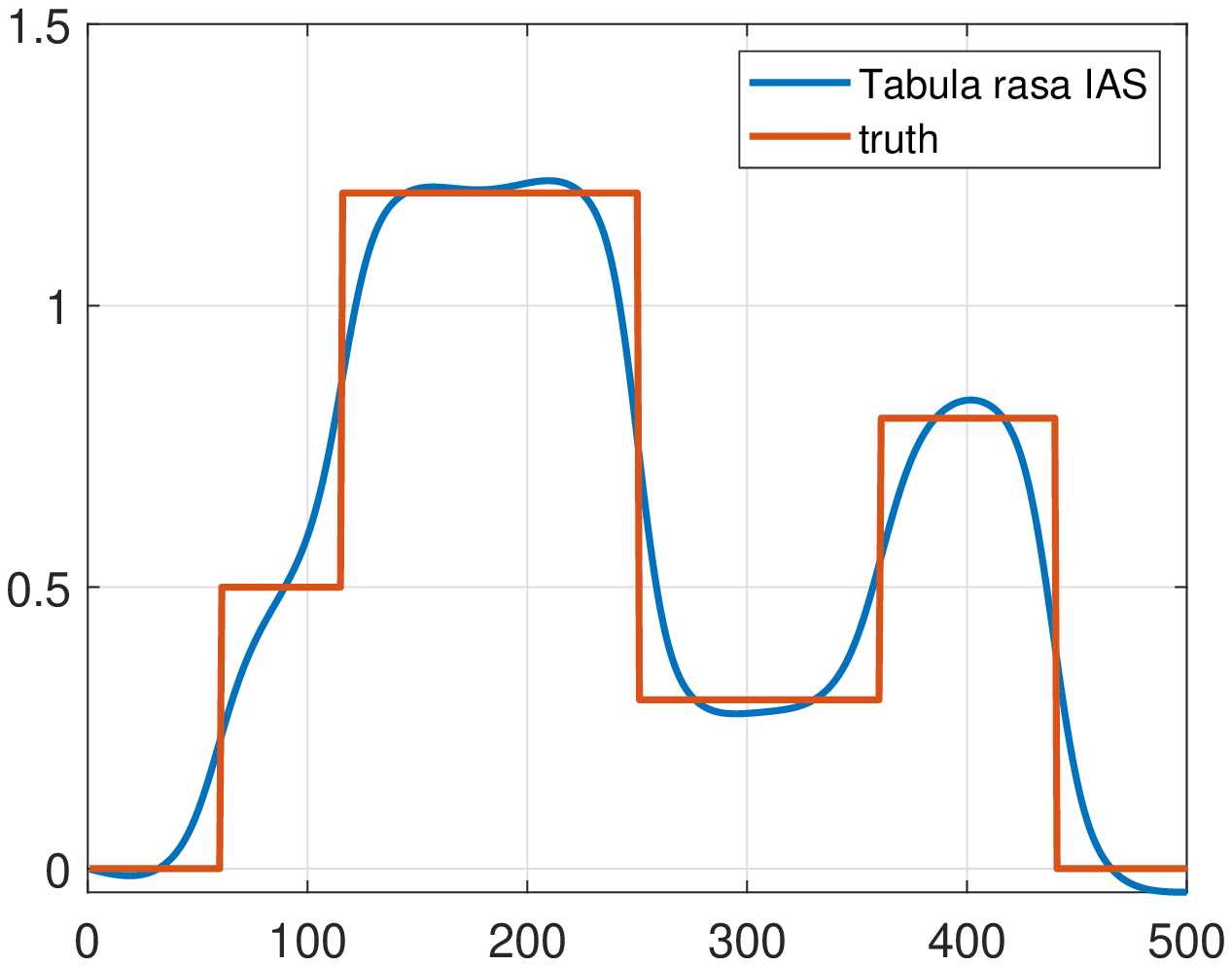}
   \end{minipage}\hfill
   \begin{minipage}{0.5\textwidth}
     \centering
     \includegraphics[width=.95\linewidth]{ 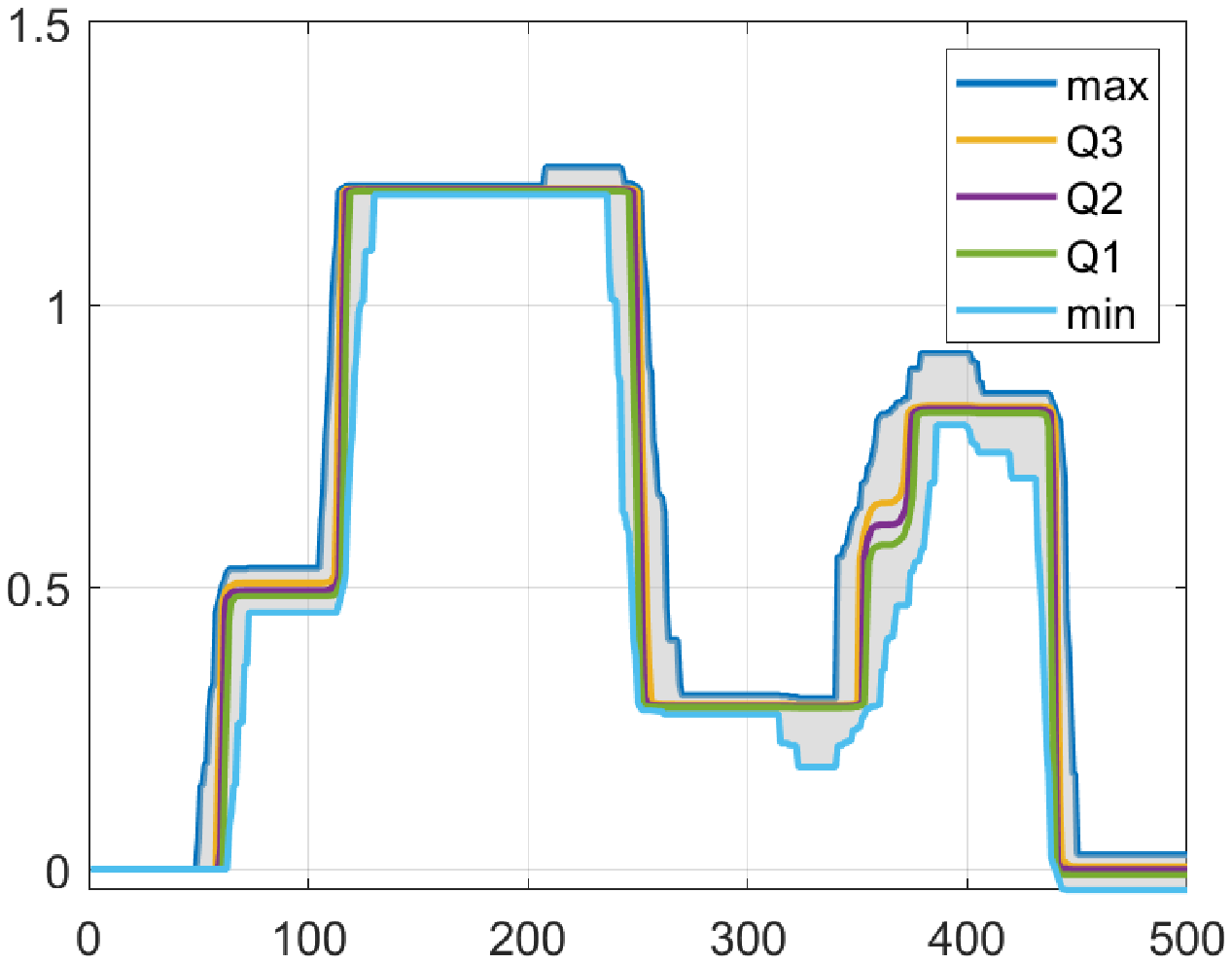}
   \end{minipage}
   \end{center}
\caption{ \textit{Left: } true signal (orange) and the recovered signal (blue) using IAS initialized with all ones and hyperparameters $(r,\eta,\vartheta) = (1.5,10^{-5},10^{-6})$. Note, the unique solution in the convex regime is not sharp. \textit{Right:} envelope plot of the recovered signal using IAS with random initialization and hyperparameters $(r,\eta,\vartheta) = (0.5,10^{-5},10^{-6})$ . Note, in the nonconvex regime the solution is sharp, but depends on initial conditions.}    
\label{fig:tias}
\end{figure}

To test the path-following approach, we vary the hyperparameters along the line from a convex setting, $(r, \eta, \vartheta) = (1.5,1.5,10^{-5})$, to the desired non-convex setting $(r, \eta, \vartheta) = (0.5,10^{-5},10^{-6})$. For simplicity, we adopt 60 equidistant time points along the path. We save more sophisticated, adaptive approaches for future work.

At initialization, we test IAS with Newton acceleration. Figure \ref{fig:1Dcost} shows the results. Newton acceleration rapidly improves convergence to the MAP estimator. After 3 initial IAS steps, Newton converges in 2 additional steps. Pure IAS requires 7 more iterations to match the first Newton step, and has not converged after 10 steps.

\begin{figure}
\begin{center}
   \begin{minipage}{0.5\textwidth}
     \centering
     \includegraphics[width=.8\linewidth]{ 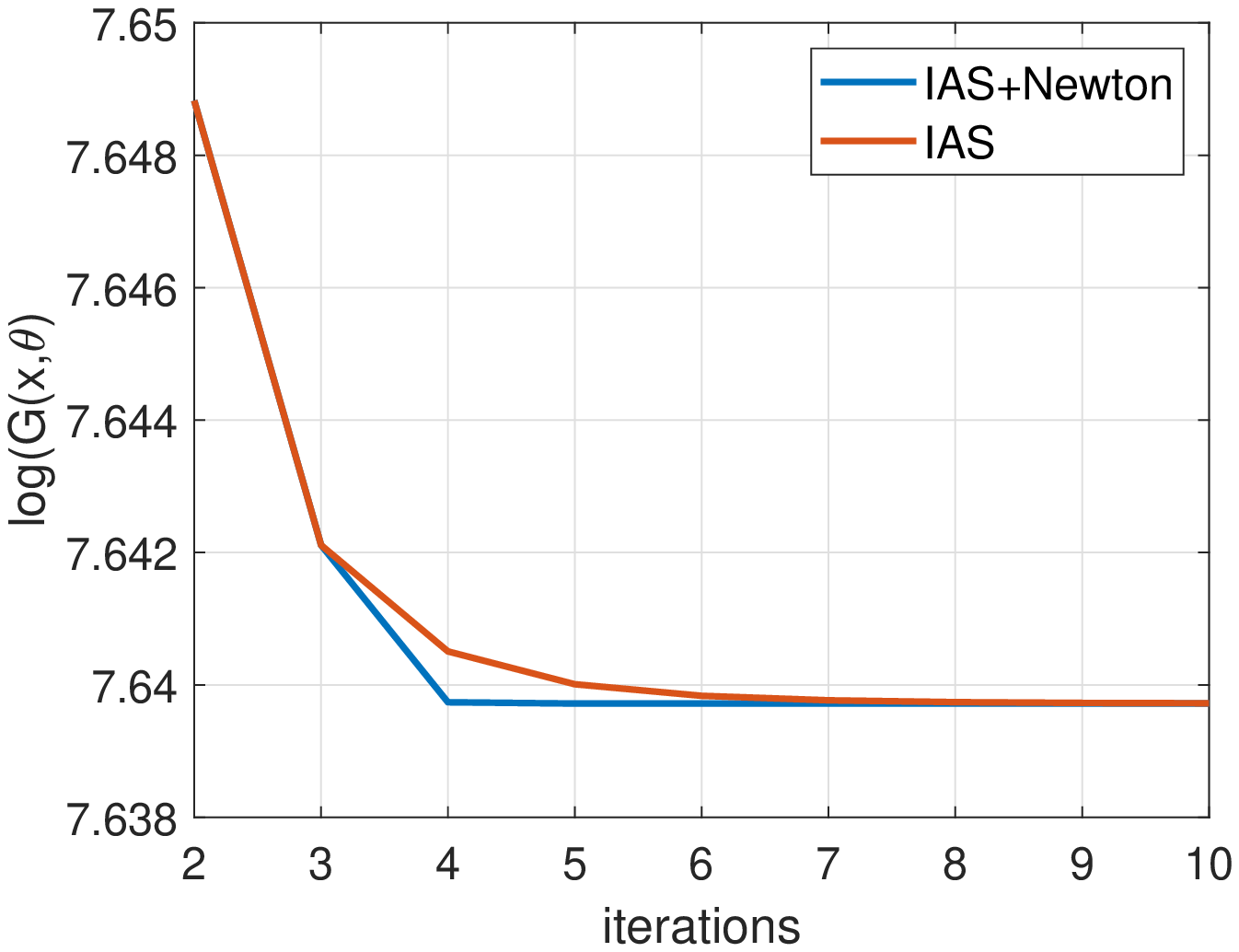}
   \end{minipage}\hfill
   \begin{minipage}{0.5\textwidth}
     \centering
     \includegraphics[width=.76\linewidth]{ 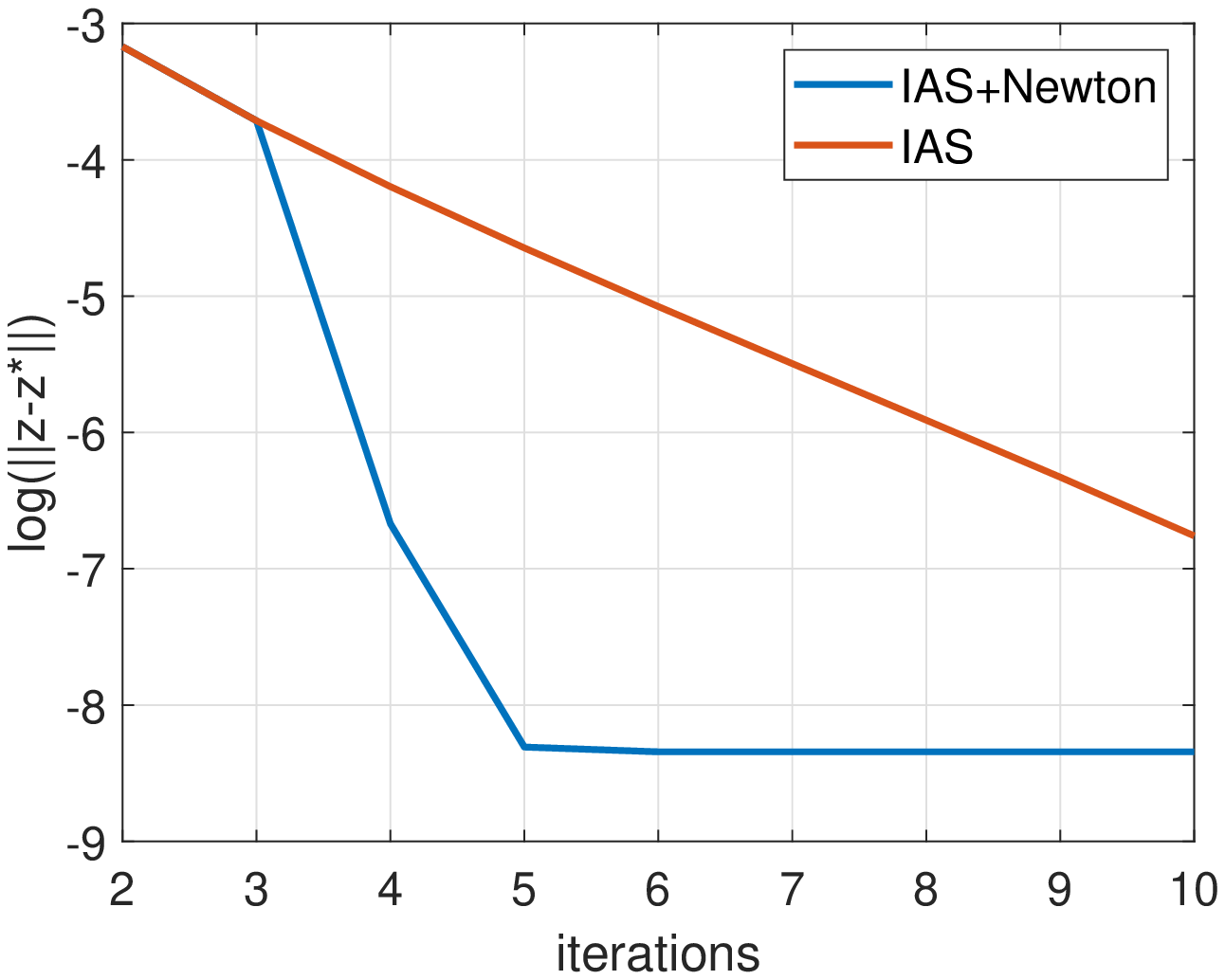}
   \end{minipage}
   \end{center}
\caption{Decay of the objective function value (left) and 2-norm error between the iterate $z_t$ and the minmizer $z^*$  (right) over repeated iterations with fixed hyperparameters. The orange line represents IAS algorithm. The blue line shows Newton accelerated IAS. Note the rapid convergence to the minimizer achieved by Newton acceleration.}
\label{fig:1Dcost}
\end{figure}

Next, we compare the solution path obtained by the following methods:

\begin{enumerate}
\item \textbf{Path-following IAS:} 1 IAS iteration at each time step.
\item \textbf{Path-following Inexact IAS:} 1 Inexact IAS iteration at each time step.
\item \textbf{Predictor-IAS corrector:} Iteratively alternate between ODE based predictor and 1 IAS iteration corrector.
\item \textbf{Predictor-Newton corrector:} Iteratively alternate between ODE based predictor step and 1 Newton iteration corrector.
\end{enumerate}

All the methods run 1 iteration in each time step and use warm starts to reuse past solutions. In practice, one could adopt more correction steps. In these examples, one Newton correction step was enough. The fact that only one correction step was required indicates the efficacy of prediction. Results using three correction steps were only marginally more accurate.

Figure \ref{fig:1Dscreenshots} shows the sparse solution $y(t)$, reconstructed signal $x(t)$, and variance $\theta(t)$ corresponding to 3 different points on the hyperparameters path recovered by these methods. Entries on the support become larger while entries off the support become smaller along the solution path. Thus, the prior promotes stronger sparsity as $r$ decreases.

The Predictor-Newton algorithm is the most accurate of all the methods tested. It correctly identifies the 5 non-zero entries of $y$ once $r(t)$ approaches 0.5. Predictor-IAS is the second most accurate method, but fails to identify the third non-zero entry correctly. Instead, it splits the third jump into 2 adjacent jumps. All methods recover signals close to the true signal at the end of the hyperparameter path, however, the scatter plots shown in the left panel of Figure \ref{fig:1Dscreenshots} indicate that all algorithms except Predictor-Newton blur the support and do not recover properly compressible signals.

The variance plots (right panel of Figure \ref{fig:1Dscreenshots}) also demonstrate that Predictor-Newton method accurately detects the support, and allows larger variances $\theta$ on the support than the other methods. The fact that Predictor-Newton produces sharper, taller spikes in the variance plot indicates that it allows sparser solutions.

Figure \ref{fig:min} displays the minimum value of the objective obtained by each method. The prediction-correction methods are significantly more accurate than path-following IAS. We show the value of the objective before and after correction to illustrate the effect of correction. The correction curves fall below their corresponding prediction curves. The change in objective is usually small, indicating that the predictor steps are accurate. Nevertheless, the small correction steps are necessary since the error introduced by the ODE based predictor can accumulate without correction. In trial runs completed without correction, the solution dramatically deviated from the correct path after accumulating a series of small errors. In some cases IAS correction achieves a larger reduction in the objective function than Newton correction, however it is less stable, and is prone to introducing errors, especially near $r = 1$. The middle and right hand panels of Figure \ref{fig:min} illustrate the accuracy of the solution paths by comparison to a gold-standard solution.

\begin{figure}
\begin{center}
	\begin{minipage}{0.3\linewidth}
		\centering
		\includegraphics[width=.99\linewidth]{ 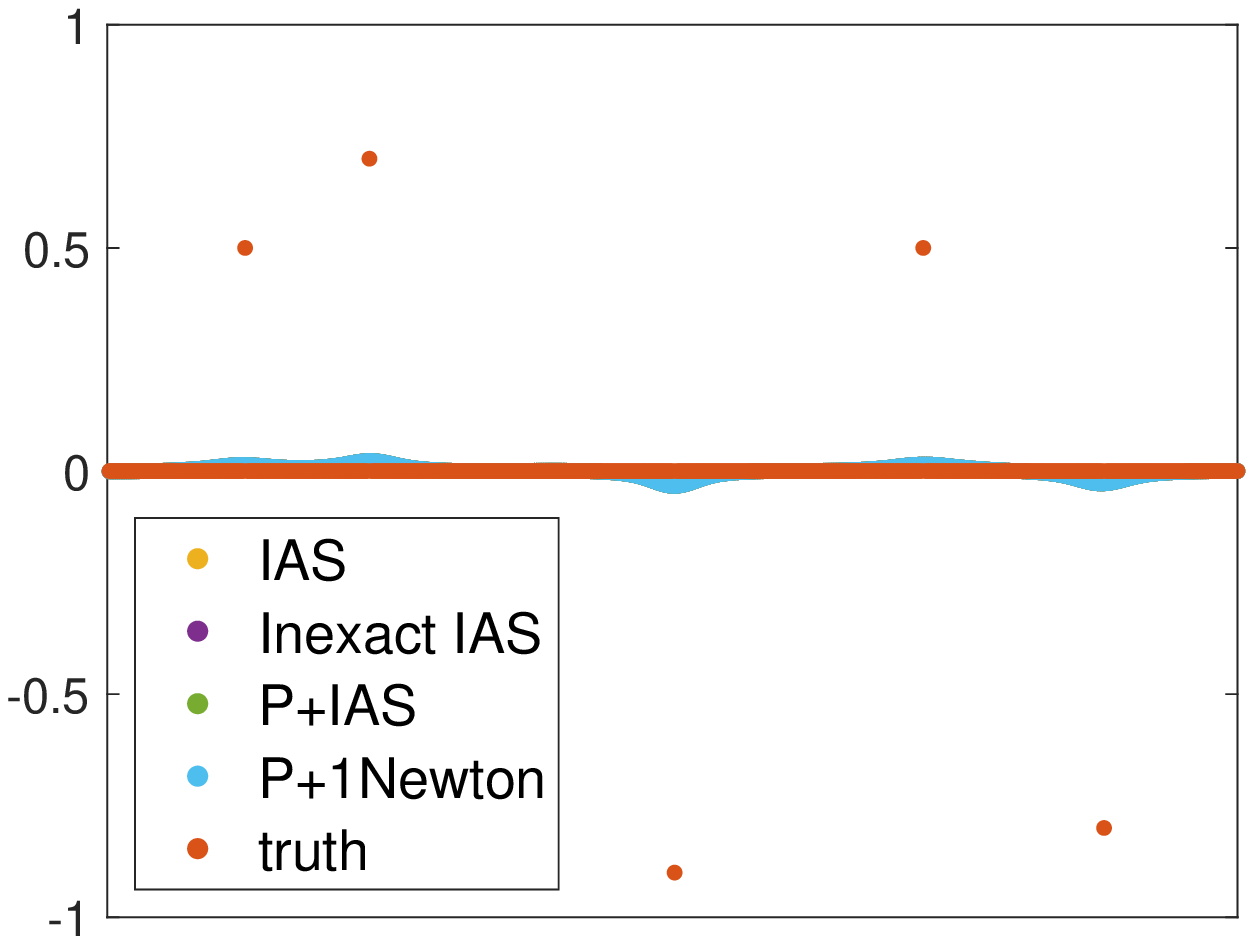}
	\end{minipage}
	\hspace{0.01cm}
	\begin{minipage}{0.3\linewidth}
		\centering
		\includegraphics[width=.98\linewidth]{ 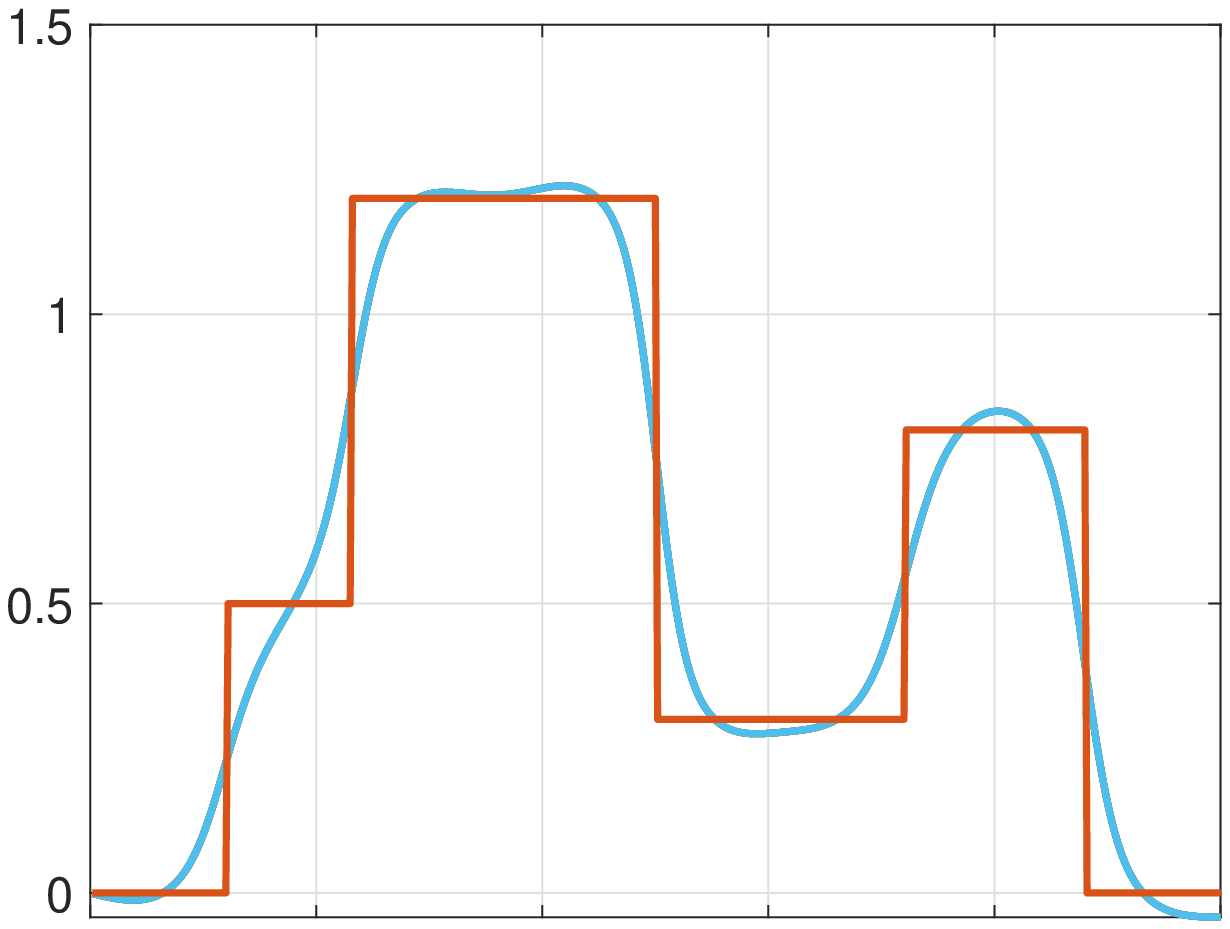}
	\end{minipage}
	\hspace{0.01cm}
	\begin{minipage}{0.3\linewidth}
		\centering
		\includegraphics[width=.95\linewidth]{ 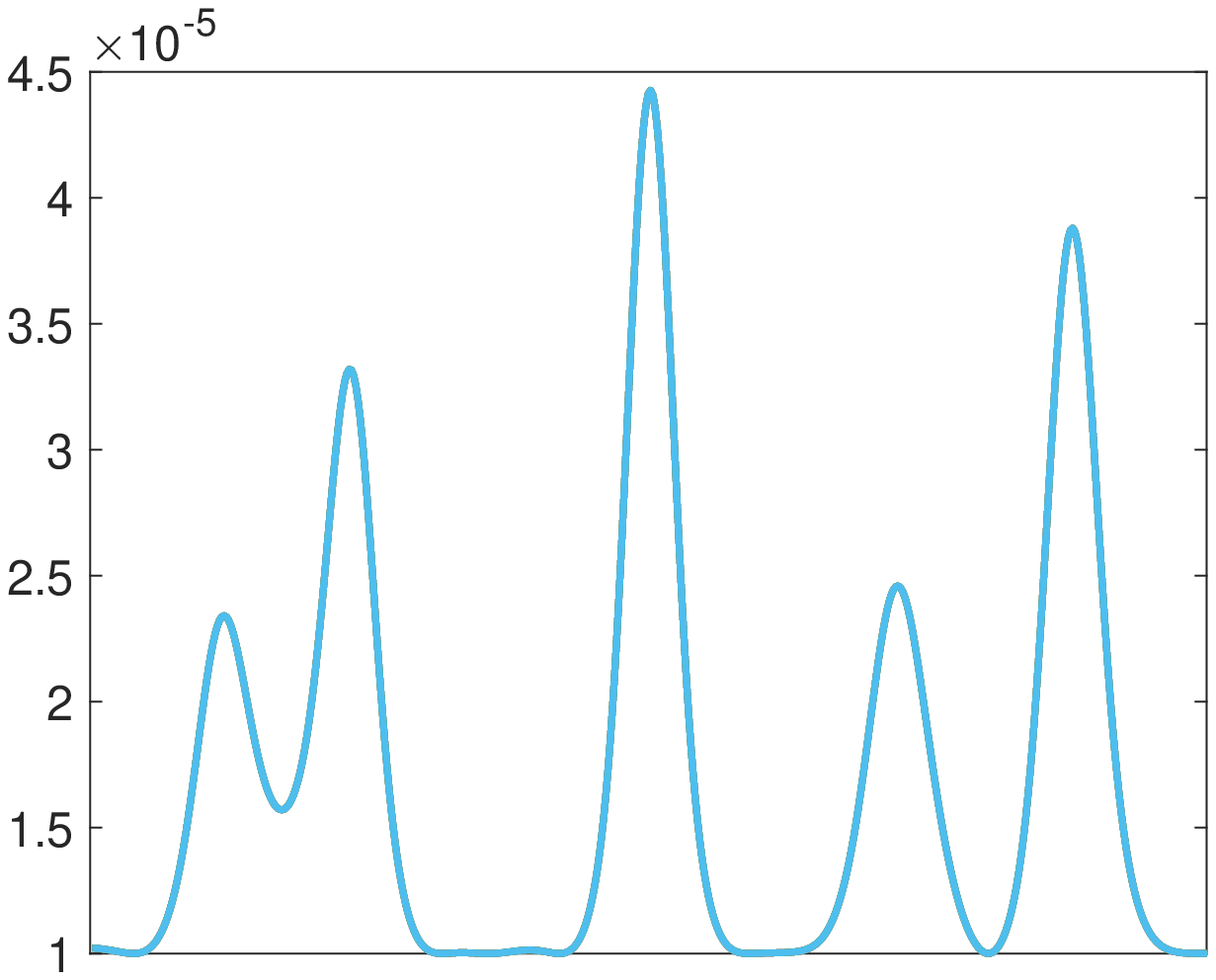}
	\end{minipage}
	\begin{minipage}{0.3\linewidth}
		\centering
		\includegraphics[width=.98\linewidth]{ 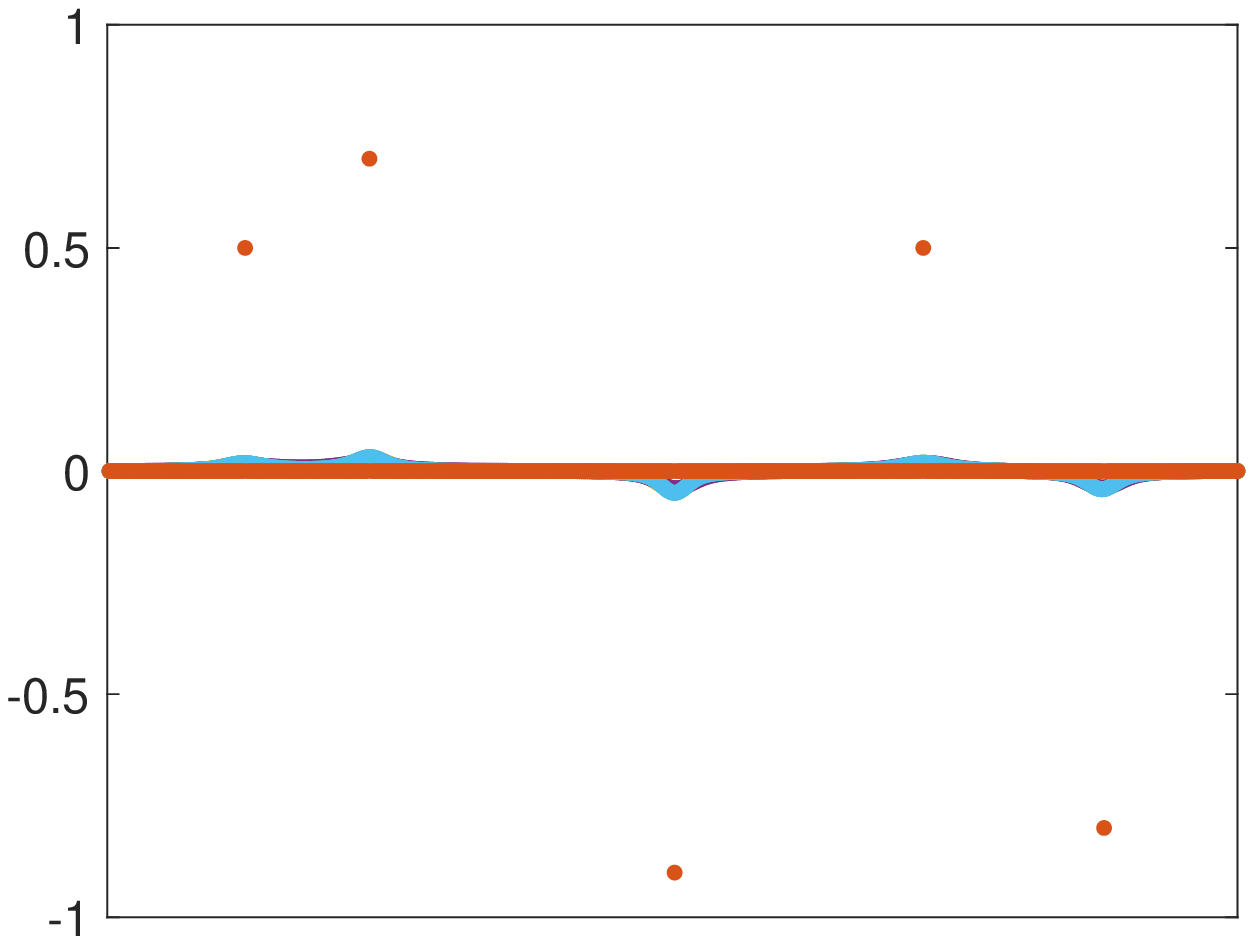}
	\end{minipage}
	\hspace{0.02cm}
	\begin{minipage}{0.3\linewidth}
		\centering
		\includegraphics[width=1\linewidth]{ 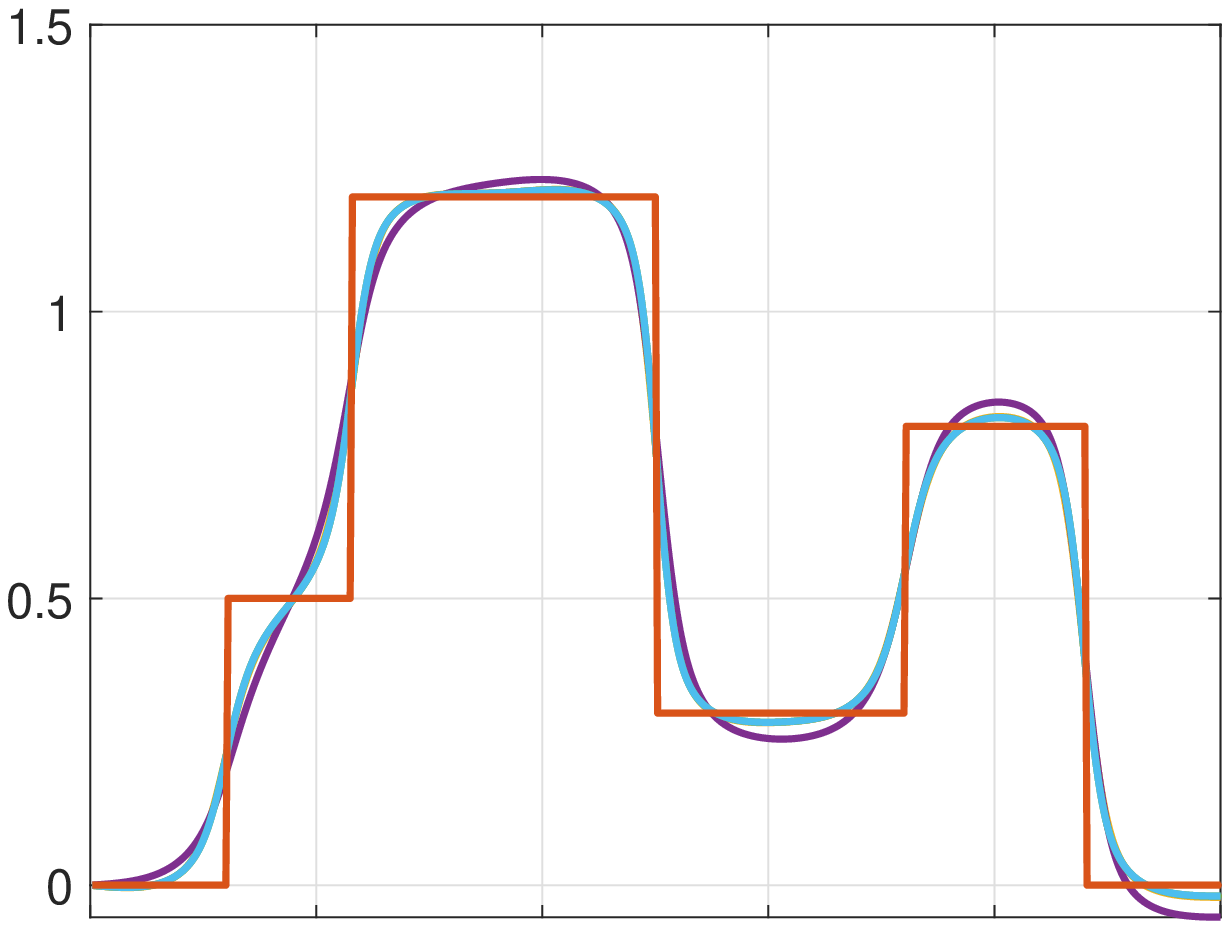}
	\end{minipage}
	\hspace{0.07cm}
	\begin{minipage}{0.3\linewidth}
		\centering
		\includegraphics[width=0.92\linewidth]{ 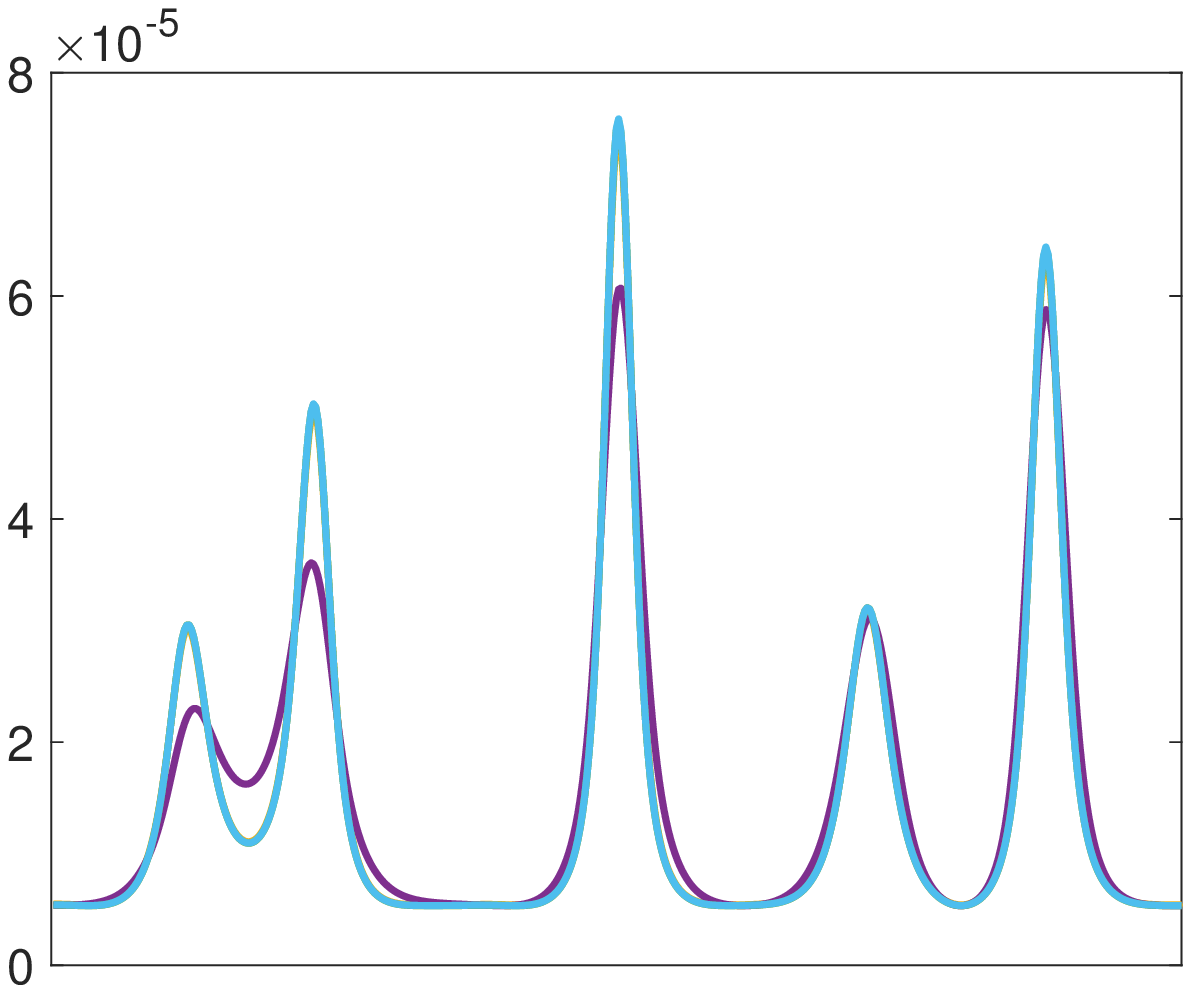}
	\end{minipage}
	\begin{minipage}{0.3\linewidth}
		\centering
		\includegraphics[width=1\linewidth]{ 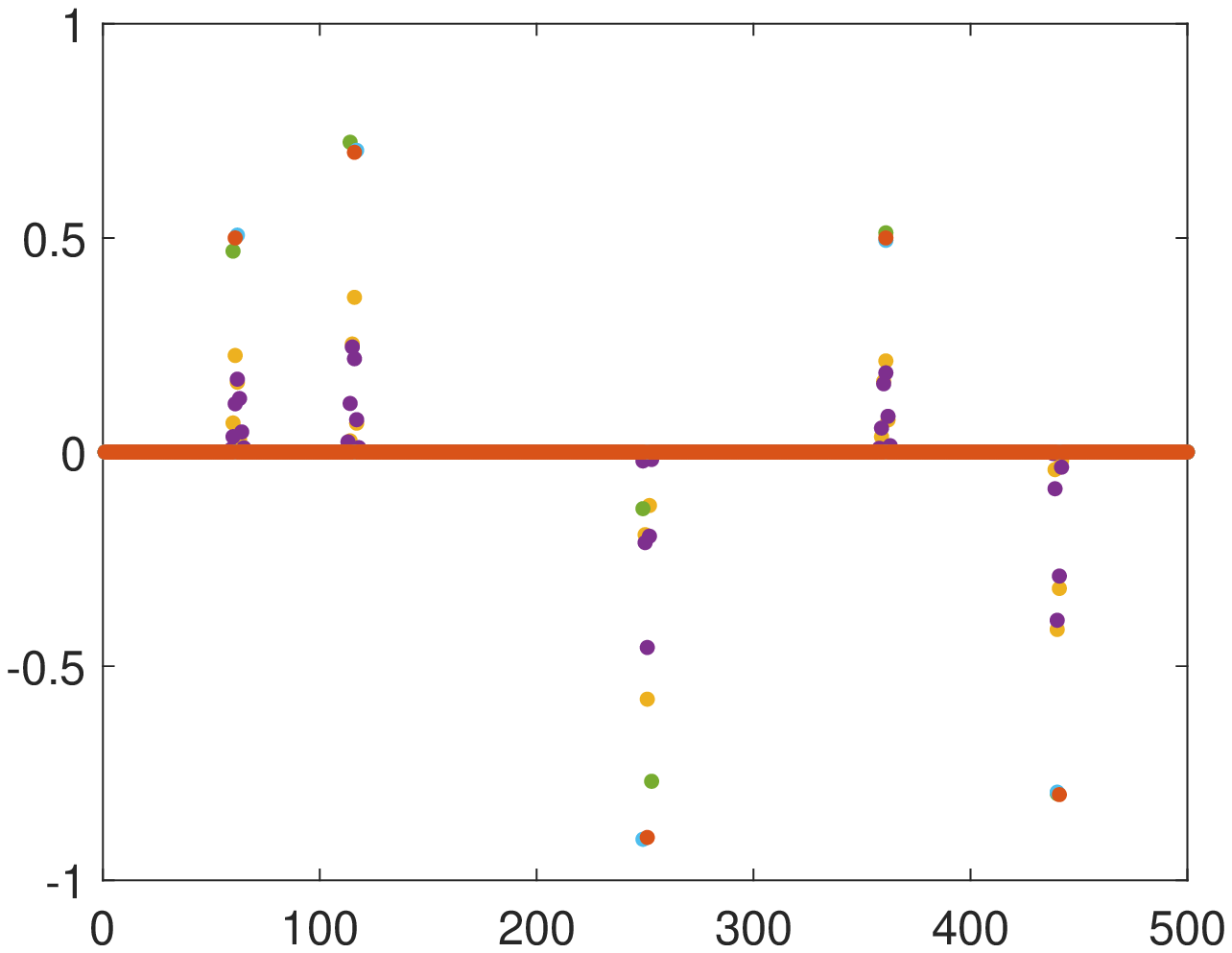}
	\end{minipage}
	\hspace{0.01cm}
	\begin{minipage}{0.3\linewidth}
		\centering
		\includegraphics[width=1\linewidth]{ 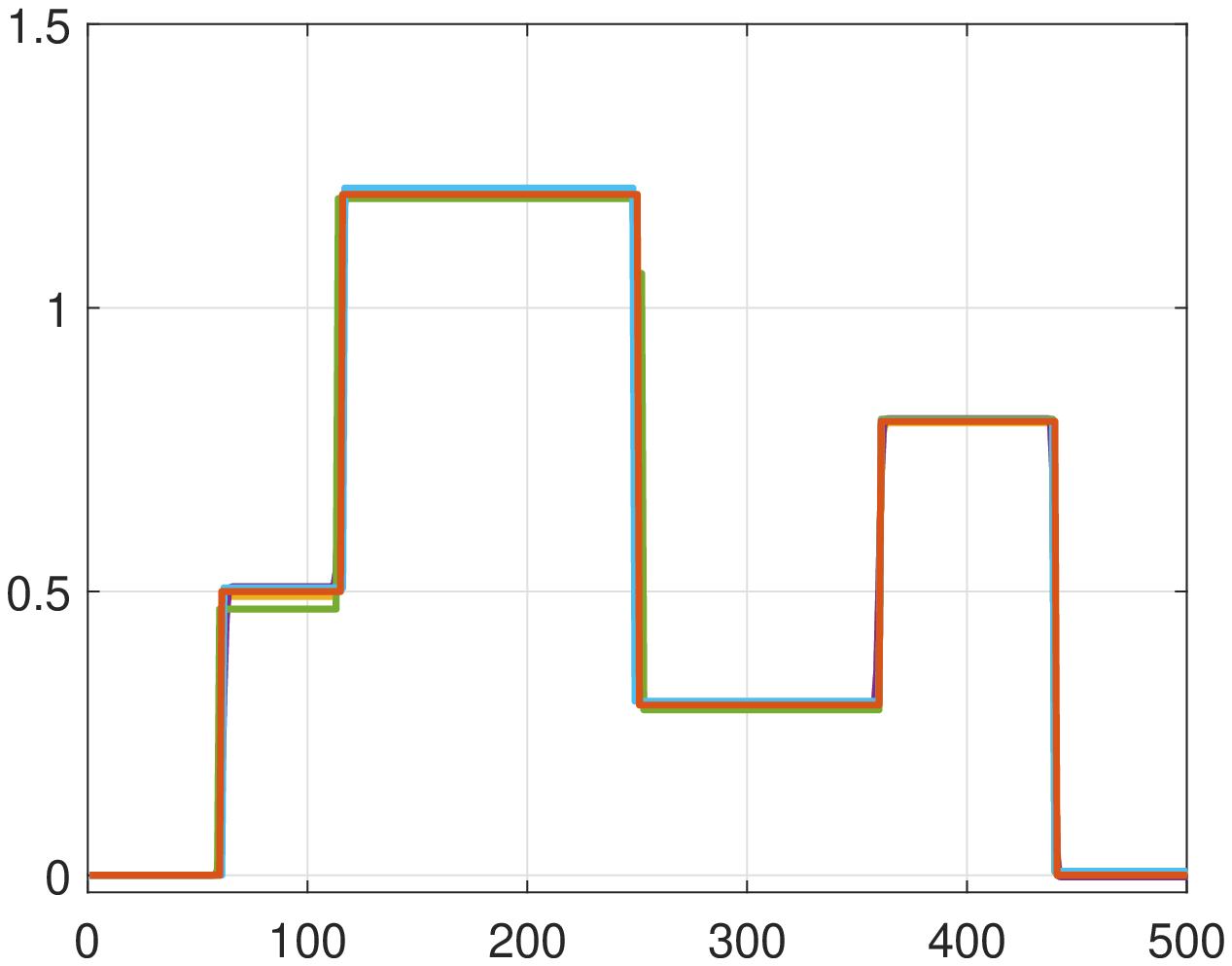}
	\end{minipage}
	\begin{minipage}{0.3\linewidth}
		\centering
		\includegraphics[width=1\linewidth]{ 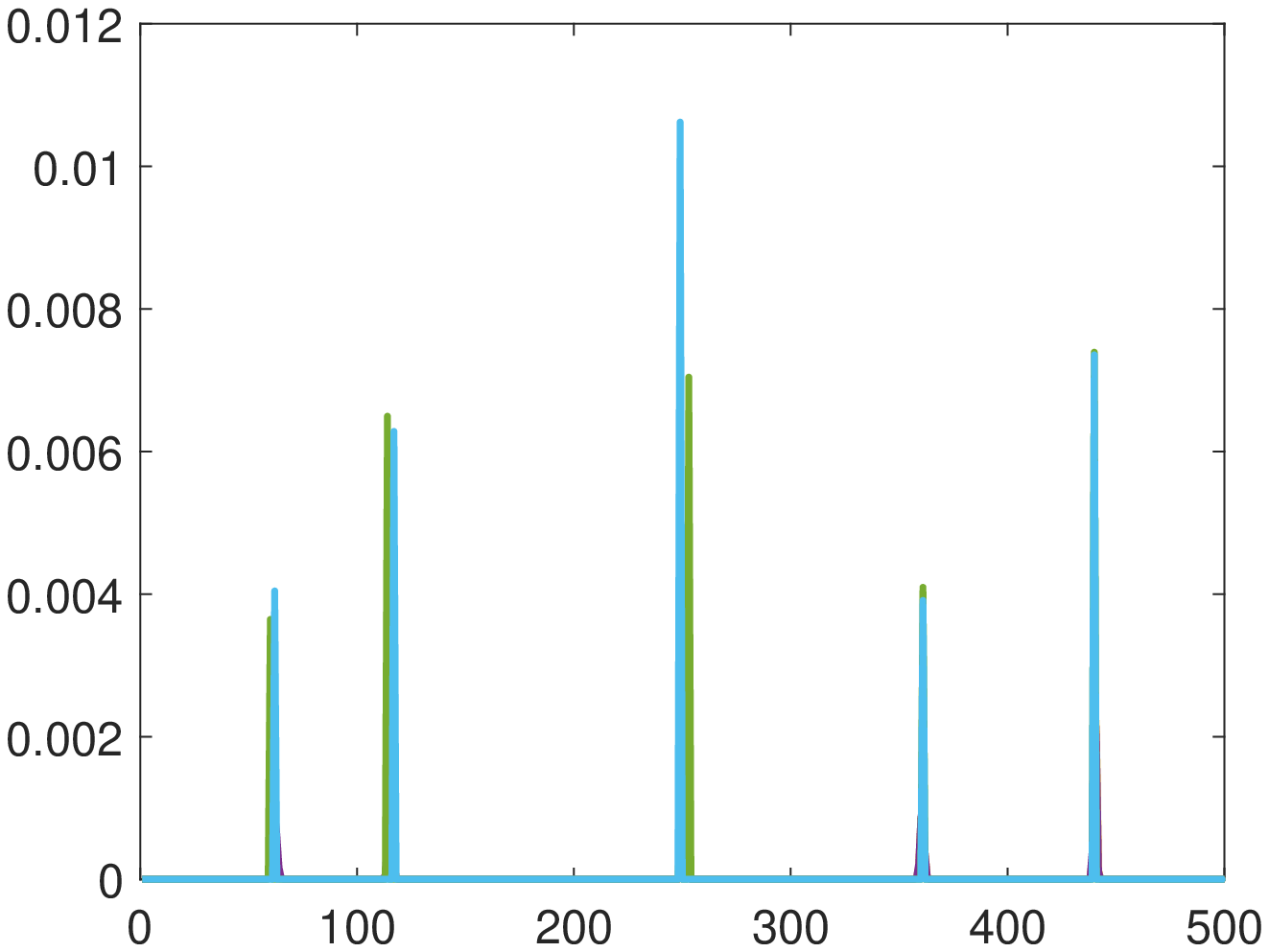}
	\end{minipage}
\end{center}
\caption{solution $y(t)$ (left), signal $x(t)$ (middle), variance $\theta(t)$ (right) corresponding to 3 different points on the
hyperparameters path (top, middle, bottom). In the top row, $r = 1.5$ , $\eta =  1.5$ and $\vartheta = 10^{-5}$, in the middle row, $r = 0.9$ , $\eta =  0.75$ and $\vartheta = 5.5\times 10^{-6}$, in the bottom row, $r = 0.5$ , $\eta =  10^{-5}$ and $\vartheta = 10^{-6}$. Note that the solution produced by Predictor-Newton is the most accurate and sparse.}
\label{fig:1Dscreenshots}
\end{figure}

\begin{figure}
\begin{center}
   \begin{minipage}{0.32\textwidth}
     \centering
     \includegraphics[width=0.93\linewidth]{ 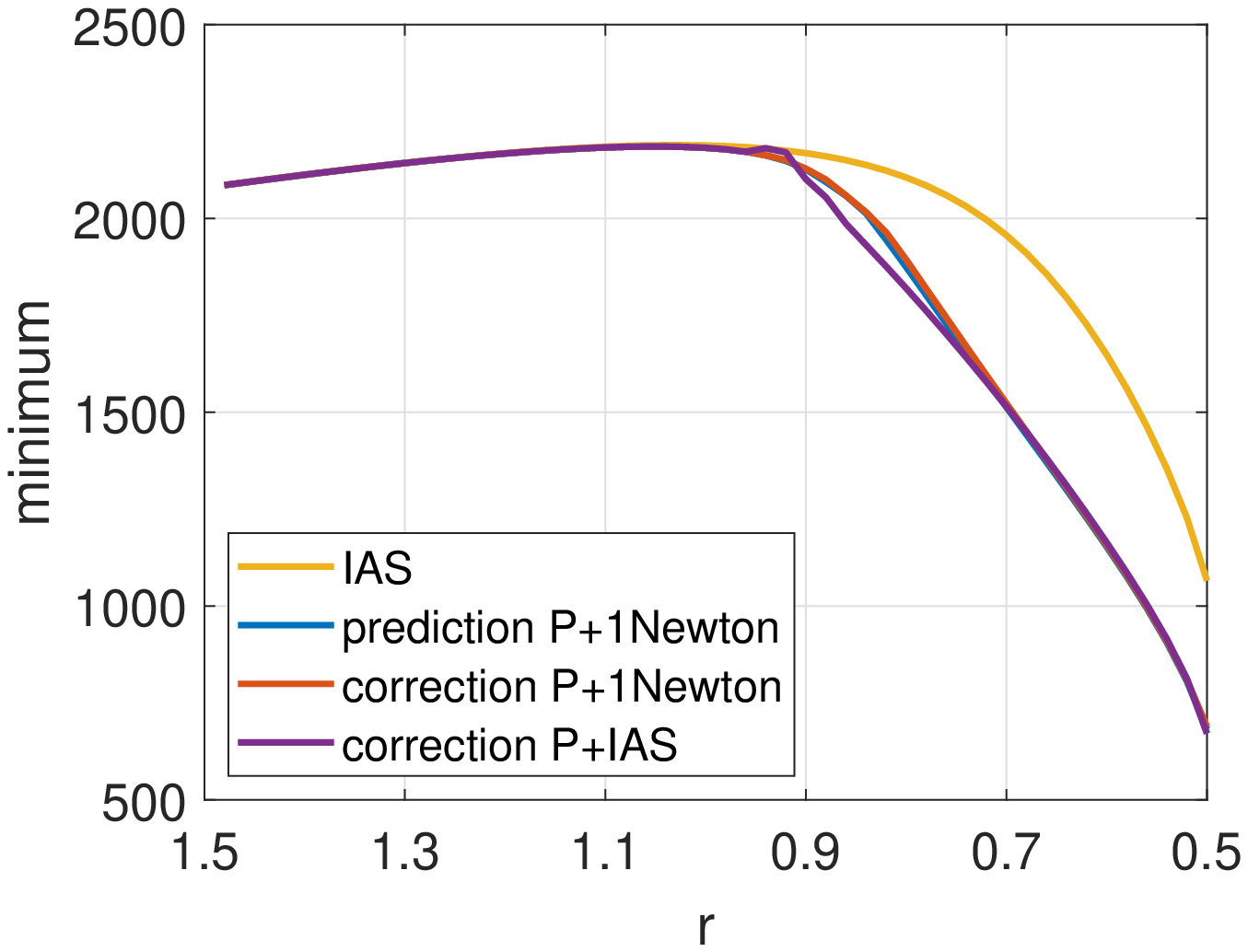}
   \end{minipage}
   \hspace{0.01cm}
   \begin{minipage}{0.32\textwidth}
     \centering
     \includegraphics[width=0.95\linewidth]{ 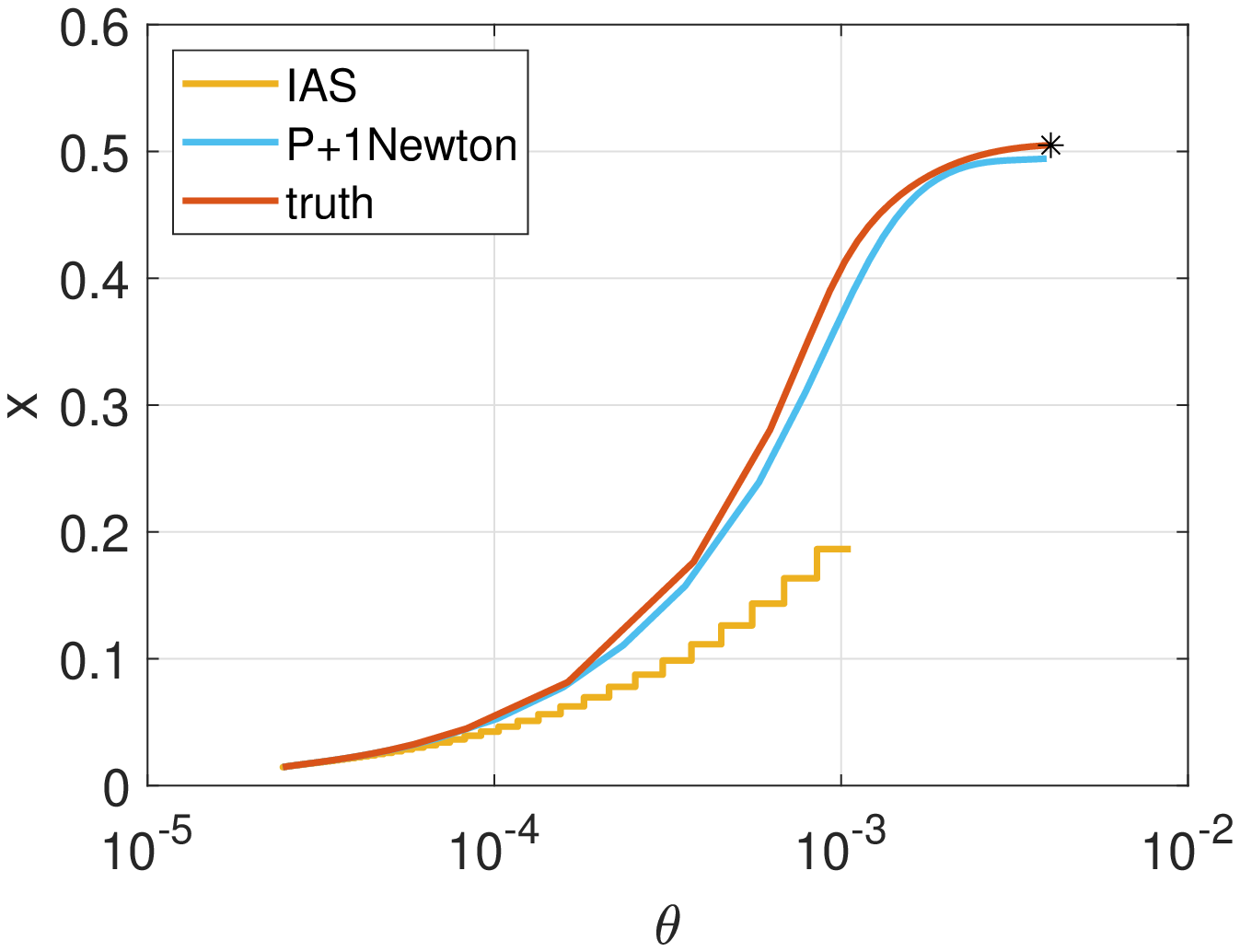}
   \end{minipage}
      \begin{minipage}{0.32\textwidth}
     \centering
     \includegraphics[width=0.93\linewidth]{ 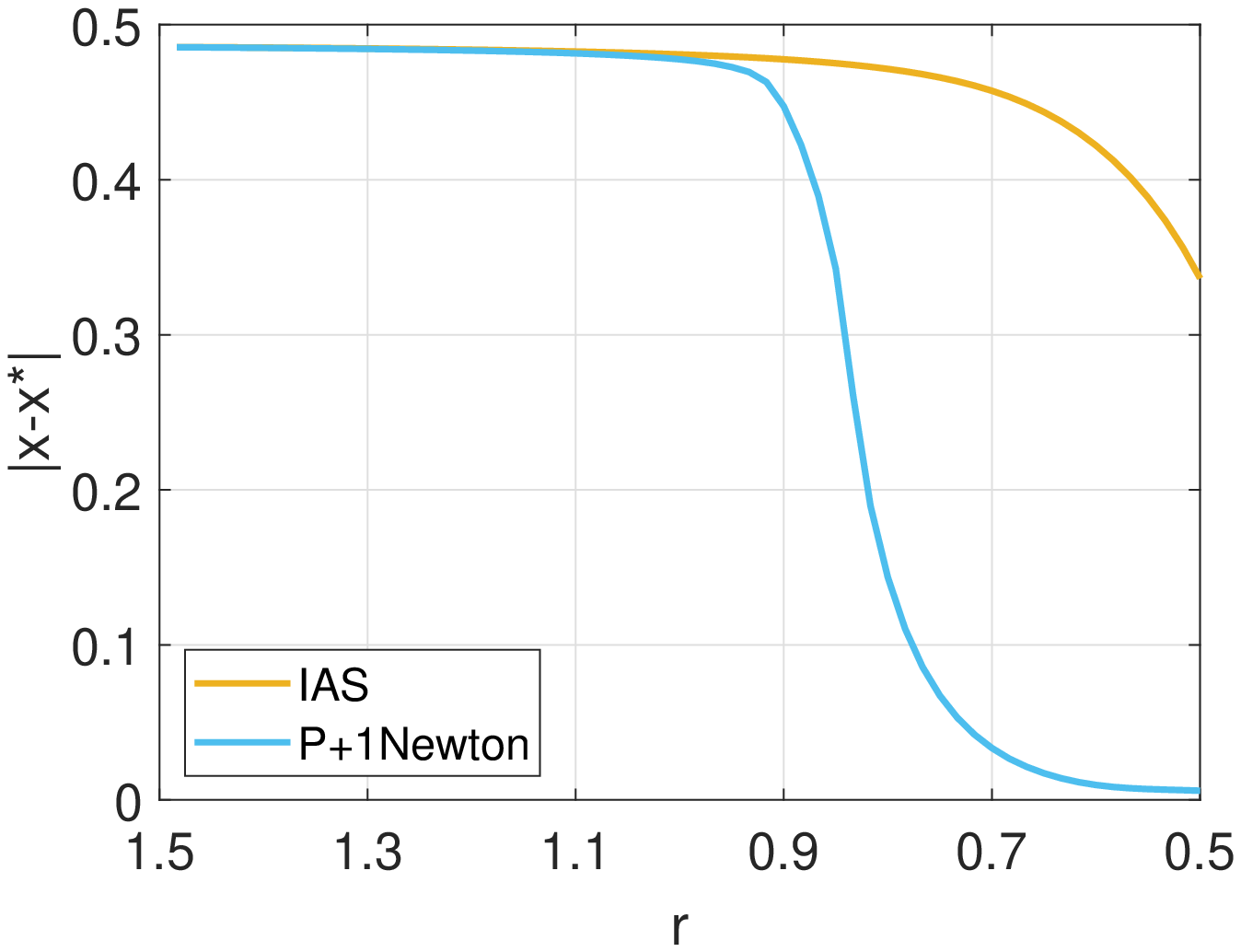}
   \end{minipage}
   \end{center}
   \caption{\textit{Left:} minimum of object function obtained by different path-following methods as the hyperparameters vary. \textit{Middle}: solution path of the third non-zero entry obtained by IAS and Predictor-Corrector. The curve labelled truth is a solution path obtained by running enough IAS iterations so that the solution converges at each time step and black asterisks represents true values at the end of the path. \textit{Right:} the distance between the third non-zero entry of the recovered solutions and the endpoint of the true solution path marked with an asterisk in the middle panel. } 
   \label{fig:min}
\end{figure}

Figure \ref{fig:path} shows the entire solution paths obtained by Predictor-Newton, path-following IAS and predictor only algorithms. The Predictor-Newton path leads to the exact solution $x$. The other methods do not. IAS alone underestimates, while prediction alone is unstable. Unlike classic Lasso paths, in which all coefficients shrink \cite{friedman2010regularization}, the entries on the true support grow and converge to their true values along Predictor-Newton path. Here the entries on the true support grow since changing $r$ amounts to moving from a relaxed $\ell_2$ penalty to a relaxed $\ell_p$ penalty with $p < 1$. 

\begin{figure}
\begin{center}
   \begin{minipage}{0.33\textwidth}
     \centering
     \includegraphics[width=1\linewidth]{ 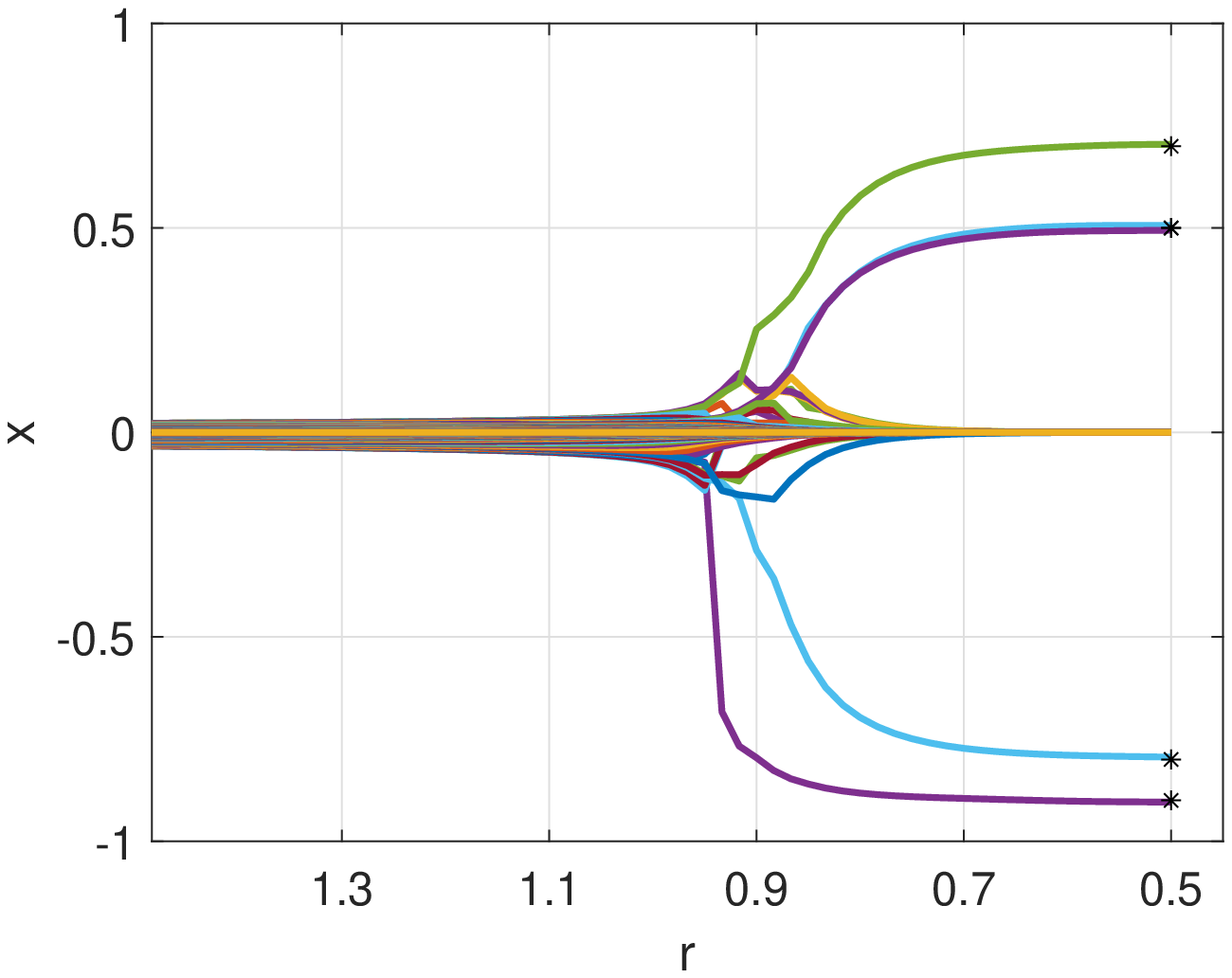}
   \end{minipage}\hfill
   \begin{minipage}{0.33\textwidth}
     \centering
     \includegraphics[width=1\linewidth]{ 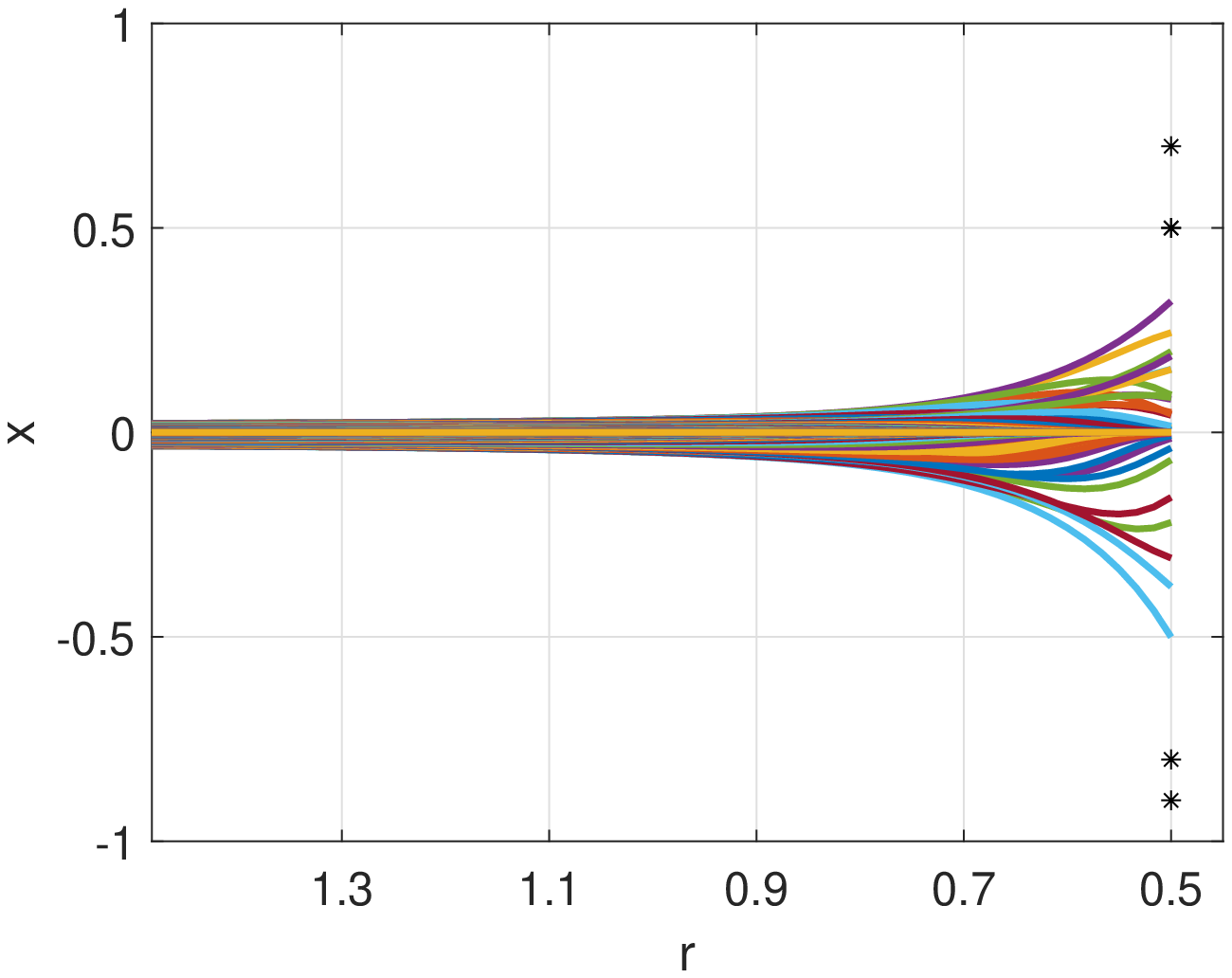}
   \end{minipage}
    \begin{minipage}{0.33\textwidth}
    \centering
    \includegraphics[width=1\linewidth]{ 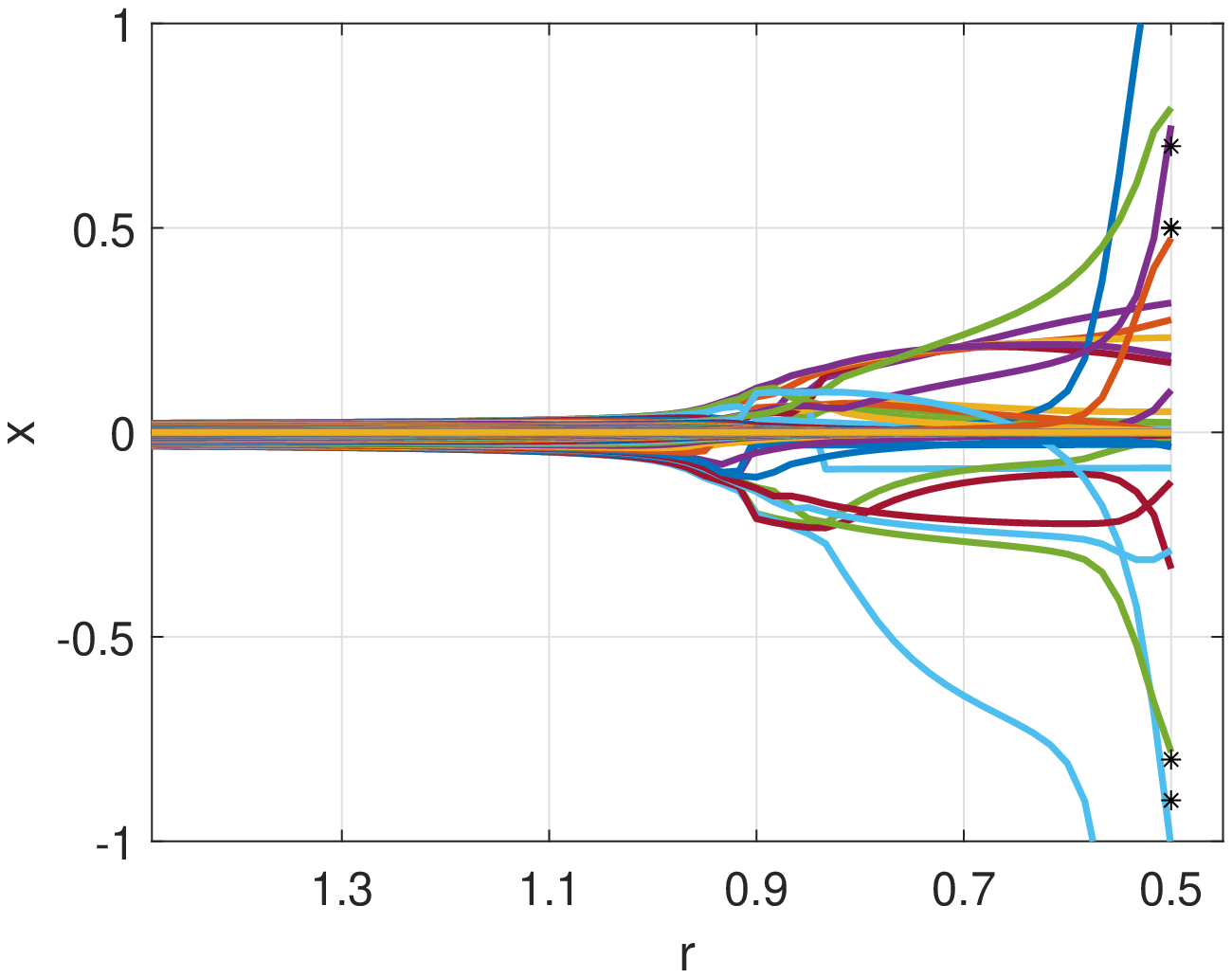}
   \end{minipage}
   \end{center}
   \caption{Entire solution paths of Predictor-Newton (left), IAS (middle) and predictor only (right). The black asterisks represent true values. Note that Predictor-Newton accurately identify the 5 non-zero entries and recovers their magnitude. IAS underestimates the magnitudes and fails to identify the true support, which indicates that 1 IAS iteration is not enough. The right panel shows that the Predictor alone is unstable.} 
   \label{fig:path}
\end{figure}

Figure \ref{fig:cond} displays the effect of preconditioning. Preconditioning radically reduces the condition number of the Hessian. The condition number approaches $10^7$ near $r = 1$. 
The preconditioned Hessian matrix is far better conditioned, with conditioner number less than 4 at all $r$. Near $r = 1$, preconditioning reduces the condition number by six orders of magnitude. 

Note that the sharp spike in the condition number of Hessian around $r = 1$ is smoothed out by the preconditioning. The region near $r = 1$ corresponds to the region where the objective function (\ref{obj}) enters the non-convex region. This observation is suggestive. When the Hessian is ill-conditioned the objective function is close to flat along some directions away from the minimum. In the extreme case when the condition number diverges the solution path may bifurcate. The spike in condition number near $r = 1$ indicates that, near the boundary between the convex and non-convex regions, the  optimization problem may allow bifurcating solutions. Bifurcating solutions likely arise from uncertainty in the support. As $r$ decreases, the MAP solution must resolve a sharp estimate of the support. The assumed support can be monitored by tracking which variances, $\theta$, remain large. For $r > 1$ some uncertainty in the assumed support is allowed, hence the peaks in the variances are broad. For $r \ll 1$ no uncertainty is allowed (c.f.~the left column of Figure \ref{fig:1Dscreenshots}). Thus, as $r$ decreases, the method must eventually resolve and commit to an assumed support. In principle this process could allow for many bifurcating solutions with similar supports, each converging to a set of sharp spikes consistent with the original, roughly localized spikes. Preconditioning eliminates this uncertainty entirely.

\begin{figure}
\begin{center}
   \begin{minipage}{0.5\textwidth}
     \centering
     \includegraphics[width=.8\linewidth]{ 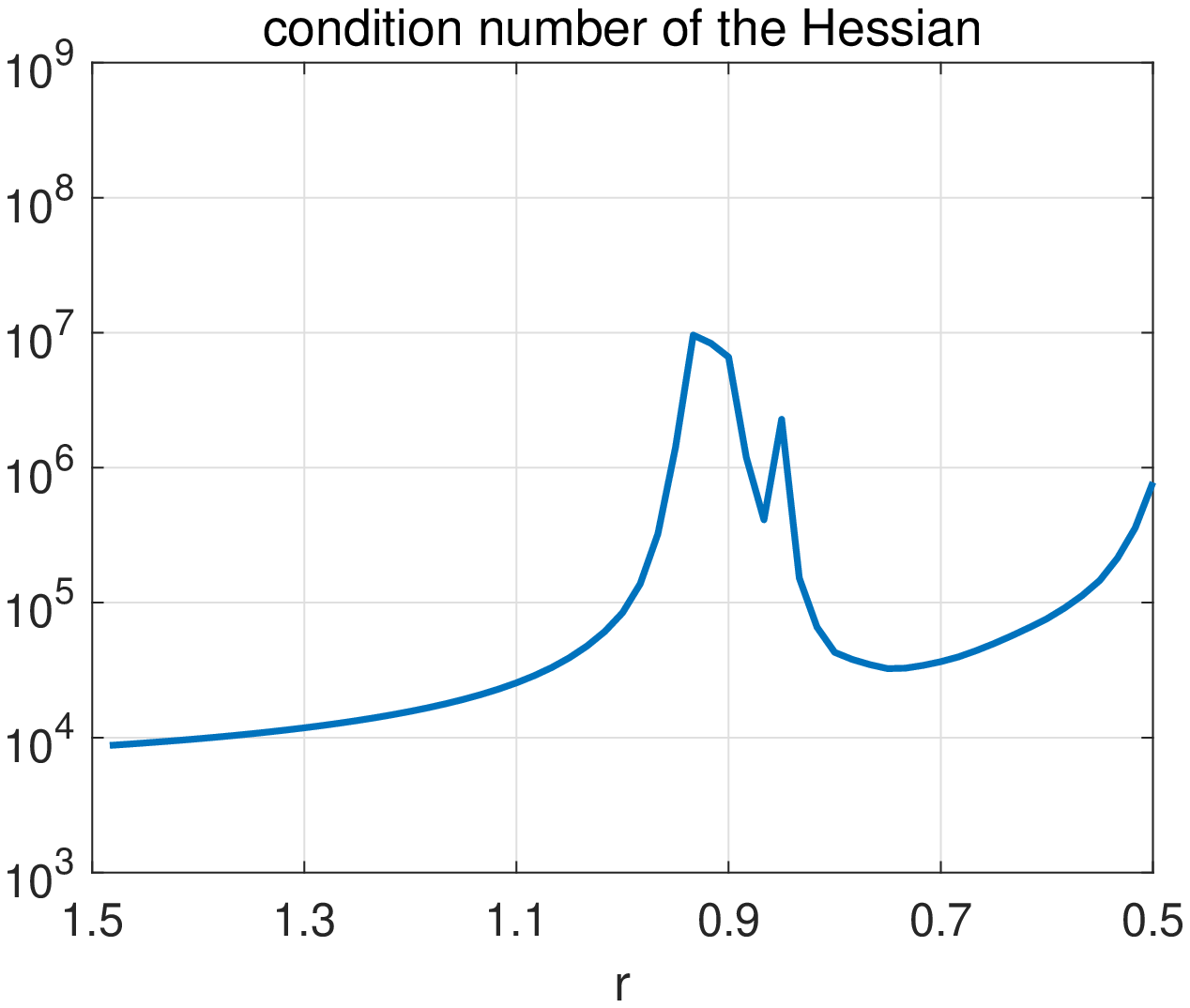}
   \end{minipage}\hfill
   \begin{minipage}{0.5\textwidth}
     \centering
     \includegraphics[width=.8\linewidth]{ 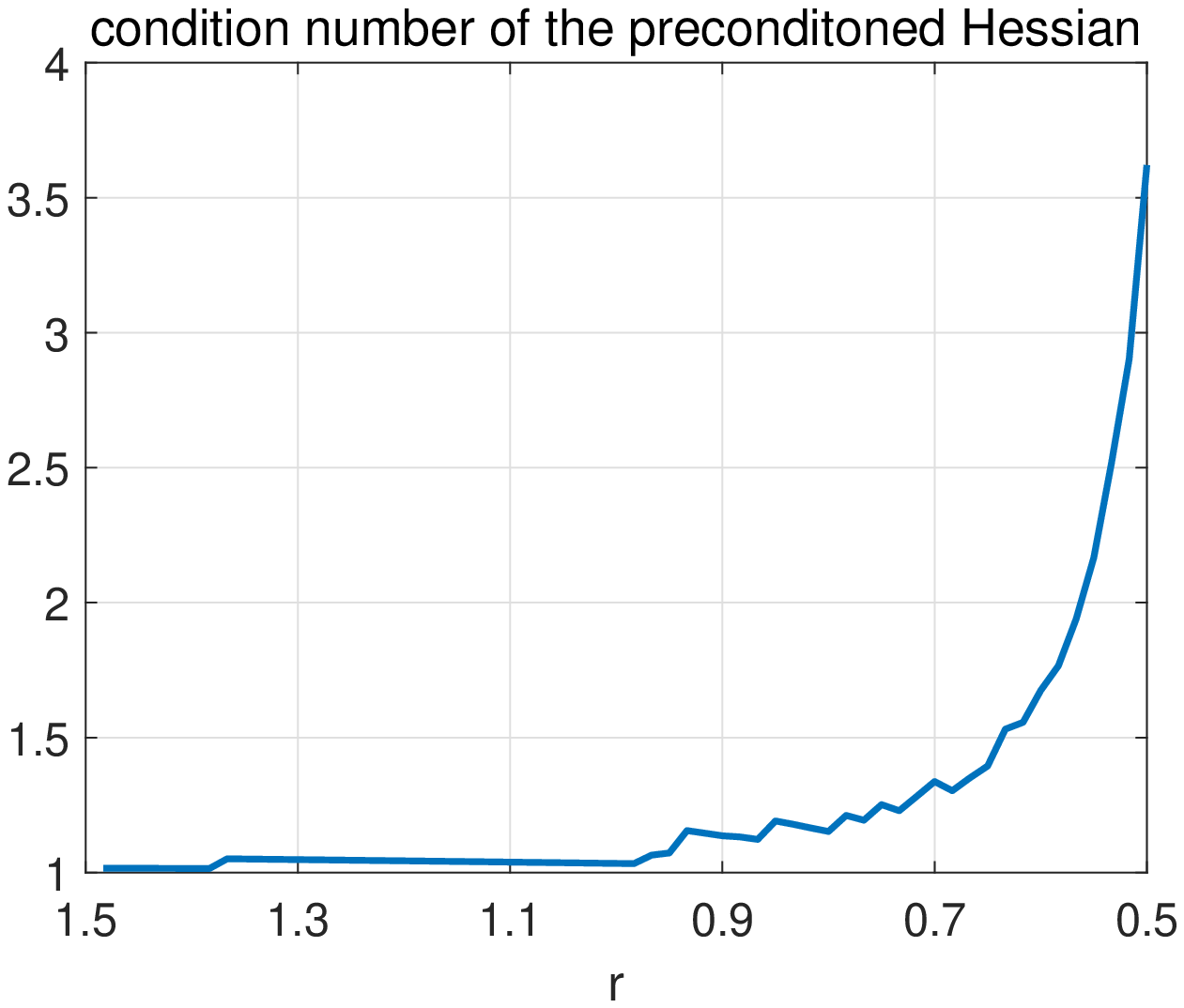}
   \end{minipage}
   \end{center}
   \caption{condition number of original Hessian (left) and preconditioned Hessian (right) along the solution path for $r$ varying from 1.5 to 0.5.}
   \label{fig:cond}
\end{figure}

The time cost of each method is displayed in Table \ref{tab:dectime}. The time cost for Predictor-Newton is subdivided into its constituent parts for comparison. IAS and inexact IAS are roughly 5 times faster than Predictor Newton, and 3 times faster than Predictor-IAS. Here, accuracy incurs computational expense. 

While Predictor-Newton is slower overall than the IAS based methods, it solves its linear systems much faster. Despite solving twice as many systems as either IAS method, all with matrices of twice the size, the total run time spent solving systems is half the time spent by IAS. This speed up indicates of the efficacy of preconditioning. 



\begin{table}[]
\footnotesize
\begin{center}
\begin{tabular}{cccccccc}
\hline
IAS & Inexact & \multicolumn{5}{c}{P-Newton} & P-IAS \\\hline
   &        & Precondition  & CGS     & Build Linear System      & Backtrack     & Total    &             \\\cline{3-7}
0.35  & 0.45        & 0.67     & 0.13     & 0.91    & 0.18     & 2.03    & 1.33     \\ \hline       
\end{tabular}
\end{center}
\caption{Run-times (in seconds) for path-following IAS/Inexact IAS, Predictor-Newton, Predictor-IAS}
\label{tab:dectime}
\end{table}


Next, we show how solutions differ along three hyperparameter paths:

\begin{enumerate}
\item \textbf{Path 1: }  $(r,\eta,\vartheta) = (1.5,1.5,10^{-5}) \rightarrow (0.5,1.5,10^{-5}) \rightarrow (0.5,10^{-5},10^{-6}) $ 
\item \textbf{Path 2:}   $(r,\eta,\vartheta) = (1.5,1.5,10^{-5}) \rightarrow (1.5,10^{-5},10^{-5}) \rightarrow (0.5,10^{-5},10^{-6}) $ 
\item \textbf{Path 3: }  $(r,\eta,\vartheta) = (1.5,1.5,10^{-5}) \rightarrow (1.5,1.5,10^{-6}) \rightarrow (0.5,10^{-5},10^{-6}) $ 
\end{enumerate}

The 3 paths have the same starting and ending point, but differ at their midpoint. In each, we vary one hyperparameter first while holding the other two fixed. In path 1, $r$ first, in path 2, $\eta$ first and in path 3, $\vartheta$ first. 

Figure \ref{fig:3path} displays the intermediate and the final solutions. Different hyperparameter paths lead to different solutions paths. By varying one hyperparameter at a time, we isolate the influence of each hyperparameter. Given $r>0$, when $\eta \rightarrow 0$, the solution converges the $\ell_{p}$ penalized solution where $p=\frac{2 r}{r+1}$ \cite{calvetti2020sparse}.

\begin{figure}
\begin{center}
   \begin{minipage}{0.5\textwidth}
     \centering
     \includegraphics[width=.8\linewidth]{ 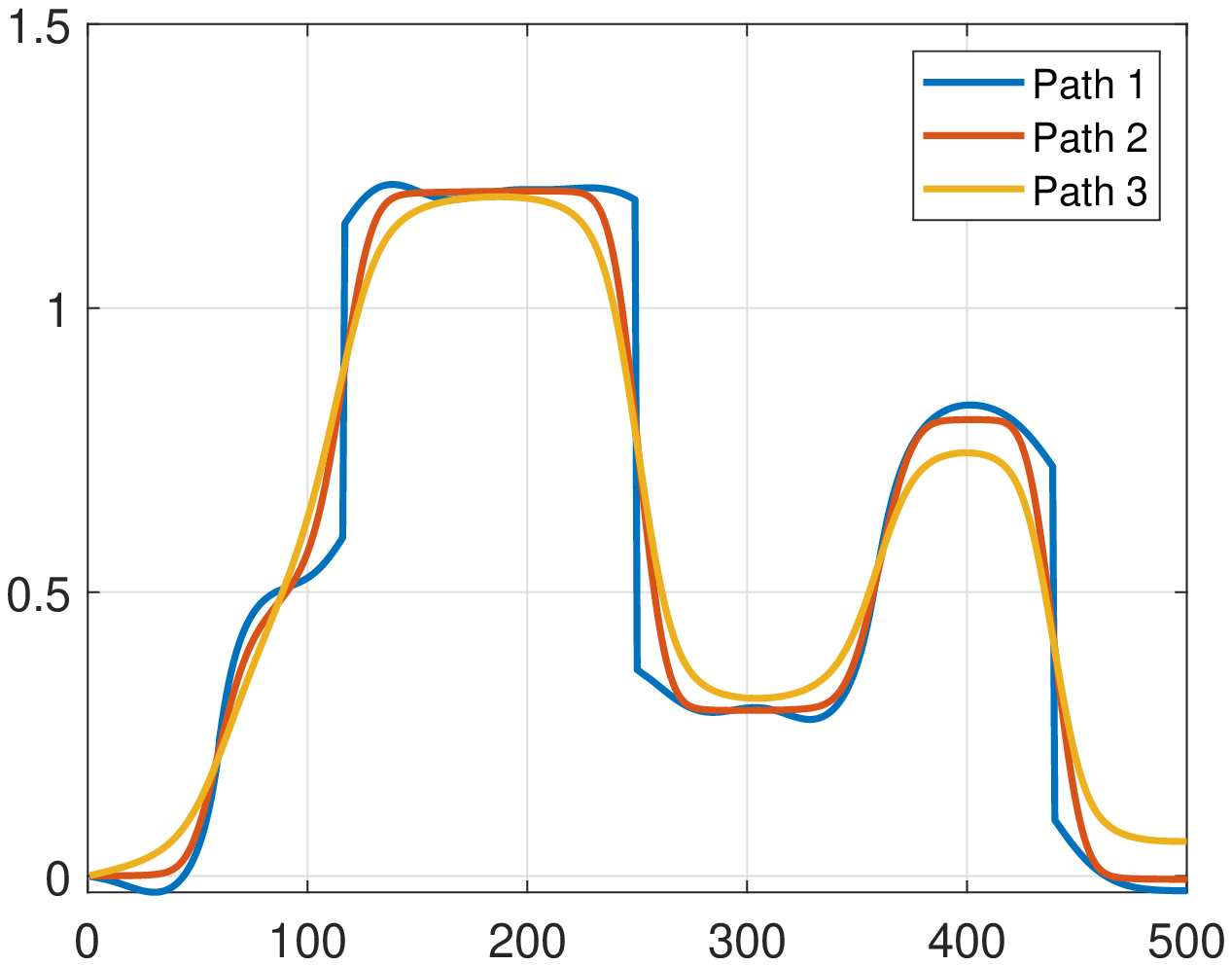}
   \end{minipage}\hfill
   \begin{minipage}{0.5\textwidth}
     \centering
     \includegraphics[width=.8\linewidth]{ 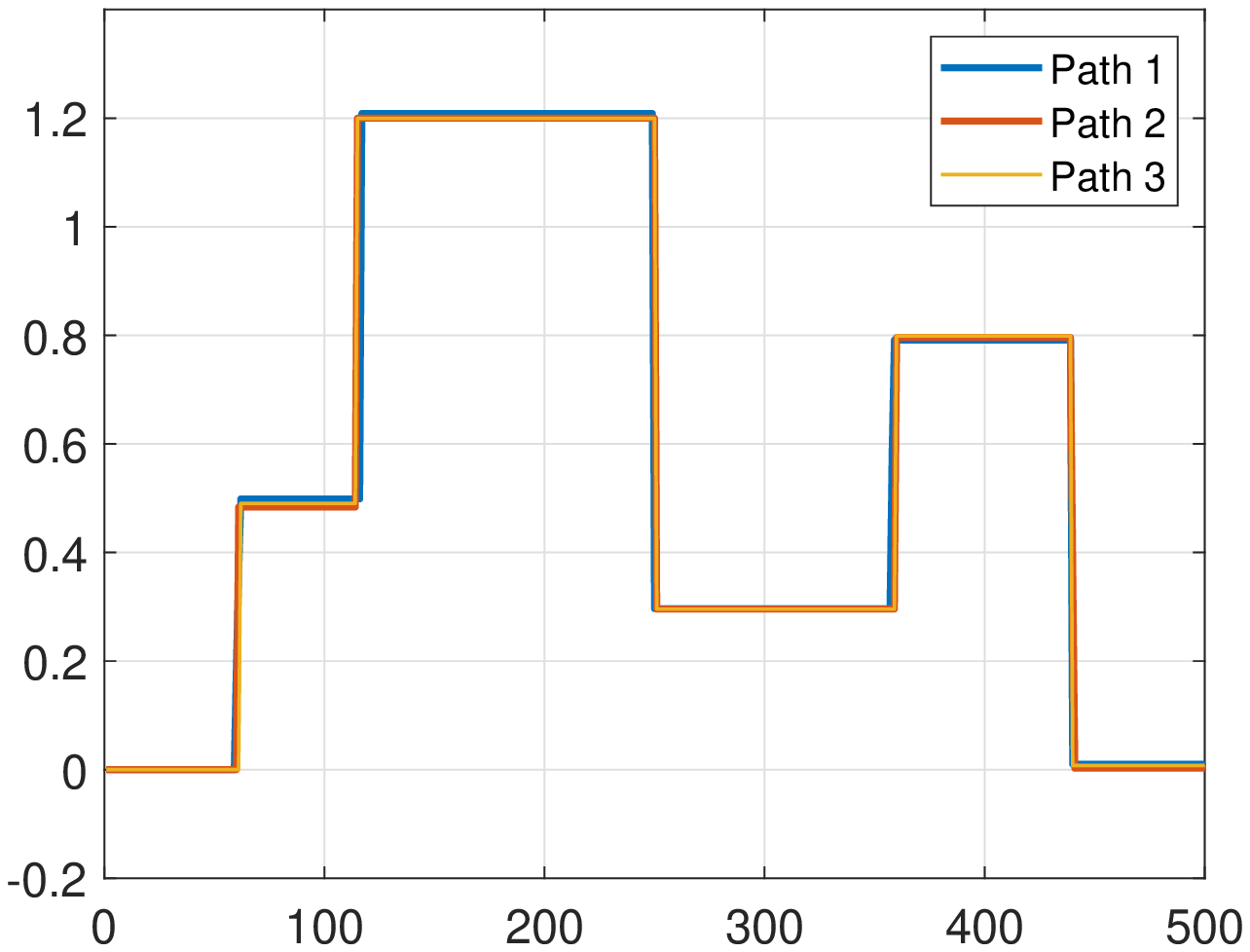}
   \end{minipage}
   \end{center}
   \caption{ \textit{Left}: solutions $x$ at the turning points of the 3 paths, where $(r,\eta,\vartheta) = (0.5,1.5,10^{-5})$ in path 1, $(1.5,10^{-5},10^{-5}) $ in path 2 and $(1.5,1.5,10^{-6})$ in path 3. \textit{Right}: solutions $x$ at the end of the hyperparameter paths. Note that while the solutions at intermediate points differ, they converge at the end of the path, suggesting that the minima selected by convex relaxation may be path independent.}
\label{fig:3path}
\end{figure}

Along path 1, $p = \frac{3}{5}$ at the midpoint, so the solution is sparser than the $\ell_{1}$ solution. However since $\eta$ is small yet, it only enforces the $\ell_{3/5}$ penalty softly. Consequently, it produces sharply defined jumps, but does not enforce that the solution remain constant between the jumps.

Along path 2, $p = \frac{6}{5} > 1$ at the midpoint, but $\eta$ is small, so the solution is effectively an $\ell_{6/5}$ penalized solution. Since $1 < \frac{6}{5} < 2$, the penalty only partially promotes sparsity, and the resulting signal does not jump sharply. However, since $\eta$ is small, the penalty is unsoftened, and the solution is constant away from the jumps.

On path 3, the scaling hyperparameter $\vartheta$ approaches zero first, so, after non-dimensionalization, the effective regularizer acts like a classic $\ell_{2}$ penalized solution. The $\ell_2$ penalized solution is the smoothest of all three penalties considered, so solutions neither jump sharply, nor remain constant away from the jumps.

Despite the observed differences along the three paths, all paths end at the same hyperparameter values, so the resulting solutions converge. Slight differences in the solutions are apparent in the right panel of Figure \ref{fig:3path}. These differences are likely a result of the distinct errors accumulated along the different paths. This experiment demonstrates that the path-following approach can recover a solution in the non-convex region that is robust to changes in the path taken from the convex region.

\subsection{Image Problem}

Next, we estimate a nearly black two-dimensional object. This example is borrowed from \cite{calvetti2020sparse}. 

 The generating model is an impulse image, defined on
 $\Omega = [0,1] \times[0,1]$,
$$
d \mu(p)=\sum_{k=1}^{J} a_{k} \delta\left(p-p_{k}\right) d p, \quad p_{k} \sim \operatorname{Uniform}(\Omega), \quad a_{k} \sim \operatorname{Uniform}([1.5,2]),
$$
The desired distribution, $d \mu$ is discretized, down-sampled, and observed after blurring with a Gaussian kernel,
$$
A\left(p, p^{\prime}\right)=\frac{1}{2 \pi w^{2}} e^{-\left\|p-p^{\prime}\right\|^{2} / 2 w^{2}}, \quad w=0.01
$$
Then, the discrete data at observation points $q_{j} \in \Omega$ is
$$
b_{j}=\int_{\Omega} A\left(q_{j}, p^{\prime}\right) d \mu\left(p^{\prime}\right)+\varepsilon_{j}=\sum_{k=1}^{K} a_{k} A\left(q_{j}, p_{k}\right)+\varepsilon_{j} .
$$
where $\epsilon$ is Gaussian noise. 

The image $\Omega$ is divided into $n=128 \times 128=16384$ pixels, denoted by $\Omega_{\ell}$. The kernel is discretized and denoted by $A$ as
$$
\int_{\Omega} A\left(q_{j}, p\right) d \mu(p) \approx \sum_{\ell=1}^{n} \underbrace{\left|\Omega_{\ell}\right| A\left(q_{j}, q_{\ell}^{\prime}\right)}_{={A}_{j \ell}} x_{\ell}, \quad x_{\ell}=\frac{1}{\left|\Omega_{\ell}\right|} \int_{\Omega_{\ell}} d \mu(p),
$$
where $q_{\ell}^{\prime}$ denotes the center point of the pixel $\Omega_{\ell}$ and $\left|\Omega_{\ell}\right|$ is its area. We assume that the number of observation points is $m=64 \times 64=4096$. Then, the forward operator is defined by a matrix ${A} \in \mathbb{R}^{m \times n}$. The signal is corrupted by scaled white noise with standard deviation approximately $1.8 \%$ of the maximum noiseless signal.

To study the sensitivity of MAP estimation, we select a hyperparameter path that starts in the convex setting $(r(0) , \eta(0), \vartheta(0)) $ $= (1.5,1.5,10^{-5})$, and ends in the non-convex setting $(r(T) , \eta(T), \vartheta(T)) = (0.5,10^{-5},10^{-6})$. We use 8 equidistant time points along a line connecting the initial and final assumptions.

At the starting point, we test the efficacy of Newton acceleration. Figure \ref{fig:imagecost} shows the results. We switch to Newton after 4 IAS iterations. After switching, the method converges in 3 Newton iterations. Pure IAS needs more than 10 iterations. 

\begin{figure}
\begin{center}
   \begin{minipage}{0.5\textwidth}
     \centering
     \includegraphics[width=.82\linewidth]{ 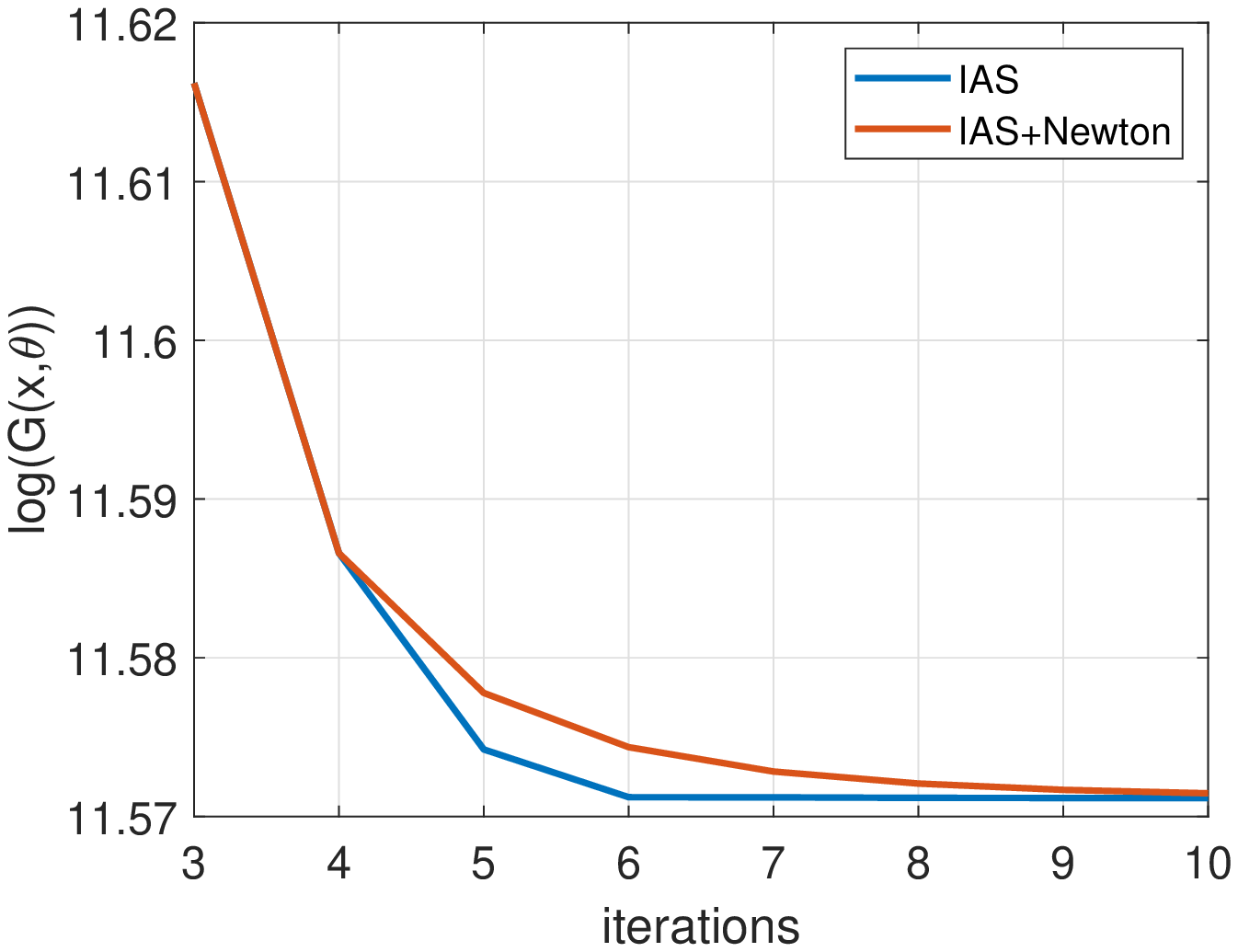}
   \end{minipage}\hfill
   \begin{minipage}{0.5\textwidth}
     \centering
     \includegraphics[width=.8\linewidth]{ 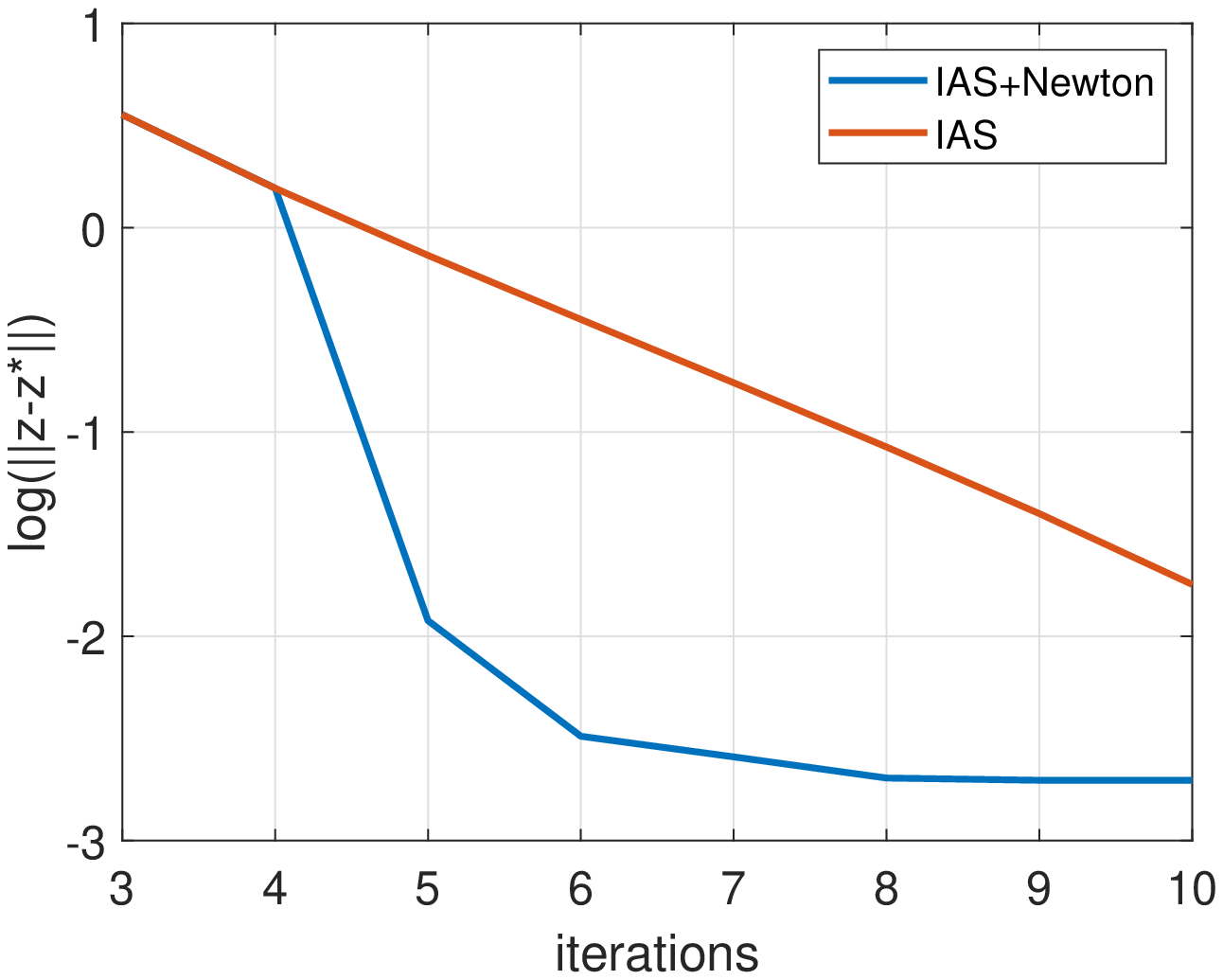}
   \end{minipage}
   \end{center}
   \caption{objective function (left) and 2-norm error (right) for IAS, and IAS with Newton acceleration.}
   \label{fig:imagecost}
\end{figure}

Figure \ref{fig:imageshots} shows the solution $x(t)$, reconstructed image, and variance $\theta(t)$ recovered by the path-following methods at 3 different points on the hyperparameter path. As $r(t)$ decreases, the obtained solution become increasingly sparse. All algorithms start from the same point hyperparameters, where the solution is accurately obtained by IAS with Newton acceleration. At the middle and end points, Predictor-Newton identifies the support most accurately and promotes sparsity most strongly.

\begin{figure}
\begin{center}
	\begin{minipage}{0.3\linewidth}
		\centering
		\includegraphics[width=1\linewidth]{ 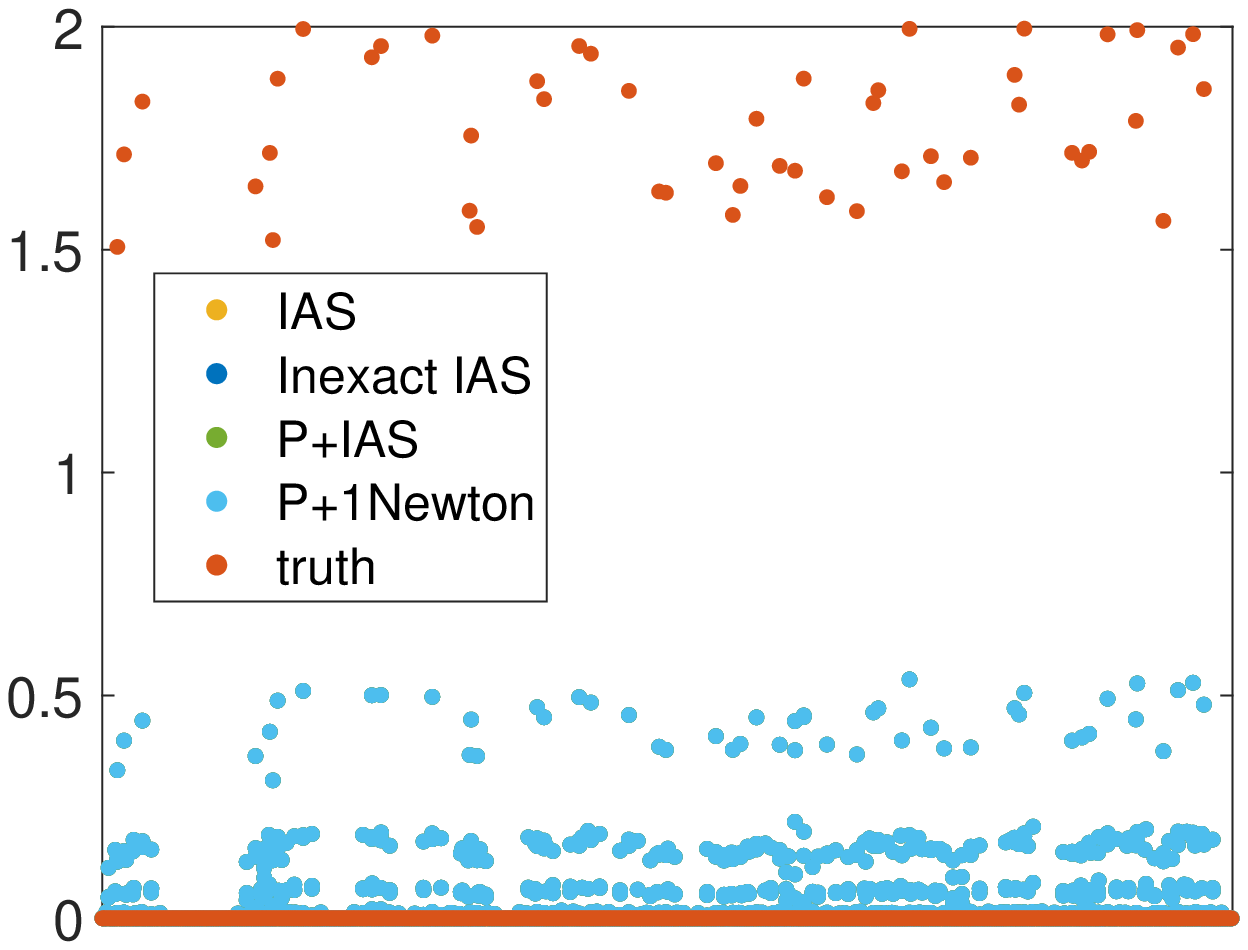}
	\end{minipage}
	\begin{minipage}{0.3\linewidth}
		\centering
		\includegraphics[width=.95\linewidth]{ 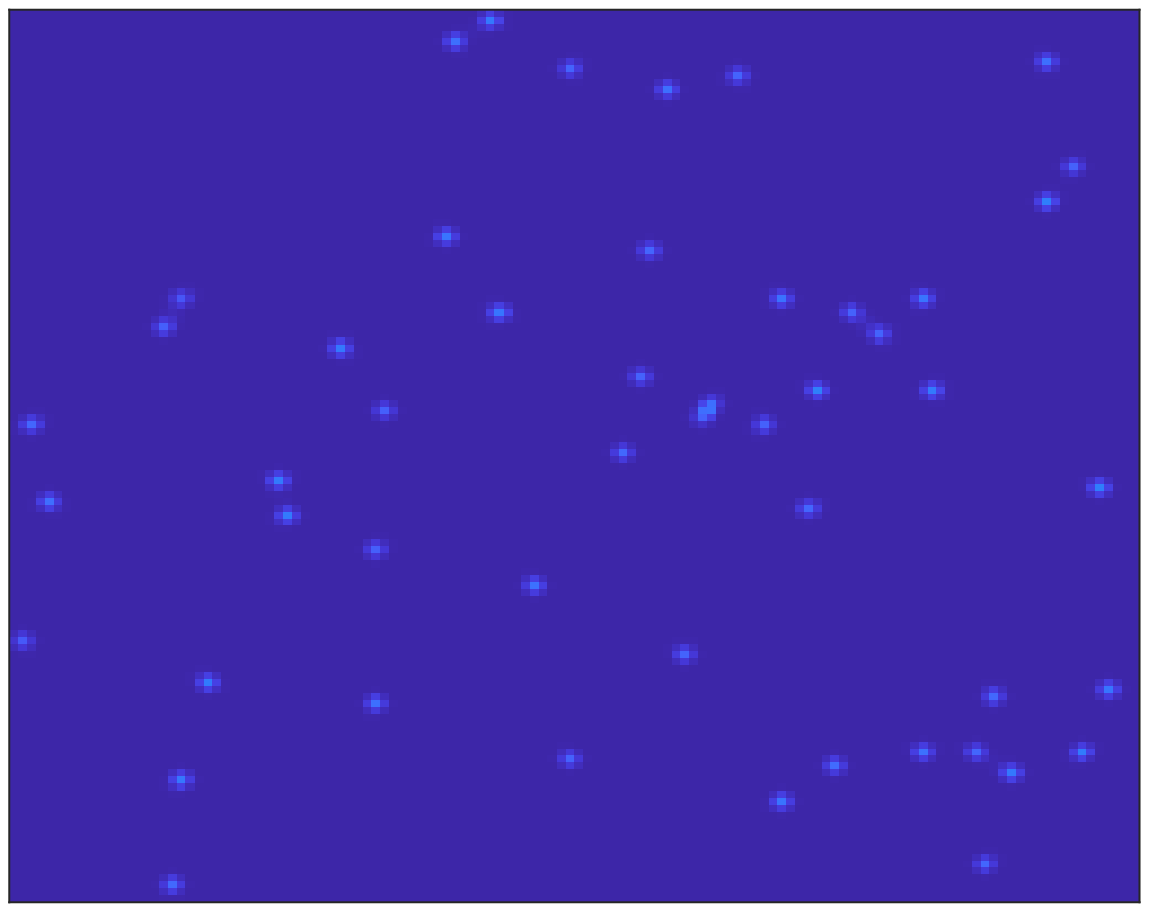}
	\end{minipage}
	\begin{minipage}{0.3\linewidth}
		\centering
		\includegraphics[width=.95\linewidth]{ 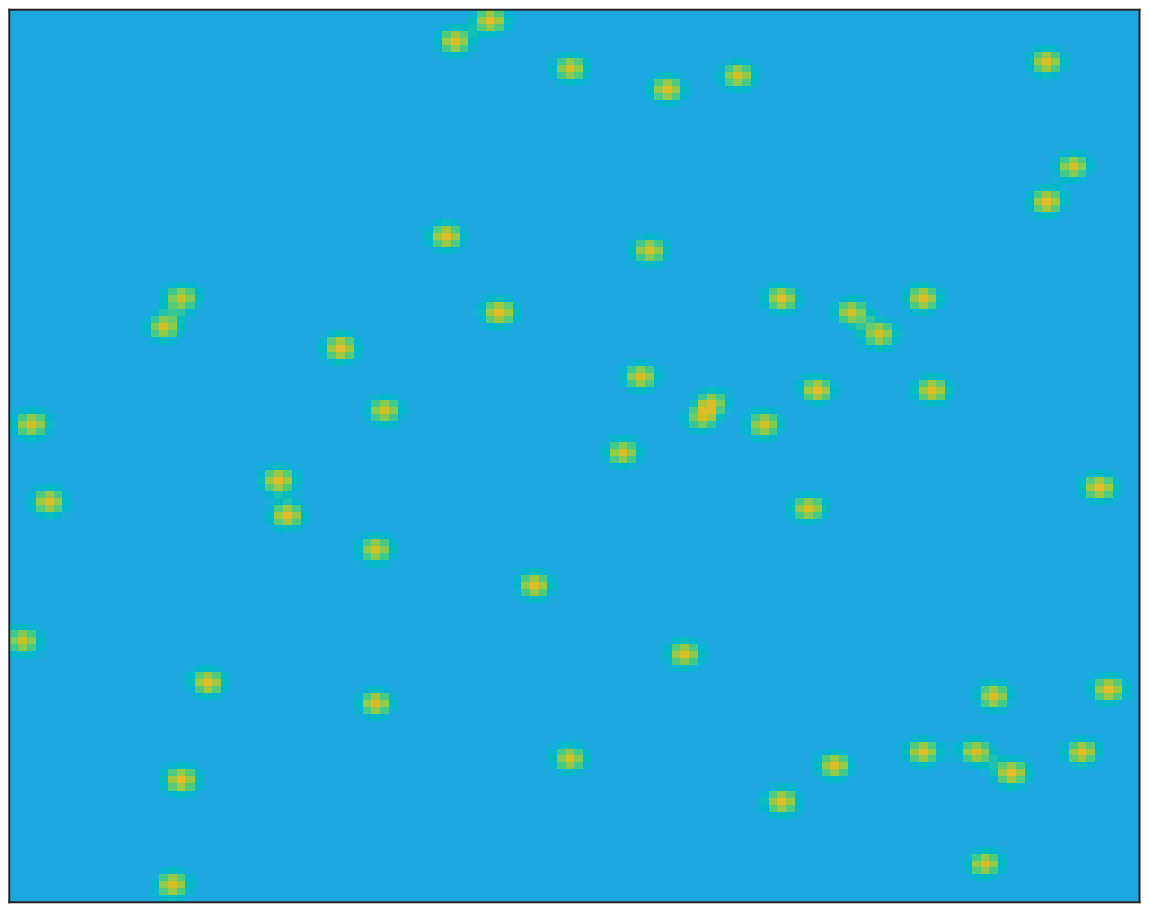}
	\end{minipage}
	
	\begin{minipage}{0.3\linewidth}
		\centering
		\includegraphics[width=1\linewidth]{ 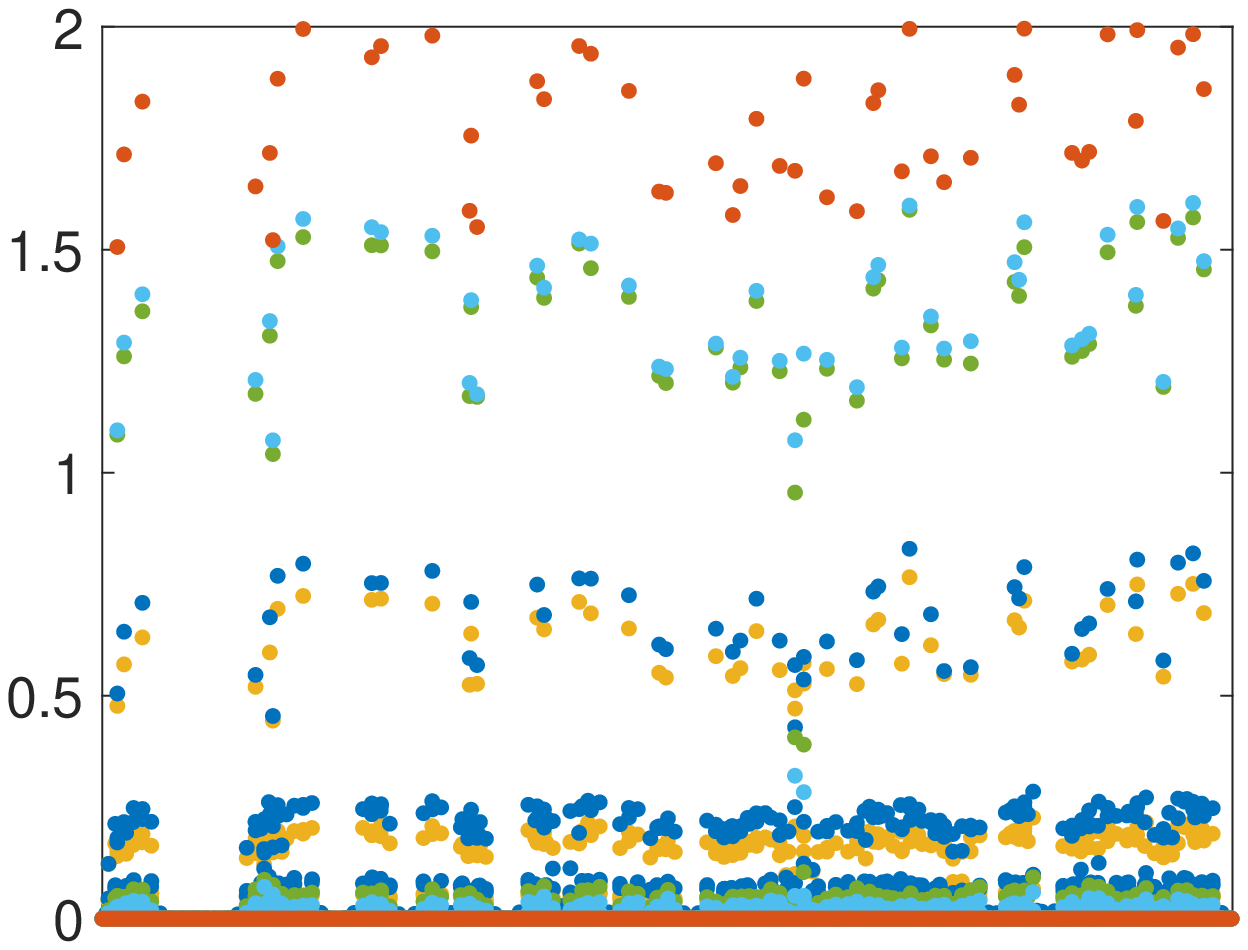}
	\end{minipage}
	\begin{minipage}{0.3\linewidth}
		\centering
		\includegraphics[width=.95\linewidth]{ 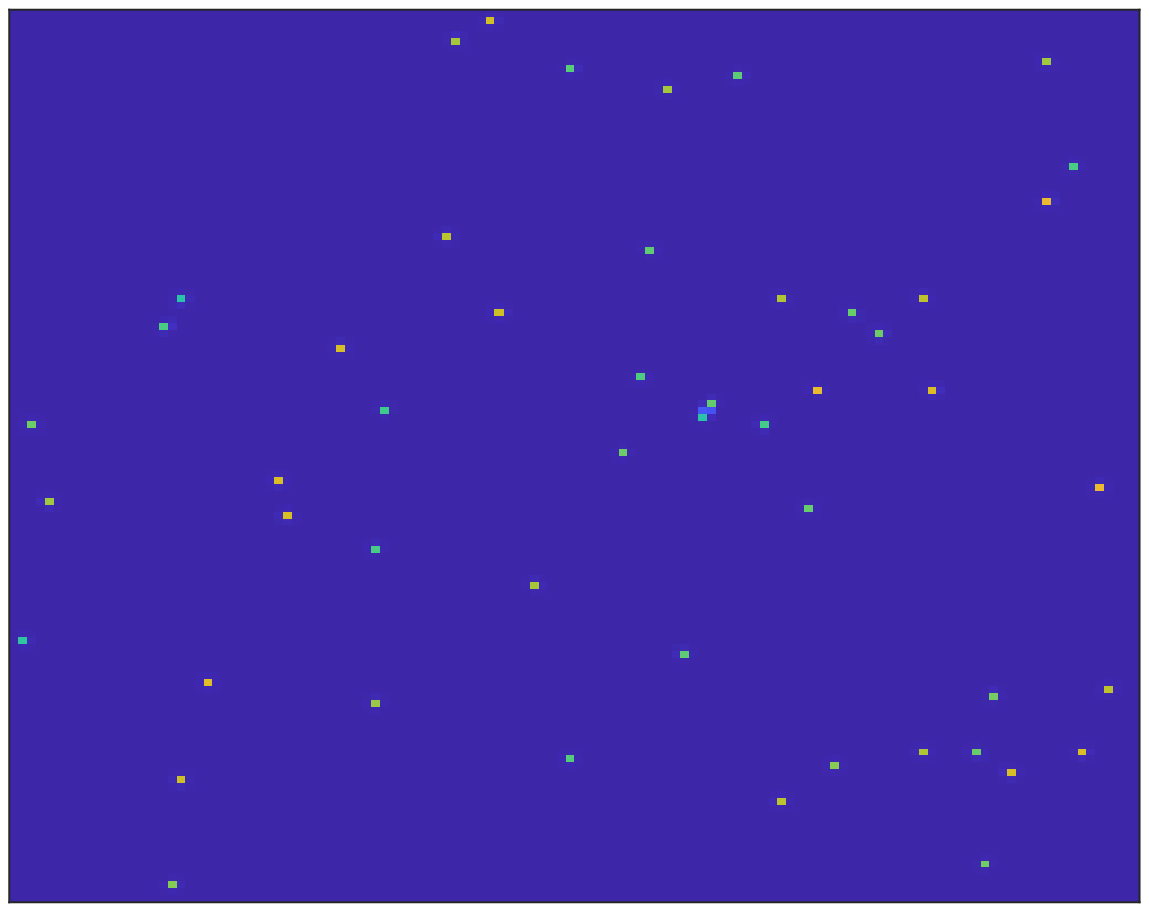}
	\end{minipage}
	\begin{minipage}{0.3\linewidth}
		\centering
		\includegraphics[width=.95\linewidth]{ 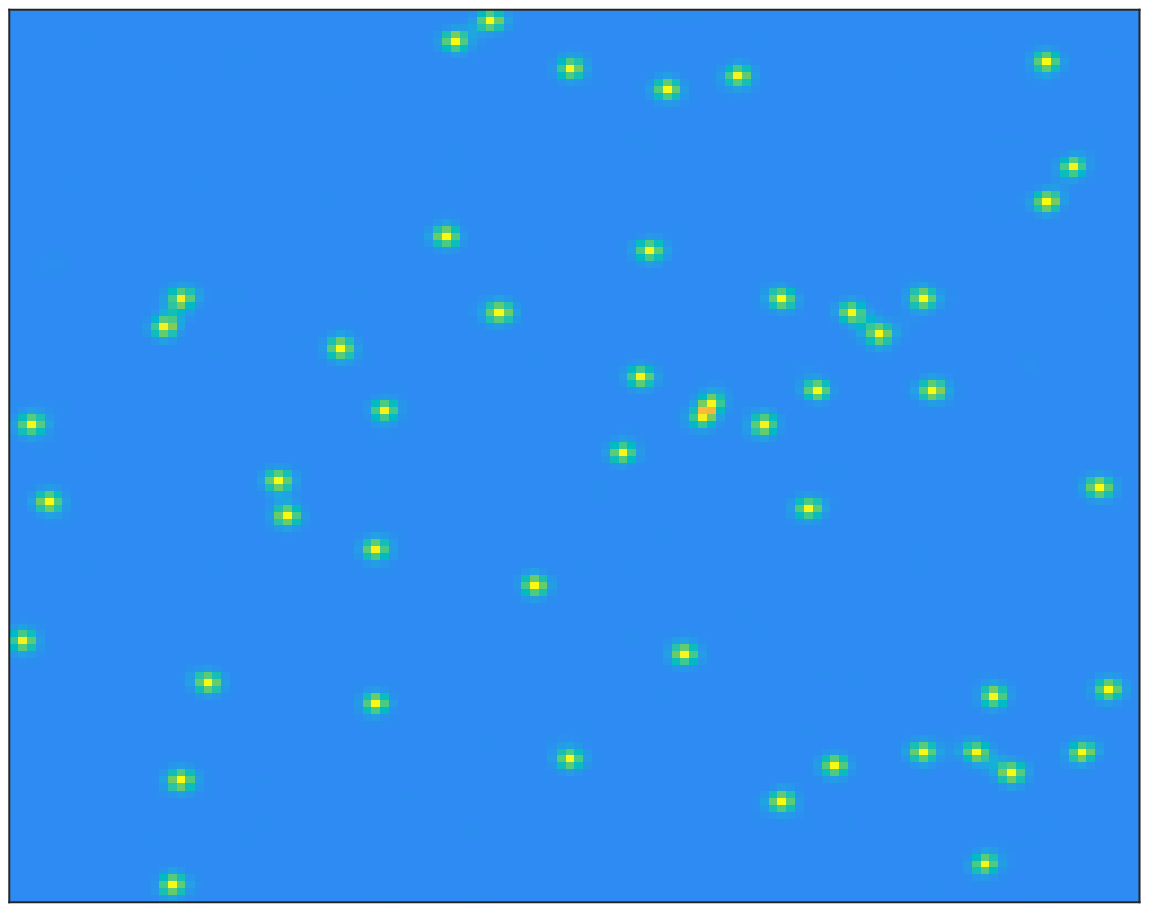}
	\end{minipage}
	\begin{minipage}{0.3\linewidth}
		\centering
		\includegraphics[width=1\linewidth]{ 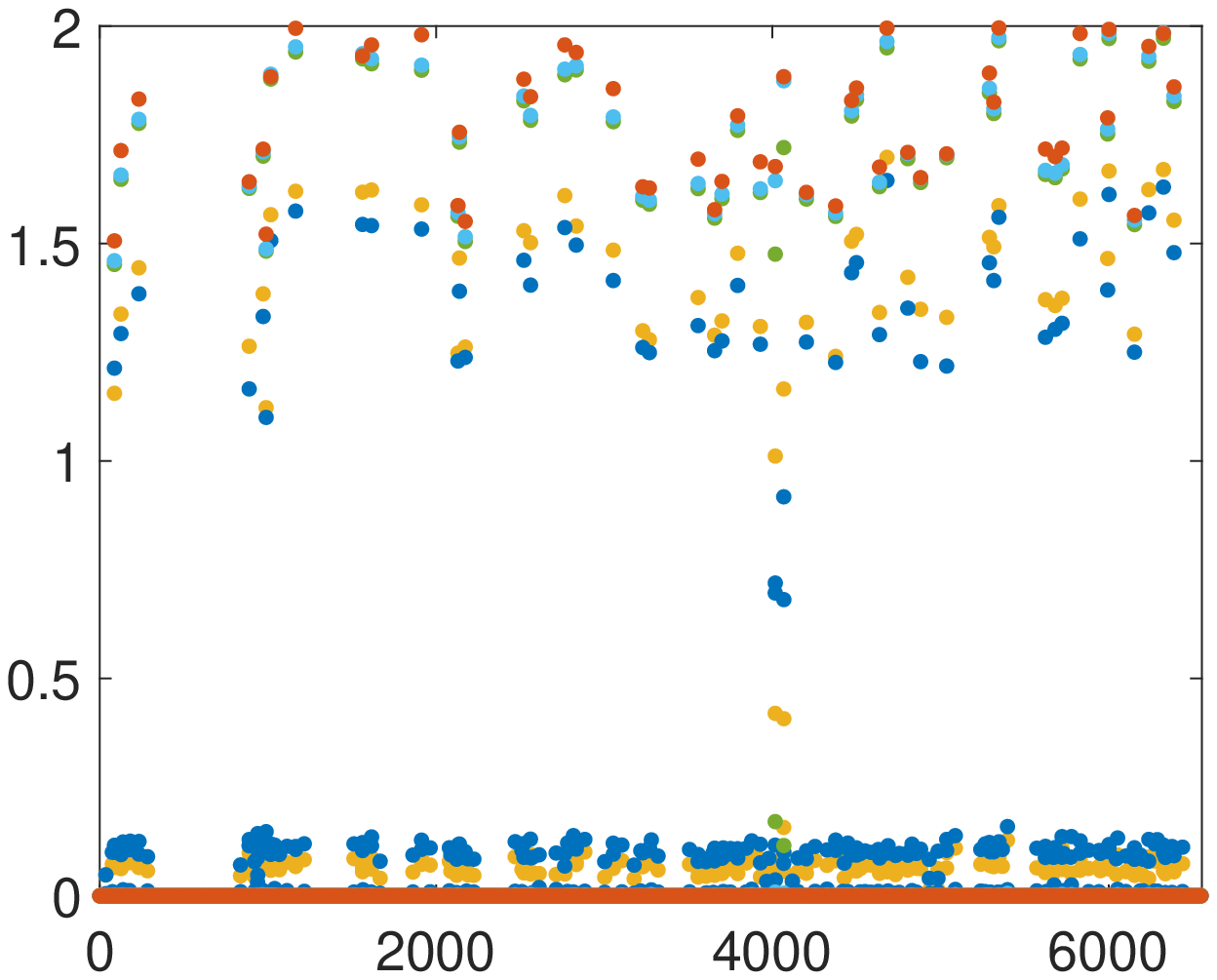}
	\end{minipage}
	\begin{minipage}{0.3\linewidth}
		\centering
		\includegraphics[width=.97\linewidth]{ 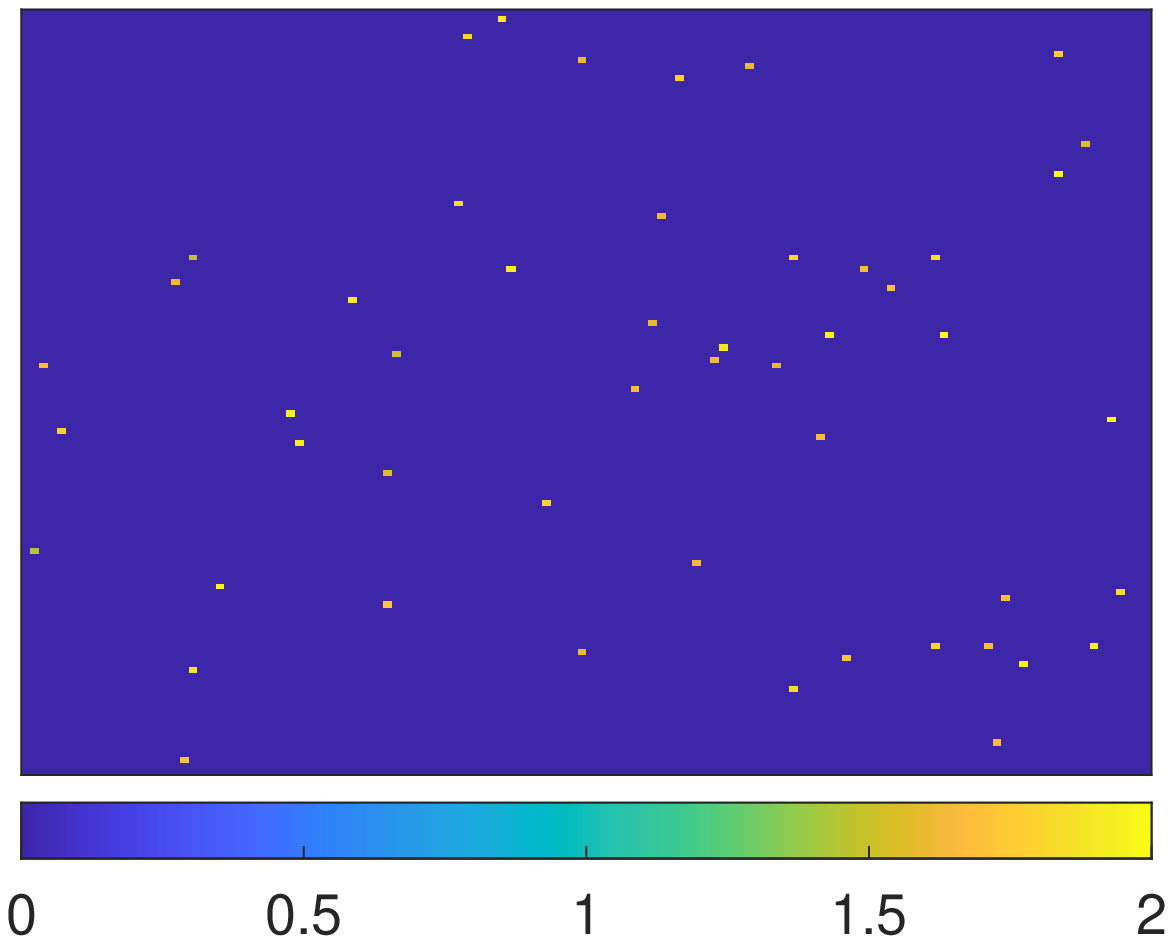}
	\end{minipage}
	\begin{minipage}{0.3\linewidth}
		\centering
		\includegraphics[width=.98\linewidth]{ 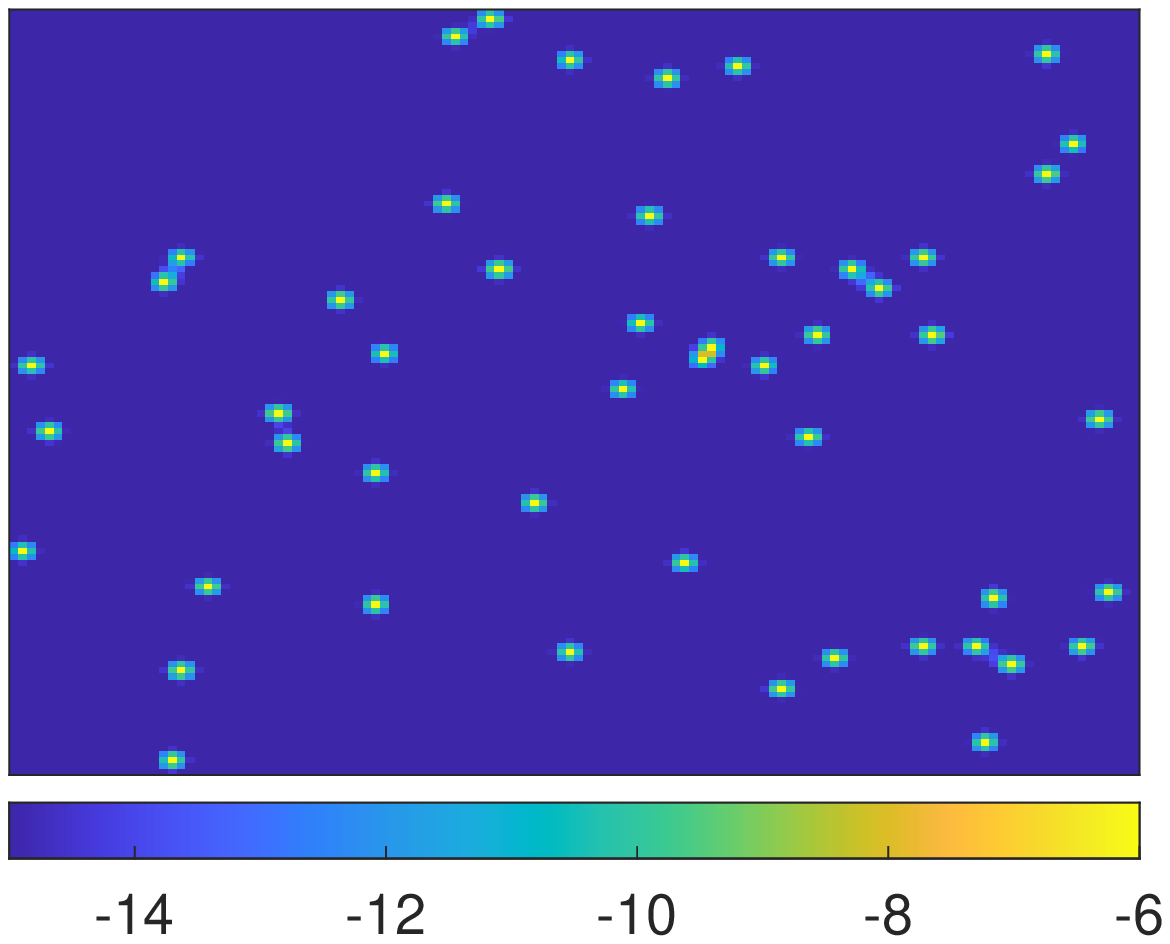}
	\end{minipage}
\end{center}
\caption{solution $x$ (left), reconstructed image via Predictor-Newton (middle) and the variance $\theta$ (right) via Predictor-Newton corresponding to 3 different points on the hyperparameters path. In the top row, $r = 1.5$ , $\eta =  1.5$ and $\vartheta = 10^{-5}$, in the middle row, $r = 1.1$ , $\eta =  0.9$ and $\vartheta = 5.5\times 10^{-6}$, in the bottom row, $r = 0.5$ , $\eta =  10^{-5}$ and $\vartheta = 10^{-6}$.}
\label{fig:imageshots}
\end{figure}

Figure \ref{fig:imagepre} demonstrates how the preconditioning strategy leverages the effective dimensionality of the problem. The dimension of original Hessian is 32,768. Screening out columns with small column sums in $D_{\theta}^{\frac{1}{2}}A^{\top}AD_{\theta}^\frac{1}{2}$ reduces the dimension by at least an order of magnitude. As $r(t)$ shrinks, the number of columns retained falls rapidly. Low rank approximation of the screened matrix further reduces the dimension. The blue line marks the effective rank the fidelity term. At largest, the effective rank is near 500. Both the dimension of the screened matrix, and effective rank converge to 50, the true number of non-zero values in the signal (stars in the original image). As $r$ decreases, the components of $\theta$ in the support become larger and the components off the support become smaller. Hence the effective rank of the Hessian decreases and converges to the real rank of the problem. Consequently, the preconditioner (\ref{precondition}) is relatively cheap to build.

\begin{figure}
\begin{center}
    \centering
     \includegraphics[width=.4\linewidth]{ 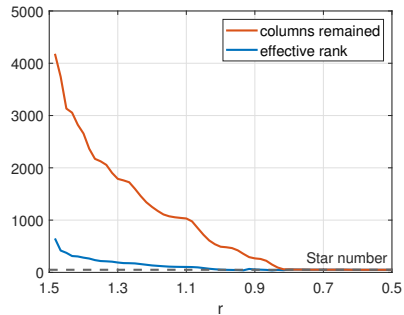}
   \end{center}
\caption{number of columns remaining after screening and the effective rank of the Hessian.}
\label{fig:imagepre}
\end{figure}

Table \ref{tab:imagetime} compares the computational costs of the path-following methods. The Predictor-Newton method is comparable to the Inexact IAS algorithm, which makes it attractive for large scale problems with sparse underlying signals. IAS type algorithms need more iterations to achieve a comparable accuracy. Note that, despite the scale of the problem, computing the preconditioner takes less time than in the first deconvolution example, since the effective rank of the image problem is smaller for large and intermediate $r$. Also notice that, despite solving twice as many linear systems, all twice as large, as either IAS method, the time spent solving linear systems by predictor-Newton is less than the runtime of the IAS based methods. For example, the time Predictor-Newton spent solving linear systems via CGLS is roughly a third the time Inexact IAS path-following spent. Thus, preconditioning is both relatively cheap, and significantly reduces the time spent solving linear systems.

\begin{table}[]
\footnotesize
\begin{center}
\begin{tabular}{cccccccc}
\hline
IAS & Inexact & \multicolumn{5}{c}{P-Newton} & P-IAS \\\hline
   &        & Precondition  & CGS     & Build Linear System      & Backtrack    & Total    &             \\\cline{3-7}
 18.83    & 5.92        & 0.59     & 1.68     & 2.55    & 0.51     & 5.26    & 21.17      \\ \hline       
\end{tabular}
\end{center}
\caption{Run-times (in seconds) for path-following IAS/Inexact IAS, Predictor-Newton, Predictor-IAS}
\label{tab:imagetime}
\end{table}

\section{Discussion}

The analysis and methods introduced in this paper can be extended in multiple directions. These avenues for future work address shortcomings of our approach, and exploit other aspects of the path-following problem.

First, the realized path of MAP solutions depends on the specific hyperparameter path. The solution does not depend on how the path is parameterized in time. Nevertheless, it remains unclear whether multiple hyperparameter paths starting from the same initial location and arriving at the same end location, produce solution paths that end at the same solution. Future work should consider the consistency of solutions under changes in the hyperparameter path used. A consistency guarantee would strongly recommend the convex relaxation approach. 

Next, there is a trade-off between step length and the accuracy of the solution path. Smaller step lengths produce more accurate solutions but require more computational effort. Ideally, step length should be chosen adaptively to balance efficiency and accuracy. Two complementary approaches are clear. First, prediction computes the sensitivity of the solution to changes in the hyperparameters, so large steps could be used when the solution is insensitive and small steps could be used when sensitive. Alternatively, the number and length of correction steps needed evaluates the inaccuracy of prediction. Thus, the number or size of corrections steps used could guide step length. When the previous correction was small and cheap, the step size should be increased. When the previous corrections were large or expensive, then step length should be decreased.  

Our proposed algorithm also does not fully exploit the continuity of the linear systems solved en route. Each steps solves a linear system that is only slightly different than previously solved. We partially exploit this structure by initializing an iterative solver from the previous solution, but do not exploit it when building the preconditioner. Future work should investigate cheap methods for updating the preconditioner. The correct updating method is not obvious, since it is important to retain the low rank complexity. Otherwise, the update may prove more expensive than explicitly recomputing the preconditioner.

Finally, other ODE solvers may prove more efficient and stable. For example, Runge-Kutta methods offer well developed higher-order solvers. In \cite{zhu2021algorithmic}, Runge-Kutta methods beats Euler methods in tracking the $\ell_2$ regularized solution.  Here we use the simplest ODE solver, forward Euler, since correction prevents the accumulation of error and relies on solving linear systems of the same kind used for prediction. Future studies could evaluate other ODE solvers.

\appendix

\section{The \texorpdfstring{$ 2\times 2 $ }{g} Case}\label{sec:2 by 2}

Here we show that, when $n=1$, the $2\times2$ Hessian is invertible for almost all $x$. To begin, we introduce some relations between $u$ and $\xi$.
 

At a MAP solution, $\xi$ is optimal with respect to $x$. There, first order optimality requires that \cite{calvetti2019hierachical}:
\begin{equation}\label{eq}
    -\frac{1}{2}u_j^{2}-\eta \xi_j + r\xi_j^{r+1} = 0
\end{equation}
for all $j$.

The hyperparameters only define a valid objective if $r > 0$ and $\eta > 0$, or if $r < 0$ and $\eta < -3/2$. In those regions, equation \eqref{eq} implicitly defines a smooth function $\varphi$ which optimizes $\xi_j$ with respect to $u_j$ \cite{calvetti2019hierachical}
\begin{equation} \label{eqn: manifold}
    \xi_j = \varphi(u_j).
\end{equation}
When $u_j > 0$, $\varphi(u_j)$ is strictly increasing. 

\begin{lemma}[$2\times2$ Hessian Invertibility] \label{lem: two by two}
When $n=1$, the Hessian \eqref{Hessian} is invertible for almost all $x$. 
\end{lemma}
\begin{proof}

As shown before, if $\hat{H}$ is invertible and $x_j \neq 0$ for all $j$, then the Hessian (\ref{Hessian}) is invertible. When $n=1$, 
$$
\hat{H}=\left[\begin{array}{cc}
a^{2} + \frac{1}{\xi} & a^{2} \\
a^{2}  & a^{2} 
+\frac{1}{2\xi} 
+ r^{2} \frac{\xi^{r}}{u^{2}}
\end{array}\right],
$$
where $a$ is a constant determined by the column scaled forward model.

The matrix $\hat{H}$ is invertible if $\det(\hat{H}) \neq 0$, where

$$
\det(\hat{H})=\left(a^{2}+\frac{1}{\xi}\right) \left( a^{2} +\frac{1}{2\xi} 
+ r^{2} \frac{\xi^{r}}{u^{2}} \right) - a^{4} = a^{2} \left(\frac{3}{2\xi}+r^{2}\frac{\xi^{r}}{u^{2}}\right)+\left(\frac{1}{2\xi^{2}}+r^{2}\frac{\xi^{r-1}}{u^{2}} \right).
$$

When $u \neq 0$, $u^{2} > 0$. Then, since $\xi > 0$, $a^{2} \ge 0$ and $r^{2}>0$, 
$$
\begin{aligned}
    &\frac{3}{2\xi}+r^{2}\frac{\xi^{r}}{u^{2}} > 0 , \\
    &\frac{1}{2\xi^{2}}+r^{2}\frac{\xi^{r-1}}{u^{2}} > 0.
\end{aligned}
$$
Thus, $\det(\hat{H}) > 0$, so $\hat{H}$ is invertible. Therefore, $\hat{H}$ is invertible for all $u$ except $u = 0$.

As introduced before, $x = \vartheta^{1/2}u$, where $\vartheta > 0$. Then
$$
u = 0 \Leftrightarrow x = 0.
$$
So the Hessian $H$ is invertible for almost all $x$. 

\end{proof}

Lemma \ref{lem: two by two} establishes that the $2\times 2$ Hessian is invertible for almost all $x$. Next, we consider the general case.

\section{Proof of \texorpdfstring{\ref{thm: invertibility}}{}} \label{sec:invertibility}

\begin{proof} 
Let $f(u) = \frac{1}{\xi}$, $g(u) = \frac{1}{2\xi} + r^{2}\xi^{r}u^{-2}$ where $\xi = \varphi(u)$, $[u, \xi]$ is the non-dimensional parameters of $z$ and $M=\hat{A}^{\top}\hat{A}$.  Then, $\hat{H}$ \eqref{newhessian} can be expressed: 
$$
\left[ \begin{array}{cc}
M + D_{f(u)}     & M \\
M     & M+ D_{g(u)} 
\end{array} \right].
$$

Permute the rows and columns so that entries that depend on $x_j$ and $\theta_j$ are adjacent,
$$
C_{ij} = \hat{H}_{p(i)p(j)},
$$
where $p(i) = 
\left \{
\begin{array}{ll}
\frac{i-1}{2} + 1     &  \text{if } i \text{ is odd} \\
\frac{i}{2} + n     & \text{if } i \text{ is even}
\end{array}
\right..
$  

The Hessian is invertible if and only if its permutation, $C$, is invertible.

By Lemma \ref{lem: two by two}, the $2\times2$ minor, $C^{(2)}$ is invertible for almost all $z$. Suppose that $C^{(2k)}$ is invertible for almost all $z$ when $k < n$. Then, we aim to show that $C^{(2k+1)}$ and $C^{(2k+2)}$ are invertible for almost all $z$.

First, consider
$$
C^{(2k+1)} = \left[ \begin{array}{cc}
C^{(2k)}     &  v \\
v^{\top}     & M_{k+1,k+1} + f(u_{k+1}) 
\end{array} \right],
$$
where $v$ is the first $2k$ entries of the $(2k+1)^{st}$ column of $C$. Let $C^{(2k+1)}_{j}$ denote the $j^{th}$ column of matrix $C^{(2k+1)}$.

Choose $z$ so that $C^{(2k)} $ is invertible. Then $C^{(2k+1)}$ is singular if and only if
\begin{equation}\label{*}
    \left[\begin{array}{c}
    v   \\
    M_{k+1,k+1} + f(u_{k+1})  
    \end{array}
    \right] \in \text{span} \{ C^{(2k+1)}_{j}\}_{j=1}^{2k}.
\end{equation}

equation \eqref{*} holds for at most one value of $f(u_{k+1})$, that is, for at most one value of $\xi_{k+1}$. At a MAP solution, $\xi_{k+1} = \varphi(u_{k+1})$ is an invertible, monotonically increasing function of $u_{k+1}$ when $u_{k+1} > 0$. Moreover, $\varphi$ is a symmetric function, such that $\varphi(u_{k+1}) = \varphi(-u_{k+1})$. Thus, if $C^{(2k)}$ is invertible, then there exist at most two value of $u_{k+1}$ such that $C^{(2k+1)}$ is non-invertible. Then, since we assumed that $C^{(2k)}$ is invertible for almost all $z$, so is $C^{(2k+1)}$. 

Next, consider
$$
C^{(2k+2)} = \left[ \begin{array}{ccc}
C^{(2k)}     &  v  &  v \\
v^{\top}  &  M_{k+1,k+1} + f(u_{k+1}) & M_{k+1,k+1} \\
v^{\top}  &  M_{k+1, k+1} & M_{k+1,k+1} + g(u_{k+1})
\end{array}\right]
$$

Suppose that $z$ is chosen so that $C^{(2k+1)}$ is invertible. Then $C^{(2k+2)}$ is singular if and only if:
\begin{equation} \label{**}
\left[\begin{array}{c}
v    \\
M_{k+1, k+1} \\
M_{k+1, k+1} + g(u_{k+1})
\end{array} \right]  \in \text{span}\{ C^{(2k+2)}_{j}\}_{j=1}^{2k+1}
\end{equation}

Equation \eqref{**} requires that there exists a $\left(2k+1\right)$-dimensional vector $y \neq 0$ such that:
\begin{equation*}
\begin{aligned}
\sum_{j=1}^{2k} C^{(2k+2)}_{j}(u_{1}, \cdots, u_{k}) y_{j} & + 
\left[\begin{array}{c}
v    \\
M_{k+1, k+1} + f(u_{k+1})\\
M_{k+1, k+1}
\end{array} \right]
y_{2k+1}\\
& =
\left[\begin{array}{c}
v    \\
M_{k+1, k+1} \\
M_{k+1, k+1} + g(u_{k+1})
\end{array} \right]
\end{aligned}
\end{equation*}

Subtract $ w + f(u_{k+1})e_{2k+1}$ from both sides, where:
$$
w = \left[\begin{array}{c}
v     \\
M_{k+1, k+1} \\
M_{k+1, k+1}
\end{array} \right].
$$

Then:
\begin{equation}\label{***}
    \sum_{j=1}^{2k} C^{(2k+2)}_{j}(u_{1}, \cdots, u_{k}) y_{j} + w(y_{2k+1}-1) = g(u_{k+1})e_{2k+2} - f(u_{k+1})e_{2k+1}.
\end{equation}

Let $s = y - e_{2k+1}$, where $s \neq - e_{2k+1}$ since $y \neq 0$. Then, equation \eqref{***} reduces to the linear system:
$$
\left[ \begin{array}{cc}
C^{(2k)}     &  v   \\
v^{\top}  &  M_{k+1,k+1} + f(u_{k+1})  \\
v^{\top}  &  M_{k+1, k+1} 
\end{array}\right]s
=
\left[\begin{array}{c}
0    \\
-f(u_{k+1})  \\
g(u_{k+1})
\end{array} \right].
$$

First, focus on the upper $2k$ rows.

The first $2k$ equations require $C^{(2k)}[s_1;\hdots;s_{2k}]$ $= - s_{2k+1} v$. By assumption, $z$ was chosen so that $C^{(2k)}$ is invertible, so $C^{(2k)} [s_1;\hdots;s_{2k}]$ $= - s_{2k+1} v$ has a unique solution for each choice of $s_{2k+1}$. Let $s_{2k+1} = \lambda$, and let $t$ be the unique solution to $C^{(2k)} t = v$. Then $s = \lambda \left[t ; 1 \right]$.

Next, focus on the last two rows. Equation \eqref{***} requires:
$$
\lambda \left[\begin{array}{cc}
v^{\top}  &  M_{k+1,k+1} + f(u_{k+1})  \\
v^{\top}  &  M_{k+1, k+1}
\end{array} \right]
\left[\begin{array}{c}
t     \\
1     
\end{array} \right]
= 
\left[\begin{array}{c}
-f(u_{k+1})  \\
g(u_{k+1})
\end{array} \right],
$$
which implies that:
\begin{equation}
\begin{aligned}
    &\lambda(v^{\top}t + M_{k+1,k+1} + f(u_{k+1})) = -f(u_{k+1}), \text{ and } \\
    & \lambda(v^{\top}t + M_{k+1,k+1}) = g(u_{k+1}).
\end{aligned}
\end{equation}

To simplify, let $v^{\top}t + M_{k+1,k+1} = \alpha$. Then, the system reduces to the pair of equations,
\begin{equation}\label{2row}
    \lambda(\alpha + f(u_{k+1})) = -f(u_{k+1}),
\end{equation}
and
\begin{equation}\label{1row}
    \lambda \alpha = g(u_{k+1}).
\end{equation}

If $f(u_{k+1}) \neq 0$ or $g(u_{k+1}) \neq 0$, then $\lambda \neq 0$. By definition, $f(u_{k+1}) = \frac{1}{\xi_{k+1}} = \frac{1}{\varphi(u_{k+1})} > 0$, so $\lambda \neq 0$. Thus, both equations must be satisfied simultaneously for some $\lambda \neq 0$. Similarly, $g(u_{k+1}) = \frac{1}{2\xi_{k+1}} + r^{2}\xi_{k+1}^{r}u_{k+1}^{-2} > 0$. Then equation \eqref{1row} requires that $\lambda$ and $\alpha$ have the same sign, and only holds if $\alpha \neq 0$. 

If both $\lambda$ and $\alpha$ are positive, 
$$
 \lambda(\alpha + f(u_{k+1})) > 0 > - f(u_{k+1}),
$$
so equation \eqref{2row} can't hold. 

Suppose both $\lambda$ and $\alpha$ are negative and substitute \eqref{1row} into \eqref{2row}. Then
\begin{equation}\label{lastone}
    g(u_{k+1}) + \frac{g(u_{k+1})f(u_{k+1})}{\alpha} = -f(u_{k+1}),
\end{equation}
which requires
\begin{equation} \label{eqn: z and xi constraint}
g(u_{k+1})f(u_{k+1}) + \alpha (f(u_{k+1}) + g(u_{k+1})) =  \frac{1}{2\xi_{k+1}^{2}} + r^{2}\frac{\xi_{k+1}^{r-1}}{u_{k+1}^{2}} + \alpha \left( \frac{3}{2\xi_{k+1}}+r^{2}\frac{\xi_{k+1}^{r}}{u_{k+1}^{2}} \right) = 0.
\end{equation}

Equation \eqref{eqn: z and xi constraint} enforces
$$
(1 + 3 \alpha \xi_{k+1} ) \frac{1}{2}u_{k+1}^{2} + r^{2}\xi_{k+1}^{r+1} + \alpha r^{2}\xi_{k+1}^{r+2} = 0 ,
$$
for all $u_{k+1} \neq 0$.

Suppose $\xi_{k+1} \neq -\frac{1}{3\alpha}$. Then:
\begin{equation}\label{funz}
    \frac{1}{2} u_{k+1}^{2} = - r^{2}\xi_{k+1}^{r+1} \frac{1+\alpha\xi_{k+1}}{1+3\alpha\xi_{k+1}}
\end{equation}

Since $u_{k+1}^{2} > 0$,  $\frac{1+\alpha\xi_{k+1}}{1+3\alpha\xi_{k+1}} < 0$. Due to $\alpha < 0$, equation \eqref{funz} only holds when $-\frac{1}{3\alpha} < \xi_{k+1} < -\frac{1}{\alpha}$.

Recall that, at a MAP solution, $-\frac{1}{2}u_{k+1}^{2} - \eta\xi_{k+1} + r\xi_{k+1}^{r+1} = 0$. Thus, when $-\frac{1}{3\alpha} < \xi_{k+1} < -\frac{1}{\alpha}$
$$
-r^{2} \frac{1+\alpha\xi_{k+1}}{1+3\alpha\xi_{k+1}} = r - \frac{\eta}{\xi_{k+1}^{r}}.
$$

Therefore, if $C^{(2k)}$ and $C^{(2k+1)}$ are invertible, $C^{(2k+2)}$ is non-invertible at $u_{k+1} \neq 0$ if and only if:
\begin{equation}\label{funxi}
    -r^{2} \frac{1+\alpha\xi_{k+1}}{1+3\alpha\xi_{k+1}} - r + \frac{\eta}{\xi_{k+1}^{r}} = 0.
\end{equation}

Suppose that $u_{k+1}$ can be chosen such that $\xi_{k+1} = \varphi(u_{k+1})$ solves equation \eqref{funxi}. Then, there is only one such solution since $\varphi$ is invertible and the left hand side of \eqref{funxi} is monotonically decreasing in $\xi_{k+1}$. To show that the left hand side is monotonically decreasing, differentiate it with respect to $\xi_{k+1}$. The derivative is:
$$
\frac{2\alpha r^{2}}{(1+3\alpha \xi_{k+1})^{2}} - \frac{r\eta}{\xi_{k+1}^{r+1}} < 0,
$$
which implies that at most one $\xi_{k+1}$ satisfies \eqref{funxi}, and \eqref{funz} holds for at most two $u_{k+1} \neq 0$.

Therefore, if $C^{(2k)}$ and $C^{(2k+1)}$ are invertible, $C^{(2k+2)}$ is  invertible for almost all $z$.  Induction follows:
\begin{enumerate}
    \item $C^{(2)}$ is invertible for almost all $z$ via Lemma \ref{lem: two by two}.
    \item Given $C^{(2k)}$ invertible for almost all $z$, both $C^{(2k+1)}$ and $C^{(2k+2)}$ are invertible for almost all $z$.
\end{enumerate}
Then, by induction, $C$ is  invertible for almost all $z$. It follows that the Hessian $H$ is invertible for almost all $z$. \end{proof}

Theorem \ref{thm: invertibility} ensures that the ODE governing the solution path $z_*(t)$, \eqref{ODE}, is well-defined for almost all $z$. Since $H$ is invertible for almost all $z$ on the solution manifold, linear systems involving $H$ evaluated at solutions admit unique solutions for almost all $z$. Consequently, if $z$ is a continuous random variable, the Hessian matrix is almost surely invertible. In practice, $z$ is random for two reasons. First, the original signal is perturbed by noise, and determines where the solution path starts. Second, any numerical algorithm will accrue random errors, so all practical methods will inherit randomly perturbed $z$. Therefore, in practice, all linear systems involving $H$ may be treated as invertible. Note that this argument does not rule out the possibility that solution paths may cross a bifurcation with nonzero probability, since, even if continuously distributed $z$ ensure $H$ is almost always invertible, the flow defined by the ODE \eqref{ODE}, may pass through a manifold of measure zero with probability one. Therefore, Theorem \ref{thm: invertibility} ensures that $H$ will be invertible at almost all $z$, for all sample $z$ observed numerically, but does not rule out the possibility that solution paths cross a lower-dimensional manifold where solution paths bifurcate.

\section{Proof of \texorpdfstring{\ref{thm: uniqueness}}{}} \label{sec: Uniqueness}

\begin{proof}
Suppose that the Hessian $H$ is invertible. Then, multiplying Equation \eqref{ODE} by the inverse of $H$ on both sides yields $\frac{d}{dt}x(t)$ explicitly:
\begin{equation*}
\begin{aligned}
&  \frac{d}{d t} z(t) = - H(z|\psi)^{-1} \left( \partial_{r} \nabla_{z} \mathcal{G}\left(z \mid \psi(t) \right) \frac{d}{d t} r(t) + \right. \\ & 
\hspace{2cm} \left. \partial_{\eta} \nabla_{z} \mathcal{G}\left(z \mid \psi(t) \right) \frac{d}{d t} \eta(t) +  \partial_{\vartheta} \nabla_{z} \mathcal{G}\left(z \mid \psi(t) \right) \frac{d}{d t} \vartheta(t) \right).
\end{aligned}
\end{equation*}
Let $h\left(z(t),t\right) = - H(z|\psi)^{-1}  \nabla_{z} (\nabla_{\psi} \mathcal{G}\left(z \mid \psi(t) \right) \cdot \frac{d}{dt} \psi(t))$. Then, the ODE can be expressed 
\begin{equation}\label{dzdt}
    \frac{d}{dt}z(t) = h\left(z(t),t\right), \quad z(t_{0})=z_{0}.
\end{equation}

If $h$ is defined on a closed rectangle containing $\left(z_{0}, t_{0}\right)$ where it is continuous in $t$ and Lipschitz continuous in $z$, then, by the Picard–Lindel\"of theorem, there exists some $\epsilon > 0$, such that ODE \eqref{dzdt} has unique solution on $\left[t_{0} - \epsilon, t_{0} + \epsilon \right]$. 

First, we show the ODE \eqref{dzdt} is well defined on an open set containing $\left(z_{0}, t_{0}\right)$. As long as $H$ is invertible, the ODE is well defined. Since $H(z,\psi)$ is a continuous matrix valued function of $z$ and $\psi$, and $\psi(t)$ is a continuous function of $t$, $H$ is a continuous function of $z$ and $t$. Then the determinant of Hessian $H$, $\det(H)$, is a continuous function of $z$ and $t$. At the initial point, $H(z_{0}, \psi(t_{0}))$ is invertible and $\det(H(z_{0}, \psi(t_{0})))\neq 0$. Therefore, on an open set in $\mathbb{R}^{2n+1}$ containing $\left(z_{0}, t_{0}\right)$, $\det(H(z, \psi(t)))\neq0$ and $H(z, \psi(t))$ is invertible  which implies the ODE \eqref{dzdt} is well defined on a closed rectangle $\mathcal{D}$ containing $\left(z_{0}, t_{0}\right)$.

Next, we check the continuity of $h(z,t)$ in $t$. Since $\psi(t)$ is continuously differentiable, $\frac{d}{dt}\psi(t)$ is continuous. The inverse of Hessian can be expressed $H^{-1} = \frac{H^{*}}{\det(H)},
$ where $H^{*}$ is the adjoint of $H$. Therefore, $H^{-1}$ is a rational function of the entries of $H$. Thus, $H^{-1}$ is a continuous function of the entries of $H$ where it exists. Since the composition of continuous functions is continuous, $H(z|\psi(t))^{-1}$ is continuous in $t$.

The right hand side of \eqref{ODE} has three parts. Since the object function has continuous second order partial derivatives, we can exchange the order of  the partials. Then, 
$$
\begin{aligned}
\partial_{r} \nabla_{z} \mathcal{G}\left(z \mid \psi(t) \right) 
& =  \nabla_{z} \partial_{r} \mathcal{G}\left(z \mid \psi(t) \right) = 
\nabla_{z}\left(\sum_{j=1}^{n}\left(\frac{\theta_{j}}{\vartheta_{j}}\right)^{r}\log \frac{\theta_{j}}{\vartheta_{j}}\right) \\
& = [0,\cdots, 0, \cdots \\
& \quad \frac{1}{\vartheta_{1}}\left(\frac{\theta_{1}}{\vartheta_{1}}\right)^{r-1}(1+r\log\frac{\theta_{1}}{\vartheta_{1}}), \cdots, \frac{1}{\vartheta_{n}}\left(\frac{\theta_{n}}{\vartheta_{n}}\right)^{r-1}(1+r\log\frac{\theta_{n}}{\vartheta_{n}})]^{\top} , \\
\partial_{\eta} \nabla_{z} \mathcal{G}\left(z \mid \psi(t) \right) 
& =
\nabla_{z} \partial_{\eta} \mathcal{G}\left(z \mid \psi(t) \right) =
\nabla_{z} \left (-\sum_{j=1}^{n} \log \frac{\theta_{j}}{\vartheta_{j}} \right) \\
& = -[0,\cdots, 0, 1/\theta_{1}, \cdots, 1/\theta_{n}]^{\top}, \\
\partial_{\vartheta} \nabla_{z} \mathcal{G}\left(z \mid \psi(t) \right) & = \nabla_{z} \partial_{\vartheta} \mathcal{G}\left(z \mid \psi(t) \right) =
\nabla_{z} \left([\eta\frac{1}{\vartheta_{1}} - r\frac{\theta_{1}^{r}}{\vartheta_{1}^{r+1}}, \cdots, \eta\frac{1}{\vartheta_{n}} - r\frac{\theta_{n}^{r}}{\vartheta_{n}^{r+1}} ]\right) \\
& = -r^{2}
\left[ \begin{array}{c}
0 \\
D_{[ \frac{\theta_{1}^{r-1}}{\vartheta_{1}^{r+1}}, \cdots, \frac{\theta_{n}^{r-1}}{\vartheta_{n}^{r+1}}]}
\end{array} \right].
\end{aligned}
$$
Each term is continuous in $t$, so $\partial_{\psi} \nabla_{z} \mathcal{G}\left(z \mid \psi(t) \right)$ is continuous in $t$. As a result, $h(z,t)$ is continuous in $t$.

Finally, all entries of $H(z \mid \psi)$ are continuously differentiable in $z$, so $H(z \mid \psi)^{-1}$ is continuously differentiable in $z$. The partial derivative $\nabla_{\psi} \nabla_{z} \mathcal{G}\left(z \mid \psi(t) \right)$ is also continuously differentiable in $z$ and $\frac{d}{dt}\psi(t)$ is independent of $z$. Thus, $h\left(z(t),t\right)$ is continuously differentiable in $z$ so is Lipschitz continuous in $z$ on $\mathcal{D}$.

It follows that, by Picard-–Lindel\"of, the ODE \eqref{ODE} has a unique solution on an interval containing $t_{0}$.
\end{proof}

\bibliographystyle{siamplain}
\bibliography{Refs.bib}
\end{document}